\pdfoutput=1

\RequirePackage[l2tabu,orthodox]{nag}

\documentclass[11pt,letterpaper,reqno,oneside]{article}

\DeclareTextFontCommand{\emph}{\slshape}

\makeatletter
\renewcommand{\paragraph}{%
	\@startsection{paragraph}{4}%
	{\z@}{1.75ex \@plus 1ex \@minus .2ex}{-0.7em}%
	{\normalfont\normalsize\bfseries}%
}
\makeatother

\usepackage[T1]{fontenc}
\usepackage{lmodern}
\usepackage[utf8]{inputenc}

\usepackage[american,german,british]{babel}

\usepackage[activate={true,nocompatibility}]{microtype}

\usepackage{csquotes}

\usepackage[backend=biber,autolang=other,style=apa,maxcitenames=5,sortcites=false]{biblatex}
\DeclareLanguageMapping{british}{british-apa}
\DeclareLanguageMapping{american}{american-apa}

\usepackage{amsmath}
\usepackage{amssymb}
\usepackage{amsfonts}
\usepackage{amsthm}
\usepackage{mathtools}
\usepackage{stmaryrd}

\let\originalleft\left
\let\originalright\right
\renewcommand{\left}{\mathopen{}\mathclose\bgroup\originalleft}
\renewcommand{\right}{\aftergroup\egroup\originalright}

\usepackage{graphicx}
\usepackage{tikz,pgfplots}
\pgfplotsset{compat=1.10}
\usetikzlibrary{shapes.multipart}
\usetikzlibrary{decorations.markings}

\usepackage[title]{appendix}

\usepackage{datetime}
\newdateformat{datestyle}{ {\THEDAY} \monthname[\THEMONTH] \THEYEAR }

\usepackage{enumerate}
\usepackage{enumitem}
\setlist[enumerate,1]{label=(\arabic*)}
\setlist[itemize,1]{label=--}
\setlist[itemize,2]{label=--}
\setlist[itemize,3]{label=--}
\setlist[itemize,4]{label=--}

\usepackage{xcolor}
\usepackage{verbatim}
\usepackage{gensymb}
\usepackage{rotating}
\usepackage{subcaption}
\usepackage{floatpag}
\floatpagestyle{plain}
\usepackage{marvosym}

\usepackage{hyphenat}
\theoremstyle{definition}

\newtheorem{theorem}{Theorem}
\newtheorem{proposition}{Proposition}
\newtheorem{lemma}{Lemma}
\newtheorem{corollary}{Corollary}
\newtheorem{remark}{Remark}
\newtheorem{observation}{Observation}
\newtheorem{example}{Example}
\newtheorem{definition}{Definition}

\newtheoremstyle{named}
	{\topsep}					
	{\topsep}					
	{}							
	{0pt}						
	{\bfseries}					
	{}							
	{5pt plus 1pt minus 1pt}	
	{\thmnote{#3}}				
\theoremstyle{named}
\newtheorem{namedthm}{}

\renewcommand{\qedsymbol}{$\blacksquare$}

\usepackage{xpatch}
\xpatchcmd{\proof}{\itshape}{\proofheadfont}{}{}
\newcommand{\proofheadfont}{\slshape}

\usepackage{hyperref}
\hypersetup{pdfborder=0 0 0}

\usepackage[nameinlink]{cleveref}
\crefname{page}{p.}{pp.}
\crefname{equation}{equation}{equations}
\crefname{section}{section}{sections}
\crefname{subsection}{section}{sections}
\crefname{subsubsection}{section}{sections}
\crefname{appsec}{appendix}{appendices}
\crefname{supplsec}{supplemental appendix}{supplemental appendices}
\crefname{footnote}{footnote}{footnotes}
\crefname{figure}{figure}{figures}
\crefname{table}{table}{tables}
\crefname{theorem}{theorem}{theorems}
\crefname{proposition}{proposition}{propositions}
\crefname{lemma}{lemma}{lemmata}
\crefname{corollary}{corollary}{corollaries}
\crefname{remark}{remark}{remarks}
\crefname{observation}{observation}{observations}
\crefname{example}{example}{examples}
\crefname{fact}{fact}{facts}
\crefname{definition}{definition}{definitions}
\crefname{assumption}{assumption}{assumptions}
\crefname{exercise}{exercise}{exercises}
\crefname{notation}{notation}{notation}
\crefname{claim}{claim}{claims}
\crefname{conjecture}{conjecture}{conjectures}

\newcommand{\eps}{\varepsilon}

\newcommand{\dd}{\mathrm{d}}

\newcommand{\DD}{\mathrm{D}}

\DeclareMathOperator*{\argmax}{arg\,max}

\DeclareMathOperator*{\inv}{inv}

\DeclareMathOperator*{\interior}{int}

\DeclareMathOperator*{\cl}{cl}

\DeclareMathOperator*{\supp}{supp}

\newcommand{\E}{\mathbf{E}}

\newcommand{\PP}{\mathbf{P}}

\newcommand{\oo}{\mathrm{o}}

\newcommand{\R}{\mathbf{R}}

\newcommand{\Q}{\mathbf{Q}}

\newcommand{\C}{\mathbf{C}}

\newcommand{\N}{\mathbf{N}}

\newcommand{\1}{\boldsymbol{1}}

\newcommand{\join}{\vee}

\newcommand{\meet}{\wedge}

\newcommand{\union}{\cup}

\newcommand{\intersect}{\cap}

\newcommand{\Union}{\bigcup}

\DeclarePairedDelimiter\abs{\lvert}{\rvert}

\newcommand*{\xslant}[2][76]{%
	\begingroup
	\sbox0{#2}%
	\pgfmathsetlengthmacro\wdslant{\the\wd0 + cos(#1)*\the\wd0}%
	\leavevmode
	\hbox to \wdslant{\hss
		\tikz[
			baseline=(X.base),
			inner sep=0pt,
			transform canvas={xslant=cos(#1)},
		] \node (X) {\usebox0};%
		\hss
		\vrule width 0pt height\ht0 depth\dp0 %
	}%
	\endgroup
}
\makeatletter
\newcommand*{\xslantmath}{}
\def\xslantmath#1#{%
	\@xslantmath{#1}%
}
\newcommand*{\@xslantmath}[2]{%
	\ensuremath{%
		\mathpalette{\@@xslantmath{#1}}{#2}%
	}%
}
\newcommand*{\@@xslantmath}[3]{%
	\xslant#1{$#2#3\m@th$}%
}
\makeatother

\makeatletter
\def\namedlabel#1#2{\begingroup
	#2%
	\def\@currentlabel{#2}%
	\phantomsection\label{#1}\endgroup
}
\makeatother

\makeatletter
\let\save@mathaccent\mathaccent
\newcommand*\if@single[3]{%
	\setbox0\hbox{${\mathaccent"0362{#1}}^H$}%
	\setbox2\hbox{${\mathaccent"0362{\kern0pt#1}}^H$}%
	\ifdim\ht0=\ht2 #3\else #2\fi
	}
\newcommand*\rel@kern[1]{\kern#1\dimexpr\macc@kerna}
\newcommand*\widebar[1]{\@ifnextchar^{{\wide@bar{#1}{0}}}{\wide@bar{#1}{1}}}
\newcommand*\wide@bar[2]{\if@single{#1}{\wide@bar@{#1}{#2}{1}}{\wide@bar@{#1}{#2}{2}}}
\newcommand*\wide@bar@[3]{%
	\begingroup
	\def\mathaccent##1##2{%
	  \let\mathaccent\save@mathaccent
	  \if#32 \let\macc@nucleus\first@char \fi
	  \setbox\z@\hbox{$\macc@style{\macc@nucleus}_{}$}%
	  \setbox\tw@\hbox{$\macc@style{\macc@nucleus}{}_{}$}%
	  \dimen@\wd\tw@
	  \advance\dimen@-\wd\z@
	  \divide\dimen@ 3
	  \@tempdima\wd\tw@
	  \advance\@tempdima-\scriptspace
	  \divide\@tempdima 10
	  \advance\dimen@-\@tempdima
	  \ifdim\dimen@>\z@ \dimen@0pt\fi
	  \rel@kern{0.6}\kern-\dimen@
	  \if#31
	    \overline{\rel@kern{-0.6}\kern\dimen@\macc@nucleus\rel@kern{0.4}\kern\dimen@}%
	    \advance\dimen@0.4\dimexpr\macc@kerna
	    \let\final@kern#2%
	    \ifdim\dimen@<\z@ \let\final@kern1\fi
	    \if\final@kern1 \kern-\dimen@\fi
	  \else
	    \overline{\rel@kern{-0.6}\kern\dimen@#1}%
	  \fi
	}%
	\macc@depth\@ne
	\let\math@bgroup\@empty \let\math@egroup\macc@set@skewchar
	\mathsurround\z@ \frozen@everymath{\mathgroup\macc@group\relax}%
	\macc@set@skewchar\relax
	\let\mathaccentV\macc@nested@a
	\if#31
	  \macc@nested@a\relax111{#1}%
	\else
	  \def\gobble@till@marker##1\endmarker{}%
	  \futurelet\first@char\gobble@till@marker#1\endmarker
	  \ifcat\noexpand\first@char A\else
	    \def\first@char{}%
	  \fi
	  \macc@nested@a\relax111{\first@char}%
	\fi
	\endgroup
}
\makeatother

\usepackage{sidecap}
\sidecaptionvpos{figure}{c}

\usepackage{textcomp}

\makeatletter
	\newcommand{\hyperdest}[1]{\Hy@raisedlink{\hypertarget{#1}{}}}
\makeatother

		\newif\ifustar
		\newif\iftransfers
		\newif\ifFone
		\newif\ifustar
		\newif\ifequal
		\newif\ifconcave
		\newif\ifmix
		\newif\ifmixlabels
		\newif\iflabels
		\newif\ifFzero
		\newif\ifcavphan
		\newif\iflab
		\newif\ifphantom
		\newif\ifMCone
		\newif\ifQ
		\newif\ifpartial
		\newif\ifside
		\newif\ifU
		\newif\ifUdl
		\newif\ifUhat
		\newif\ifUhatdl
		\newif\ifT
		\newif\ifplus
		\newif\ifhull
		\newif\iffunction
		\newif\ifmore
		\newif\ifnew
		\newif\iftriangle
		\newif\iffunctions

\title{\scshape Screening for breakthroughs%
\thanks{We are deeply grateful to Eddie Dekel for many detailed comments.
Sinander thanks him, Alessandro Pavan and Bruno Strulovici for their guidance and support.
We have profited from comments from
anonymous referees,
Dilip Abreu,
Miguel Ballester,
Lori Beaman,
Iván Canay,
Gabe Carroll,
Sylvain Chassang,
Janet Currie,
Théo Durandard,
Piotr Dworczak,
Andrew Ellis,
Jeff Ely,
Alex Frug,
George Georgiadis,
Brett Green,
Faruk Gül,
Yingni Guo,
Marina Halac,
Oliver Hart,
Ian Jewitt,
Shengwu Li,
Alessandro Lizzeri,
Guido Lorenzoni,
Erik Madsen,
Thomas Mariotti,
David Martimort,
Eric Maskin,
John Moore,
Matt Notowidigdo,
Wojciech Olszewski,
Harry Pei,
Wolfgang Pesendorfer,
Dan Quigley,
Debraj Ray,
Ronny Razin,
Wolfgang Ridinger,
Marciano Siniscalchi,
Juuso Toikka,
Arne Uhlendorff,
Can Ürgün,
Chris Wallace,
Asher Wolinsky,
Leeat Yariv
and audiences at Bonn, Harvard, LSE, Lund, Northwestern, NYU, Oxford, Princeton,
\emph{Seminars in Economic Theory,} Toulouse,
and the 2021 \emph{REStud} Tour.
Curello acknowledges support from the German Research Foundation (DFG) through CRC TR 224 (Project B02).%
}%
}

\author{%
Gregorio Curello \\
University of Mannheim
\and
Ludvig Sinander \\
University of Oxford}

\date{13 March 2025}

\makeatletter
	\AtBeginDocument{ \hypersetup{
		pdftitle = {Screening for breakthroughs},
		pdfauthor = {Gregorio Curello and Ludvig Sinander}
		} }
\makeatother

\begin{document}

\maketitle

\begin{abstract}
	How best to incentivise \emph{prompt} disclosure?
	We study this question in a general model in which a technological breakthrough occurs at an uncertain time and is privately observed by an agent,
	and a principal must incentivise disclosure via her control of a payoff-relevant physical allocation.
	We uncover a deadline structure of optimal mechanisms:
	they have a simple deadline form in an important special case,
	and a graduated deadline structure in general.
	We apply our results to the design of unemployment insurance schemes.
\end{abstract}

\section{Introduction}
\label{sec:intro}

Society advances by finding better ways of doing things.
When such a technological breakthrough occurs, it frequently becomes known only to certain individuals with particular expertise.
Only if such individuals share their knowledge promptly can the promise of progress be unlocked.

The resulting need to incentivise prompt disclosure
engenders a screening problem in which the agent's private information is about \emph{when,} rather than about \emph{what.}
We call this \emph{screening for breakthroughs.}

The need to screen for breakthroughs is widespread.
One example is the much-discussed problem of talent-hoarding in organisations \parencite[see][]{Haegele2022}.
The manager of a team is well-placed to know when one of her subordinates acquires a skill.
When this happens, headquarters may wish to re-assign the worker to a new role better-suited to her abilities.
Managers, however, have a documented tendency to want to hold on to their workers.
Careful design is thus needed to incentivise prompt disclosure.

Another example is unemployment insurance:
since unemployed workers are typically privately informed about when they receive a job offer,
benefits must be designed with a view to incentivising them to start work promptly.
A third example concerns technical innovations that reduce firms' greenhouse-gas emissions, at the price of raising production costs.%
	\footnote{Such innovations are expected to account for the bulk of abatement
	in the cement industry,
	currently the source of about $7\%$ of all $\text{CO}_2$ emissions
	(\citeauthor*{CziglerEtal2020}, \citeyear{CziglerEtal2020}).}
Only with suitable regulation
will firms which discover such innovations
choose to adopt them.

In this paper, we study the general problem of screening for breakthroughs.
We introduce a model in which
an agent privately observes when a new productive technology arrives.
This breakthrough expands utility possibilities for the agent and principal, but generates a conflict of interest between them.
The agent decides whether and when to disclose the breakthrough,
and the principal controls a payoff-relevant physical allocation over time.
Our model deliberately focusses on the screening-for-breakthroughs problem, excluding well-understood frictions such as the need to incentivise the agent to exert unobservable effort.
In an extension, we show that adding such a moral-hazard friction to the model does not affect our results.

We ask how the principal can best incentivise prompt disclosure of the breakthrough.
Our answer uncovers a deadline structure of optimal mechanisms:
only simple \emph{deadline mechanisms} are optimal in an important special case,
while a graduated deadline structure characterises optimal incentives in general.
We apply these insights to the design of unemployment insurance schemes.

\subsection{Overview of model and results}
\label{sec:intro:overview}

A breakthrough occurs at a random time, making available a new technology that expands utility possibilities for an agent and a principal.
There is a conflict of interest:
were the principal to operate the old and new technologies in her own interest,
the agent would be better off under the old one.
The agent privately observes when the breakthrough occurs, and (verifiably) discloses it at a time of her choosing.
The principal controls a physical allocation that determines the agent's utility over time.
(The description of a physical allocation may include a specification of monetary payments to the agent; in this case, the conflict-of-interest assumption requires that the agent be protected by (at least a degree of) limited liability prior to disclosure.)

To focus on the robust qualitative features of optimal screening,
we study \emph{undominated} mechanisms,
meaning those such that no alternative mechanism is weakly better for the principal under any arrival distribution of the breakthrough
and strictly better under some distribution.
We further describe, for any given breakthrough distribution, the principal's optimal choice among undominated mechanisms.

Toward our deadline characterisation,
we first study how undominated mechanisms incentivise the agent.
We show that the agent should be indifferent at all times between prompt and delayed disclosure (\Cref{proposition:indiff}).
This is despite the fact that the standard argument fails:
were the agent strictly to prefer prompt to delayed disclosure,
then lowering the agent's post-disclosure utility would \emph{not} necessarily benefit the principal.

We then elucidate the deadline structure of undominated mechanisms
when the pre-breakthrough technology's utility possibilities have an affine shape.
\Cref{theorem:deadline} asserts that in this case, all undominated mechanisms belong to a small class of simple \emph{deadline mechanisms.}
Absent disclosure, these mechanisms give the agent a Pareto-efficient utility $u^0$ before a deadline, and an inefficiently low utility $u^\star$ afterwards.%
	\footnote{$u^0$ and $u^\star$ are functions of the technologies, so the deadline is the only free parameter.}
The proof of \Cref{theorem:deadline} argues
(loosely) that any mechanism may be improved by \emph{front-loading} the agent's pre-disclosure utility, making it higher early and lower late while preserving its total discounted value.
We further characterise the principal's optimal choice of deadline
as a function of the breakthrough distribution (\Cref{proposition:opt_deadline}).

Outside of the affine case,
optimal mechanisms exhibit a graduated deadline structure (\Cref{theorem:opt}):
absent disclosure, the agent's utility still starts at the efficient level $u^0$
and declines monotonically toward the inefficiently low level $u^\star$,
but the transition may be gradual.
For any given breakthrough distribution, we describe the optimal transition (\Cref{proposition:opt_transition}).

We then apply our results to the design of unemployment insurance schemes.
An unemployed worker (agent) receives a job offer at a random time,
and chooses whether to accept, and if so how soon to start.
Offers are private, but the state (principal) observes when the worker starts a job.
The state controls unemployment benefits and income taxes,
and cares both about the worker's welfare and net tax revenue.

Many countries, such as Germany and France,
pay a generous unemployment benefit until a deadline,
and provide only a low benefit to those remaining unemployed beyond this deadline.
Our results provide a potential rationale for such deadline schemes:
they are approximately optimal provided
that either (a)~the worker's consumption utility has limited curvature,
or (b)~tax revenue is comparatively unimportant for social welfare.
Conversely, our analysis suggests that where neither (a) nor (b) is satisfied,
substantial welfare gains could be achieved by tapering benefits gradually, as in Italy.

We conclude by examining the robustness of our general results to the introduction of additional realistic frictions.
We focus on moral hazard, a friction that is important in applications such as unemployment insurance (where search effort is required to generate job offers).
In our extended model, the agent decides in each period whether (unobservably) to exert effort at a cost. Effort improves the chance of a breakthrough.
The principal must now incentivise both effort and disclosure. We show that optimal mechanisms retain their deadline structure: \Cref{theorem:deadline,theorem:opt} remain true (verbatim).

\subsection{Related literature}
\label{sec:intro:lit}

This paper belongs to the literature on incentive design for a proposing agent, initiated by \textcite{ArmstrongVickers2010}.%
	\footnote{See also \textcite{NockeWhinston2013} and \textcite{GuoShmaya2023}.
	Our account of the literature follows the latter authors' insightful discussion.
	The literature has precedents in applied work on corporate finance \parencite{BerkovitchIsrael2004} and antitrust \parencite{Lyons2003}.}
In their (static) model, the agent privately observes which physical allocations are available, then proposes one (or several).
The key assumptions are that
\begin{enumerate}[label=(\alph*)]

	\item \label{bullet:evidence}
	the agent can propose only available allocations, and that

	\item \label{bullet:permission}
	the principal can implement only proposed allocations.

\end{enumerate}
Our dynamic problem shares these key features:
the new technology
\ref{bullet:evidence} can only be disclosed (proposed) once available, and
\ref{bullet:permission} can be utilised by the principal only once disclosed.

\textcite{BirdFrug2019} study a different dynamic environment with features \ref{bullet:evidence} and \ref{bullet:permission}.
Payoffs are simple:
there is an allocation $\alpha$ preferred by the principal
and a default allocation favoured by the agent,%
	\footnote{There is an extension to multiple allocations $\alpha$; little changes.}
and the principal can furthermore reward the agent at a linear cost.
In each period, the agent privately observes whether $\alpha$ is available;
it can \ref{bullet:evidence} be disclosed only if available, and
\ref{bullet:permission} be implemented only if disclosed.
Were rewards unrestricted,
$\alpha$ could be implemented whenever available
by rewarding the agent just enough to induce disclosure.
(And this is optimal; thus there is no conflict of interest in our sense.)
The authors instead subject promised rewards to a dynamic budget constraint,%
	\footnote{They assume in particular that the agent can be rewarded only using exogenous reward `opportunities', which arrive randomly over time;
	but nothing changes if rewards take other forms, e.g. (flow) monetary payments subject to a per-period cap.}
and study how the budget should be spent over time.
By comparison, we allow for general payoffs (technologies) and impose no dynamic constraints, focussing instead on a conflict of interest.

Feature \ref{bullet:evidence} means that the agent's disclosures are verifiable,
a possibility first studied by \textcite{Viscusi1978,GrossmanHart1980,Milgrom1981,Grossman1981}.
A strand of the subsequent literature examines the role of commitment in static models,%
	\footnote{Particularly \textcite{GlazerRubinstein2004,GlazerRubinstein2006,Sher2011,HartKremerPerry2017,BenporathDekelLipman2019}.}
while another studies the timing of disclosure absent commitment;%
	\footnote{See \textcite{DyeSridhar1995,AcharyaDemarzoKremer2011,GuttmanKremerSkrzypacz,CampbellEdererSpinnewijn2014,Curello2023b,Curello2023}.
	The last three papers feature `breakthroughs',
	but these engender no conflict of interest in our sense;
	the incentive problem is instead that of deterring shirking.}
our environment features both commitment and dynamics.%
	\footnote{So does recent work on revenue management,
	where a firm contracts with customers who arrive unobservably over time and choose when verifiably to reveal themselves; see
	\textcite{PaiVohra2013,BoardSkrzypacz2016,Mierendorff2016,Garrett2016,Garrett2017,GershkovMoldovanuStrack2018,DilmeLi2019}.}
These models lack property \ref{bullet:permission}:
the agent cannot constrain the principal.

More distantly related is the large literature on dynamic adverse-selection models
with cheap-talk communication (contrast with \ref{bullet:evidence})
and no scope for the agent to constrain the principal's choice of allocation (contrast with \ref{bullet:permission}).
The strand on dynamic `delegation' allows for non-transferable utility, as we do;%
	\footnote{See \textcite{JacksonSonnenschein2007,MatsushimaMiyazakiYagi2010,Frankel2016jet,Guo2016,LiMatouschekPowell2017,LipnowskiRamos2020,GuoHorner2020,DeclippelEliazFershtmanRozen2021}.}
otherwise the literature tends to focus on monetary transfers.%
	\footnote{E.g. \textcite{Roberts1982,BaronBesanko1984,CourtyLi2000,Battaglini2005,EsoSzentes2007restud,EsoSzentes2007rand,Board2007,PavanSegalToikka2014}.}
A recent strand examines models which, like ours, feature private information about \emph{when,} rather than about \emph{what.}
For example, \textcite{GreenTaylor2016} show how moral hazard may be mitigated by conditioning pay and termination on cheap-talk `progress reports'.%
	\footnote{See also \textcite{FengEtal2023}.}
In their model, the agent privately observes the arrival of a signal which indicates that project completion is within reach (given enough effort).
Completion is observable.
There is no conflict of interest in our sense;
instead, the challenge is to incentivise unobservable (completion-hastening) effort.
(Absent this moral hazard, the principal would have no reason to elicit the signal.)
Relatedly, \textcite{Madsen2022} studies how cheap-talk progress reports may be elicited by conditioning pay and termination on a contractible signal.
In his model, the agent privately observes when a project `expires', and the principal decides when to terminate the project. The principal (agent) prefers termination close to expiry (as late as possible).
Crucially, there is a noisy contractible signal of expiry.%
	\footnote{If there were no contractible signal, then non-trivial screening would be impossible, since the agent's preferences are the same whatever her type (expiry date).}
Both of these papers use the term `deadline', as we do, but mean quite different things by it.%
	\footnote{Deterministic hard deadlines, as in our result, appear only in the benchmark case of Green and Taylor in which there is no signal (a case unrelated to our model and Madsen's).
	In Green and Taylor, `(soft) deadline' means a time after which termination may randomly occur if the agent has not yet reported the signal's arrival.
	Madsen uses `(soft) deadline' to mean that termination depends on the realisation of the contractible signal.}

\subsection{Roadmap}
\label{sec:intro:roadmap}

We introduce the model in the next section, then formulate the principal's problem in §\ref{sec:problem}.
In §\ref{sec:indiff}, we show that undominated mechanisms incentivise the agent by keeping her always indifferent.
We then describe the deadline structure of optimal mechanisms (§\ref{sec:deadline} and §\ref{sec:opt}).
In §\ref{sec:pf}, we apply our results to the design of unemployment insurance schemes.
We conclude in §\ref{sec:mh} by showing that our results extend to a richer model featuring moral hazard.

\section{Model}
\label{sec:model}

There is an agent and a principal,
whose utilities are denoted by $u \in [0,\infty)$ and $v \in [-\infty,\infty)$, respectively.
A frontier $F^0 : [0,\infty) \to [-\infty,\infty)$ describes utility possibilities:
$F^0(u)$ is the highest utility that the principal can attain subject to giving the agent utility $u$.
We assume that $F^0$ is concave and upper semi-continuous,
that it has a unique peak $u^0>0$ (namely, $F^0\left( u^0 \right) > F^0(u)$ for every $u \neq u^0$),
and that it is finite on $\left(0,u^0\right]$. Such a frontier is depicted in \Cref{fig:frontiers}.
\ustartrue													
\equalfalse													
\concavetrue												
\transfersfalse												
\mixfalse													
\mixlabelsfalse												
\labelstrue													
\Fzerotrue													
\Fonetrue													
\cavphanfalse												
\begin{SCfigure}
	\centering
	\input{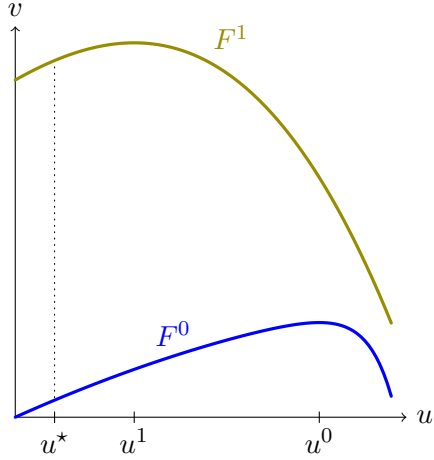}
	\caption{Utility possibility frontiers. The new technology expands utility possibilities ($F^1 \geq F^0$), but creates a conflict of interest ($u^1 < u^0$).
	$u^\star$~denotes the rightmost point to the left of $u^0$ at which $F^0,F^1$ have equal slopes.}
	\label{fig:frontiers}
\end{SCfigure}%

Time $t \in \R_+$ is continuous.
The principal controls the agent's flow utility $u$ (and thus her own utility $F^0(u)$) over time, and is able to commit.

We interpret this abstract description of utility possibilities in the standard fashion:
there is an (unmodelled) set of feasible physical allocations over which the agent and principal have preferences,
and the principal decides which allocation prevails in each period. She thus effectively controls the agent's flow utility. We illustrate and interpret further in §\ref{sec:model:interp} below.

At a random time $\tau$, a \emph{breakthrough} occurs:
a new technology becomes available which expands the utility possibility frontier to $F^1 \geq F^0$.
The new frontier is likewise concave and upper semi-continuous,
with a unique peak denoted by $u^1 \geq 0$.
(Note that we allow for the possibility that $u^1=0$, in which case $F^1$ is decreasing.)
The breakthrough engenders a conflict of interest:
the new frontier peaks at a strictly lower agent utility ($u^1 < u^0$),
so that the breakthrough would hurt the agent were the principal to operate both technologies in her own interest.
This is illustrated in \Cref{fig:frontiers}.

The breakthrough is observed only by the agent.
At any time $t \geq \tau$ after the breakthrough,
she can verifiably disclose to the principal that it has occurred.
(That is, she can \emph{prove} that the new technology is available.)
The new technology can be used only once its availability has been disclosed.

The agent and principal discount their flow payoffs at rate $r>0$ and have expected-utility preferences,
so that their respective payoffs from random flow utilities $t \mapsto x_t$ and $t \mapsto y_t$ are
\begin{equation*}
	\E\left( r \int_0^\infty e^{-rt} x_t \dd t \right)
	\quad \text{and} \quad
	\E\left( r \int_0^\infty e^{-rt} y_t \dd t \right) .
\end{equation*}
The random time $\tau$ at which the breakthrough occurs is distributed according to an arbitrary cumulative distribution function $G$.

We write $u^\star$ for the rightmost $u \in \left[ 0, u^0 \right]$ at which the old and new frontiers $F^0,F^1$ have equal slope (in the sense of sharing a supergradient---see \textcite[part~V]{Rockafellar1970}), and let $u^\star \coloneqq 0$ in case no such $u \in \left[0,u^0\right]$ exist.
This utility level will feature prominently in our analysis.
To avoid trivialities, we impose the weak genericity assumption that $u^\star$ is a strict local maximum of $F^1-F^0$. Note that $u^\star \leq u^1 < u^0$.

\subsection{Interpreting the frontiers}
\label{sec:model:interp}

In the simplest applications, there are finitely many (old) allocations, and the agent privately observes when a single new allocation becomes available.
For example,
a manager may observe when a worker on her team acquires a skill,
or a firm may discover an emissions-reducing innovation.
Each allocation provides some utilities $(u,v)$ to the agent and principal, which may be plotted as in \Cref{fig:projects}.
\ustarfalse													
\newtrue													
\triangletrue												
\functionstrue												
\transfersfalse												
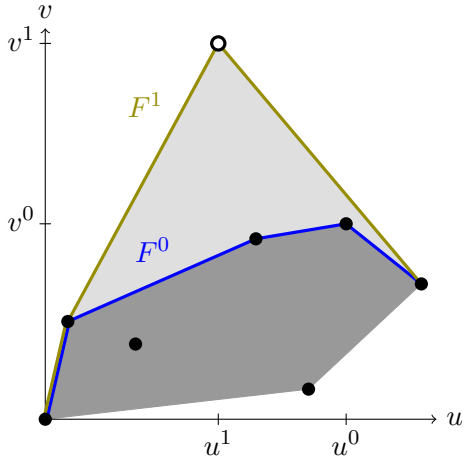
\begin{SCfigure}
	\centering







\begin{tikzpicture}[scale=1, line cap=round]

\pgfmathsetmacro{\ticklength}{1/14};					
\pgfmathsetmacro{\radius}{2.5};							

\iftransfers
	\pgfmathsetmacro{\umax}{5.5000};					
\else
	\pgfmathsetmacro{\umax}{5.2000};					
\fi
\pgfmathsetmacro{\vmax}{5.2000};						
\pgfmathsetmacro{\hoffset}{0.021};						

\pgfmathsetmacro{\uzero}{4.0};							
\pgfmathsetmacro{\vzero}{2.6};							
\pgfmathsetmacro{\uone}{2.3};							
\pgfmathsetmacro{\vone}{5.0};							
\pgfmathsetmacro{\uzeroA}{0.3};							
\pgfmathsetmacro{\vzeroA}{1.3};							
\pgfmathsetmacro{\uzeroB}{2.8};							
\pgfmathsetmacro{\vzeroB}{2.4};							
\pgfmathsetmacro{\uzeroC}{5.0};							
\pgfmathsetmacro{\vzeroC}{1.8};							
\pgfmathsetmacro{\uzeroD}{3.5};							
\pgfmathsetmacro{\vzeroD}{0.4};							
\pgfmathsetmacro{\uzeroE}{1.2};							
\pgfmathsetmacro{\vzeroE}{1.0};							

\iftriangle

	\draw[fill, gray, opacity=0.8]
		(0,0) 
		-- (\uzeroA,\vzeroA)
		-- (\uzeroB,\vzeroB)
		-- (\uzero ,\vzero )
		-- (\uzeroC,\vzeroC)
		-- (\uzeroD,\vzeroD)
		-- cycle;

	\ifnew
		\draw[fill, lightgray, opacity=0.5]
			(\uzeroA,\vzeroA)
			-- (\uone  ,\vone  )
			-- (\uzeroC,\vzeroC)
			-- (\uzero ,\vzero )
			-- (\uzeroB,\vzeroB)
			-- cycle;
	\fi

\fi

\ifustar
	\draw[-,dotted]
		( \uone,
		{ (\uone/(\uzeroA+\uzero)) * \vzeroA
		+ (1-(\uone/(\uzeroA+\uzero))) * \vzero } ) 
		-- ( \uone, \vone );
\fi

\draw[->] (0,0) -- ( \umax, 0.0 );						
\draw ( \umax, 0.0 )									
	node[anchor=west] {$u$};
\draw[->] (0,0) -- ( 0.0, \vmax );						
\draw ( 0.0, \vmax )									
	node[anchor=south] {$v$};

\iffunctions

	\ifnew
		\draw[very thick, blue]							
			(\uzeroA,\vzeroA)
			-- (\uzeroB,\vzeroB)
			-- (\uzero ,\vzero )
			-- (\uzeroC,\vzeroC);	
	\else
		\draw[very thick, blue]							
			( 0, 0 ) 
			-- (\uzeroA,\vzeroA)
			-- (\uzeroB,\vzeroB)
			-- (\uzero ,\vzero )
			-- (\uzeroC,\vzeroC);	
	\fi
	\draw												
		( { 0.4 * \uzeroA + 0.6 * \uzeroB },
		  { 0.4 * \vzeroA + 0.6 * \vzeroB } )
		node[anchor=south east, blue]
		{$F^0$};

	\iftransfers

		\draw[-, very thick, dotted, blue]
		(\uzeroC,\vzeroC)
		-- (\umax,
		{ \vzeroC - \umax + \uzeroC } );

	\fi

	\ifnew

		\draw[very thick, olive]						
			(\uzeroA,\vzeroA)
			-- (\uone  ,\vone  )
			-- (\uzeroC,\vzeroC);
		\draw											
			( { 0.3 * \uzeroA + 0.7 * \uone },
			  { 0.3 * \vzeroA + 0.7 * \vone } )
			node[anchor=south east, olive]
			{$F^1$};

		\draw[very thick, blue]							
			( {0+\hoffset} , { 0 } ) 
			-- ( { \uzeroA+\hoffset }, { \vzeroA } );
		\draw[very thick, olive]						
			( {0-\hoffset} , { 0 } ) 
			-- ( { \uzeroA-\hoffset }, { \vzeroA } );

		\iftransfers

			\draw[-, very thick, dotted, olive]
			(\uone  ,\vone  )
			-- (\umax,
			{ \vone - \umax + \uone } );

		\fi

	\fi

\fi

\fill[thick] ( 0, 0 )									
	circle[radius=\radius pt];
\fill[thick] (\uzero,\vzero)							
	circle[radius=\radius pt];
\ifnew
	\fill[thick,white] (\uone,\vone)					
		circle[radius=\radius pt];
	\draw[very thick] (\uone,\vone)						
		circle[radius=\radius pt];
\fi
\fill[thick] (\uzeroA,\vzeroA)							
	circle[radius=\radius pt];
\fill[thick] (\uzeroB,\vzeroB)							
	circle[radius=\radius pt];
\fill[thick] (\uzeroC,\vzeroC)							
	circle[radius=\radius pt];
\fill[thick] (\uzeroD,\vzeroD)							
	circle[radius=\radius pt];
\fill[thick] (\uzeroE,\vzeroE)							
	circle[radius=\radius pt];

\draw[-] ( \uzero, - \ticklength )						
	-- ( \uzero, \ticklength );
\draw ( \uzero, 0 )										
	node[anchor=north] {$\strut u^0$};
\ifnew
	\draw[-] ( \uone, - \ticklength )					
		-- ( \uone, \ticklength );
	\draw ( \uone, 0 )									
		node[anchor=north] {$\strut u^1$};
\fi

\draw[-] ( - \ticklength, \vzero )						
	-- ( \ticklength, \vzero );
\draw ( 0, \vzero )										
	node[anchor=east] {$\strut v^0$};
\ifnew
	\draw[-] ( - \ticklength, \vone )					
		-- ( \ticklength, \vone );
	\draw ( 0, \vone )									
		node[anchor=east] {$\strut v^1$};
\fi

\end{tikzpicture}

	\caption[Finitely many allocations:
	the old (solid), the new (hollow),
	and utility possibilities (grey). Here $u^\star = u^1$.]%
	{Finitely many allocations:
	the old (%
	\begin{tikzpicture}[scale=1, line cap=round]
		\fill[thick] (0,0) circle[radius=2.5 pt];
	\end{tikzpicture}%
	), the new (%
	\begin{tikzpicture}[scale=1, line cap=round]
		\draw[very thick] (0,0) circle[radius=2.2 pt];
	\end{tikzpicture}%
	),
	and utility possibilities (grey). Here $u^\star = u^1$.}
	\label{fig:projects}
\end{SCfigure}%
The utility possibility set is the convex hull of these profiles,%
	\footnote{In-between profiles are achieved by rapidly switching back and forth (or randomising).}
and the frontier $F^0$ is its upper boundary.
The agent privately observes when a new allocation $\left( u^1, v^1 \right)$ becomes available.
The principal likes the new allocation better than any other,
whereas the agent prefers the principal's favourite old allocation $\left( u^0, v^0 \right)$.
Thus utility possibilities expand, but there is a conflict of interest.

\begin{example}
	\label{example:talent-hoarding}
	The simplest formalisation of the talent-hoarding story from the introduction is as follows. A worker belongs to a team in an organisation. Her productivity on the team is $v^0>0$, while her productivity outside of the team is strictly lower, normalised to zero. At some uncertain time, she acquires a skill that can be exercised only outside of her current team, at productivity $v^1>v^0$. (This could be the skill to manage a team of her own, for example.) Headquarters (the principal) cares about output, while the worker's manager (the agent) has a pure empire-building motive: her payoff is $u=1$ if the worker is on her team and $u=0$ otherwise. In this case, the frontiers are given by $F^0(u) = u v^0$ and $F^1(u) = (1-u)v^1 + uv^0$ for each $u \in [0,1]$.%
		\footnote{And $F^0(u) = F^1(u) = -\infty$ for all $u \in (1,\infty)$, since $u>1$ is impossible.}
	These frontiers satisfy our model assumptions; in particular, the conflict-of-interest assumption holds since $u^1 = 0 < 1 = u^0$.
\end{example}

Richer applications feature (infinitely) many allocations.
In our application to unemployment insurance (§\ref{sec:pf}), for example, an allocation specifies the worker's consumption and (if she is employed) her labour supply.

Our abstract treatment of allocations allows for a broad range of applications.
Allocations may be multi-dimensional, for example,
with some dimensions corresponding to observable actions taken by the agent.
(The principal controls these by issuing action recommendations, backed by the threat of giving the agent zero utility forever unless she complies.)
One dimension of the allocation may describe monetary payments to the agent; we discuss this possibility in §\ref{sec:model:discussion} below.

Rich downstream interactions between the principal and agent can be accommodated by re-interpreting the frontier $F^1$ in lifetime terms, so that $F^1(u)$ is the principal's continuation utility from the post-disclosure interaction when she is constrained to provide the agent with a continuation utility of $u$.%
	\footnote{The legitimacy of this re-interpretation is formally established in §\ref{sec:problem:after} below. Note that the pre-breakthrough frontier $F^0$ cannot be re-interpreted in this `lifetime' fashion.}
The post-disclosure interaction could be one of contracting under (rich, possibly dynamic) moral hazard, for example: that yields a frontier $F^1$ which satisfies our shape assumptions \parencite[see e.g.][Figure 1]{Sannikov2008}.

\subsection{Discussion of the assumptions}
\label{sec:model:discussion}

Two of our assumptions are economically substantive.
First, the agent privately observes a technological breakthrough, but cannot utilise the new technology without the principal's knowledge.
Many economic environments have this feature: in unemployment insurance, for instance, the state observes the worker's employment status (from e.g. tax records).

Secondly, there is a conflict of interest, captured by $u^1 < u^0$.
Such conflicts arise naturally in applications: in unemployment insurance, for example, the state (principal) would like an employed worker (agent) to work and pay taxes, but the worker would rather not.
Absent a conflict of interest, the principal can attain first-best (see \Cref{remark:first-best} below).

Many of the remaining model assumptions are innocuous, as we next briefly relate.
Further details are provided in \cref{suppl:ext}.

\paragraph{Utility possibilities {\normalfont(details: §\ref{suppl:ext:drop_F1-F0_strict}--§\ref{suppl:ext:participation})}.}
The assumption that $F^1 \geq F^0$ is without loss of generality (the old technology remains available after the breakthrough, so the principal can still attain utility $\geq F^0(u)$ while giving the agent utility $u$, for every $u \in [0,\infty)$).
The assumption that the frontiers are concave is likewise without loss:
if one of them were not,
then the principal could get arbitrarily close to any point on its concave upper envelope by rapidly switching back and forth between agent utility levels.
Upper semi-continuity is similarly innocuous.
The stipulation that $u^\star$ is a strict local maximum of $F^1-F^0$
essentially just rules out a saddle point,
and is anyway dispensable.

Not every agent utility $u \in [0,\infty)$ need be feasible:
if no physical allocation provides utility $u$, then we let $F^j(u) \coloneqq -\infty$, ensuring that $u$ is never chosen by the principal.
Our assumption that $F^0$ is finite on $\left(0,u^0\right]$ is without loss.

We have required the agent's flow utility $u$ to be non-negative,
meaning that there is a bound (normalised to zero) on how much misery the principal can inflict on the agent.
This assumption may be replaced with a participation constraint without affecting our results.

\paragraph{Distribution.}
The distribution $G$ of the breakthrough time is unrestricted:
it can have atoms, for example, and need not have full support.
We show in §\ref{sec:mh} below that our results extend to the case in which $G$ is endogenously generated by the agent's unobservable exertion of costly effort.

\paragraph{Uncertain technology {\normalfont(details: §\ref{suppl:ext:random_F1})}.}
Our analysis applies unchanged if the new frontier $F^1$ is random, provided the agent does not have private information about its realisation.

\paragraph{Cheap talk.}
Nothing changes if the agent's disclosures are non-verifiable, provided the principal observes her own payoffs in real time,
since she can then verify cheap-talk reports at negligible cost.%
	\footnote{Following a report, the principal can provide utility $u^1$ for a short time,
	earning $F^1( u^1 )$ if the breakthrough really did occur
	and $F^0( u^1 ) < F^1( u^1 )$ if not.}

\paragraph{(Non-)transferable utility.}
The frontiers $F^0,F^1$ can encode monetary transfers between the principal and agent; our model assumptions restrict such transfers only by requiring that \emph{before} the breakthrough, the agent is protected by (at least a degree of) limited liability.
In detail, write $\widetilde{F}^0,\widetilde{F}^1$ for the frontiers describing utility possibilities \emph{absent} monetary transfers. If the principal gives the agent gross utility $\widetilde{u} \in [0,\infty)$ and pays her $w \in \R$, then net flow utilities are $\widetilde{u} + w$ for the agent and $\widetilde{F}^j\bigl(\widetilde{u}\bigr) - w$ for the principal when technology $j \in \{0,1\}$ is used.
Any constraints on payments, such as limited liability, are captured by constraint sets $W^0 \subseteq \R$ before disclosure and $W^1 \subseteq \R$ after disclosure. For $j \in \{0,1\}$, the utility possibility frontier $F^j$ equals the concave upper semi-continuous upper envelope of
\begin{equation*}
	u \mapsto \sup_{w \in W^j}
	\left\{ \widetilde{F}^j\left(\widetilde{u}\right) - w : \widetilde{u}+w = u \right\} .
\end{equation*}
Assume that $\widetilde{F}^0,\widetilde{F}^1$ satisfy the model assumptions. Then $F^0,F^1$ also satisfy all model assumptions, \emph{except} possibly for the conflict-of-interest assumption $u^1 < u^0$. 
What is needed for $u^1 < u^0$ to hold is that the agent be protected by a degree of limited liability \emph{before} the breakthrough, i.e. $\inf W_0 \geq -k$ for some $k \in \R_+$; in particular, this condition with $k=0$ is sufficient, and it is necessary for this condition to hold with \emph{some} $k \geq 0$.%
	\footnote{Write $\widetilde{u}^0,\widetilde{u}^1$ for the peaks of $\widetilde{F}^0,\widetilde{F}^1$, and note that $u^1 \leq \widetilde{u}^1 < \widetilde{u}^0 \geq u^0$. If $\inf W_0 \geq 0$ then $u^0 = \widetilde{u}^0$, so $u^1 \leq \widetilde{u}^1 < \widetilde{u}^0 = u^0$. If $\inf W_0 < -k$ for every $k \in \R_+$, then $u^0 = 0 \leq u^1$.}
The model assumptions imply no restrictions on \emph{post}-disclosure payments $W^1$.

\subsection{Mechanisms and incentive-compatibility}
\label{sec:model:ic}

A \emph{mechanism} specifies, for each period $t \in \R_+$,
the flow utility $x^0_t$ that the agent enjoys at $t$ if she has not yet disclosed,
as well as the continuation utility $X^1_t$ that she earns by disclosing at $t$.
Formally, a mechanism is a pair $\left( x^0, X^1 \right)$,
where $x^0 : \R_+ \to \R_+$ and $X^1 : \R_+ \to [0,\infty]$ are Lebesgue-measurable.
We call $x^0$ the \emph{pre-disclosure flow,} and $X^1$ the \emph{disclosure reward.}

(Our notation uses lowercase for flows and uppercase for stocks: flow utilities are $x_t \in [0,\infty)$, while continuation payoffs are $X_t \in [0,\infty]$. As usual, `$x$' and `$X$' denote the functions $t \mapsto x_t$ and $t \mapsto X_t$, respectively.)

Note that the description of a mechanism does not specify what utility flow $s \mapsto x^{1,t}_s$ the agent enjoys after disclosing at $t$, only its present value
\begin{equation*}
	X^1_t = r \int_t^\infty e^{-r(s-t)} x^{1,t}_s \dd s 
\end{equation*}
(which may be equal to $\infty$). Nor does the definition specify which technology is used when both are available.
These omissions do not matter for the agent's incentives,
so we shall address them when we formulate the principal's problem (next section).

A mechanism is \emph{incentive-compatible (`IC')} iff the agent prefers disclosing promptly to \ref{item:ic:delay}~disclosing with a delay or \ref{item:ic:never}~never disclosing. Formally:

\begin{definition}
	\label{definition:ic}
	A mechanism $\left( x^0, X^1 \right)$ is \emph{incentive-compatible (`IC')} iff for every period $t \in \R_+$,
	\begin{enumerate}[label=(\alph*)]

		\item \label{item:ic:delay}
		$X^1_t \geq r \int_t^{t+d} e^{-r(s-t)} x^0_s \dd s + e^{-rd} X^1_{t+d}$
		\;for every $d > 0$, and

		\item \label{item:ic:never}
		$X^1_t \geq r \int_t^\infty e^{-r(s-t)} x^0_s \dd s$.

	\end{enumerate}
\end{definition}

By a revelation principle, 
we may restrict attention to incentive-compatible mechanisms.
(See \cref{suppl:revelation_principle} for details.)

\begin{remark}
	\label{remark:first-best}
	Although we have not yet stated the principal's problem,
	it is clear that her first-best is the mechanism $\left( x^0, X^1 \right) \equiv \left( u^0, u^1 \right)$,
	which fails to be incentive-compatible due to the conflict of interest ($u^1 < u^0$).
	If there were no conflict of interest ($u^1 \geq u^0$), then the first-best would be IC.
\end{remark}

In the sequel, we equip the set $\R_+$ of times with the Lebesgue measure,
so that a `null set of times' means a set of Lebesgue measure zero,
and `almost everywhere (a.e.)' means `except possibly on a null set of times'.

Observe that two IC mechanisms $\left( x^0, X^1 \right)$ and $\bigl( x^{0\dag}, X^1 \bigr)$
which differ only in that $x^0 \neq x^{0\dag}$ on a null set are payoff-equivalent.%
	\footnote{$x^0$ enters payoffs as $\smash{\E_G\left( \int_0^\tau e^{-rt} x^0_t \dd t \right)}$ and $\smash{\E_G\left( \int_0^\tau e^{-rt} F^0\left( x^0_t \right) \dd t \right)}$, respectively.
	Modifying $x^0$ on a null set has no effect on the integrals,
	and thus leaves both players' payoffs unchanged, no matter what the breakthrough distribution $G$.}
For this reason, we shall not distinguish between such mechanisms in the sequel, instead treating them as identical.%
	\footnote{\label{footnote:mech_defn_formal}%
	We term such $( x^0, X^1 )$ and $( x^{0\dag}, X^1 )$ \emph{versions} of each other.
	A mechanism is really an equivalence class:
	a maximal set whose every element is a version of every other.}

\section{The principal's problem}
\label{sec:problem}

In this section, we formulate the principal's problem,
and define undominated and optimal mechanisms.
We then derive an upper bound on the agent's utility in undominated mechanisms.

\subsection{After disclosure}
\label{sec:problem:after}

To determine the principal's payoff,
we must fill in the gaps in the definition of a mechanism.
So fix a mechanism $\left( x^0, X^1 \right)$,
and suppose that the agent discloses at time $t$.
For each of the remaining periods $s \in [t,\infty)$, the principal must determine
\begin{enumerate}

	\item \label{item:prin:tech}
	which technology ($F^0$ or $F^1$) will be used, and

	\item \label{item:prin:util}
	what flow utility $x^{1,t}_s$ the agent will enjoy.

\end{enumerate}
Part~\ref{item:prin:tech} is straightforward: the principal is always (weakly) better off using the new technology.

For \ref{item:prin:util}, the principal must choose a (measurable) utility flow $x^{1,t} : [t,\infty) \to [0,\infty)$ subject to providing the agent with the continuation utility specified by the mechanism:
\begin{equation*}
	r \int_t^\infty e^{-r(s-t)} x^{1,t}_s \dd s 
	= X^1_t .
\end{equation*}
She chooses so as to maximise her post-disclosure payoff
\begin{equation*}
	r \int_t^\infty e^{-r(s-t)} F^1\left( x^{1,t}_s \right) \dd s .
\end{equation*}
Since the frontier $F^1$ is concave, the constant flow $x^{1,t} \equiv X^1_t$ is optimal.

Parts~\ref{item:prin:tech} and \ref{item:prin:util} together
imply that the principal earns a flow payoff of $F^1\left( X^1_t \right)$ forever following a time-$t$ disclosure in a mechanism $\left( x^0, X^1 \right)$.

\subsection{Undominated and optimal mechanisms}
\label{sec:problem:opt}

The principal's payoff from an incentive-compatible mechanism $\left( x^0, X^1 \right)$ is
\begin{equation*}
	\Pi_G\left( x^0, X^1 \right)
	\coloneqq \E_G\left(
	r \int_0^\tau e^{-rt} F^0\left( x^0_t \right) \dd t
	+ e^{-r\tau} F^1\left( X^1_\tau \right)
	\right) ,
\end{equation*}
where the expectation is over the random breakthrough time $\tau \sim G$.%
	\footnote{For IC mechanisms $( x^0, X^1 )$ such that $X^1_\tau = \infty$ with positive probability, we interpret `$F^1(\infty)$' as $\lim_{u \uparrow \infty} F^1(u) = -\infty$, so that $\Pi_G(x^0,X^1) \coloneqq -\infty$.}
Her problem is to maximise her payoff by choosing among IC mechanisms.

A basic adequacy criterion for a mechanism
is that it not be \emph{dominated} by another mechanism,
by which we mean that the alternative mechanism is weakly better under every distribution
and strictly better under at least one:

\begin{definition}
	\label{definition:dominated}
	Let $\left( x^0, X^1 \right)$ and $\bigl( x^{0\dag}, X^{1\dag} \bigr)$
	be incentive-compatible mechanisms.
	The former \emph{dominates} the latter iff
	\begin{equation*}
		\Pi_G\left( x^0, X^1 \right)
		\geq \mathrel{(>)} \Pi_G\left( x^{0\dag}, X^{1\dag} \right)
		\quad \text{for every (some) distribution $G$.}
	\end{equation*}
	An IC mechanism is \emph{undominated} iff no IC mechanism dominates it.
\end{definition}

Domination is a distribution-free concept: the principal weakly prefers a dominating mechanism whatever her belief $G$ about the likely time of the breakthrough. When the principal's belief $G$ makes her exactly indifferent between two mechanisms, one of which dominates the other, choosing the dominating mechanism means maximising the principal's ex-post payoff (which cannot hurt, and seems more prudent if the principal entertains even a little doubt about $G$).

\begin{definition}
	\label{definition:opt}
	An incentive-compatible mechanism is \emph{optimal} for a distribution $G$ iff it maximises $\Pi_G$ and is undominated.
\end{definition}

We show in \cref{suppl:undom_opt_properties} that
undominated and optimal mechanisms exist.

\subsection{An upper bound on the agent's utility}
\label{sec:problem:lequ0}

Absent incentive concerns, the principal never wishes to give the agent utility strictly exceeding $u^0$, since both frontiers are downward-sloping to the right of $u^0$.
The principal could use utility promises in excess of $u^0$ as an incentive tool, however.
This is never worthwhile:

\setcounter{lemma}{-1}
\begin{lemma}
	\label{lemma:lequ0}
	Any undominated incentive-compatible mechanism $\left( x^0, X^1 \right)$
	satisfies $x^0 \leq u^0$ almost everywhere.
\end{lemma}

\begin{proof}
	Let $\left( x^0, X^1 \right)$ be an IC mechanism
	in which $x^0 > u^0$ on a non-null set of times.
	Consider the alternative mechanism
	$\left( \min\left\{ x^0, u^0 \right\}, X^1 \right)$
	in which the agent's pre-disclosure flow is capped at $u^0$.
	This mechanism dominates the original one:
	its pre-disclosure flow is lower, strictly on a non-null set,
	and the frontier $F^0$ is strictly decreasing on $\left[ u^0, \infty \right)$.
	And it is incentive-compatible:
	prompt disclosure is as attractive as in the original (IC) mechanism,
	and disclosing with delay (or never disclosing) is weakly less attractive
	since the agent earns a lower flow payoff $\min\left\{ x^0, u^0 \right\} \leq x^0$ while delaying.
\end{proof}

\section{Keeping the agent indifferent}
\label{sec:indiff}

In this section, we describe how undominated mechanisms incentivise the agent.
This result is a stepping stone to the deadline characterisation of undominated mechanisms that we develop in next two sections.

To formulate the agent's problem in a mechanism $\left( x^0, X^1 \right)$,
let $X^0_t$ denote the period-$t$ present value of the remainder of the pre-disclosure flow $x^0$:
\begin{equation*}
	X^0_t \coloneqq r \int_t^\infty e^{-r(s-t)} x^0_s \dd s .
\end{equation*}
In a period $t$ in which the agent has observed but not yet disclosed the breakthrough,
she chooses between
\begin{itemize}

	\item disclosing promptly (payoff $X^1_t$),

	\item disclosing with any delay $d>0$ (payoff $X^0_t + e^{-rd} \bigl( X^1_{t+d} - X^0_{t+d} \bigr)$), and

	\item never disclosing (payoff $X^0_t$).

\end{itemize}
Incentive-compatibility demands precisely that the agent weakly prefer the first option.
Our first result asserts that in an undominated mechanism, she must in fact be indifferent between all three alternatives:

\setcounter{proposition}{-1}
\begin{proposition}
	\label{proposition:indiff}
	Any undominated incentive-compatible mechanism $\left( x^0, X^1 \right)$
	satisfies $X^0 = X^1$.
\end{proposition}

That is, the reward $X^1_t$ for disclosure must equal the present value $X^0_t = r \int_t^\infty e^{-r(s-t)} x^0_s \dd s$ of the remainder of the pre-disclosure flow $x^0$.

A naïve intuition for \Cref{proposition:indiff} is that, were the agent strictly to prefer prompt disclosure in some period $t$, the principal could reduce her disclosure reward $X^1_t$ without violating IC.
The trouble with this idea is that if $X^1_t \leq u^1$, then lowering $X^1_t$ would \emph{hurt} the principal (refer to \Cref{fig:frontiers} on \cpageref{fig:frontiers}).
This is no mere quibble, for (as we shall see) undominated mechanisms will spend time in $\left[ 0, u^1 \right]$.
More broadly, in a general dynamic environment, it is not clear that IC ought to bind everywhere.

The proof is in \cref{app:pf_theorem_indiff}.
Below, we outline the main idea in discrete time,
then highlight the additional details that arise in continuous time.

\begin{proof}[Sketch proof]
	Let time $t \in \{0,1,2,\dots\}$ be discrete,
	and write $\beta \coloneqq e^{-r}$ for the discount factor.
	A mechanism $\left( x^0, X^1 \right)$ is incentive-compatible iff
	in each period $s$,
	the agent prefers prompt disclosure
	to delaying by one period
	and to never disclosing:
	\begin{align*}
		X^1_s &\geq (1-\beta) x^0_s + \beta X^1_{s+1}
		\label{eq:delay_IC}
		\tag{delay IC}
		\\
		X^1_s &\geq X^0_s .
		\label{eq:nondisc_IC}
		\tag{non-disclosure IC}
	\end{align*}
	(\hyperref[eq:delay_IC]{Delay IC} also deters delay by two or more periods.)
	We shall show that undominatedness requires
	that the \ref{eq:delay_IC} inequalities be equalities;
	we omit the argument that \ref{eq:nondisc_IC} must also hold with equality.

	So let $\left( x^0, X^1 \right)$ be an IC mechanism with \ref{eq:delay_IC} slack in some period $t$:
	\begin{equation*}
		X^1_t > (1-\beta) x^0_t + \beta X^1_{t+1} .
	\end{equation*}
	Observe that if the terms $x^0_t$ and $X^1_{t+1}$ on the right-hand side are $\geq u^1$,
	then the left-hand side $X^1_t$ must strictly exceed $u^1$. Equivalently, it must be that either
	\begin{equation*}
		\text{\hyperdest{item:indiff_discrete_pf:case1}(i)
		$X^1_t > u^1$,}\quad
		\text{\hyperdest{item:indiff_discrete_pf:case2}(ii)
		$x^0_t < u^1$,}\quad\text{or}\quad
		\text{\hyperdest{item:indiff_discrete_pf:case3}(iii)
		$X^1_{t+1} < u^1$.}
	\end{equation*}
	In each of these cases, we shall find a mechanism that dominates $\left( x^0, X^1 \right)$.
	We will use the fact that \ref{eq:nondisc_IC} is slack in each period $s \leq t$.%
		\footnote{To prove this, use induction on $s \in \{t,t-1,\dots,2,1,0\}$. In the base case $s=t$,
		$X^1_t
		> (1-\beta) x^0_t + \beta X^1_{t+1}
		\geq (1-\beta) x^0_t + \beta X^0_{t+1}
		\equiv X^0_t$
		by period-$t$ \ref{eq:delay_IC} (which is slack) and period-$(t+1)$ \ref{eq:nondisc_IC}. For the induction step, suppose that period-$(s+1)$ \ref{eq:nondisc_IC} is slack, where $s<t$; then
		$X^1_s
		\geq (1-\beta) x^0_s + \beta X^1_{s+1}
		> (1-\beta) x^0_s + \beta X^0_{s+1}
		\equiv X^0_s$,
		where the weak inequality holds by period-$s$ \ref{eq:delay_IC}.}

	In \hyperlink{item:indiff_discrete_pf:case1}{case~(i)}, the naïve intuition is vindicated: lowering $X^1_t$ toward $u^1$ really does improve the principal's payoff (strictly in case of a breakthrough in period $t$).
	And this preserves IC:
	the (slack) period-$t$ \ref{eq:delay_IC} and \ref{eq:nondisc_IC} hold for a small enough decrease,
	while \ref{eq:delay_IC} \emph{slackens} in period $t-1$ and is unaffected in all other periods, and
	\hyperref[eq:nondisc_IC]{non-disclosure IC} is unaffected in all periods other than $t$.

	In \hyperlink{item:indiff_discrete_pf:case2}{case~(ii)}, increase $x^0_t$ toward $u^1$, by an amount small enough to preserve the (slack) period-$t$ \ref{eq:delay_IC} and period-$s$ \ref{eq:nondisc_IC} for each $s \leq t$.
	Other periods' \ref{eq:delay_IC} is undisturbed, and so is \ref{eq:nondisc_IC} in periods $s > t$.
	Since $F^0$ increases strictly to the left of $u^1 < u^0$,
	the principal's payoff improves (strictly in case of a breakthrough after $t$).

	Finally, in \hyperlink{item:indiff_discrete_pf:case3}{case~(iii)}, \emph{increase} $X^1_{t+1}$ toward $u^1$. (The opposite of the naïve intuition.)
	The principal is better off (strictly in case of a period-$(t+1)$ breakthrough).
	Period-$t$ \ref{eq:delay_IC} abides provided the modification is small,
	while \ref{eq:delay_IC} is loosened in period $t+1$ and unaffected in other periods.
	\hyperref[eq:nondisc_IC]{Non-disclosure IC} is clearly preserved.
\end{proof}

The proof in \cref{app:pf_theorem_indiff} is based on the logic of the sketch above,
but must handle two issues that arise in continuous time.
First, in \hyperlink{item:indiff_discrete_pf:case2}{case~(ii)}, $x^0$ must be increased on a \emph{non-null} set of times if the principal's payoff is to increase strictly under some distribution.
Secondly, in cases \hyperlink{item:indiff_discrete_pf:case1}{(i)} and \hyperlink{item:indiff_discrete_pf:case3}{(iii)}, it is typically not possible to modify $X^1$ in a single period while preserving IC.

In light of \Cref{proposition:indiff}, an undominated incentive-compatible mechanism $\left( x^0, X^1 \right)$ is pinned down by the pre-disclosure flow $x^0$, since the disclosure reward $X^1$ must always equal the present value of the remainder of $x^0$:
\begin{equation*}
	X^1_t 
	= X^0_t 
	= r \int_t^\infty e^{-r(s-t)} x^0_s \dd s 
	\quad \text{for each $t \in \R_+$.}
\end{equation*}
We therefore drop superscripts in the sequel, writing an IC mechanism simply as $(x,X)$, where $X_t \coloneqq r \int_t^\infty e^{-r(s-t)} x_s \dd s$ for each $t \in \R_+$.
Since mechanisms of this form are automatically IC, we refer to them simply as `mechanisms'.
By \Cref{lemma:lequ0}, we need only consider mechanisms $(x,X)$ that satisfy $x \leq u^0$ a.e.

\section{Deadline mechanisms}
\label{sec:deadline}

In this section, we uncover a deadline structure of undominated mechanisms when the old utility possibility frontier $F^0$ is affine on $\left[ 0, u^0 \right]$, as in \Cref{fig:frontiers_affine}.
\ustartrue													
\concavefalse												
\begin{SCfigure}
	\centering
	\input{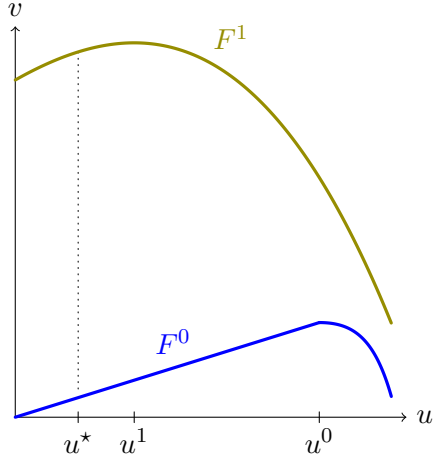}
	\caption{Utility possibility frontiers in the affine case. $u^\star$ is where the frontiers are furthest apart.}
	\label{fig:frontiers_affine}
\end{SCfigure}%
We further characterise the optimal choice of deadline, given the breakthrough distribution.

We start with the affine case partly for reasons of conceptual clarity:
this case lays bare a `front-loading' force that will provide the key to understanding undominated mechanisms in general.
The affine case is also important in its own right, since affineness frequently arises in applications, for two basic reasons.
The first reason is convexification (recall \Cref{fig:projects} on \cpageref{fig:projects}).
In the simplest case, with just two allocations,
the utility possibility frontier is the straight line connecting the two feasible utility profiles.%
	\footnote{In-between profiles are attained by rapidly switching back and forth (or randomising).}
More generally, the utility possibility set is the convex hull of all feasible utility profiles, so its upper boundary $F^0$ is affine if there are two profiles such that the line segment connecting them lies above all other profiles.

The second reason is that in (utilitarian) policy applications, such as unemployment insurance (§\ref{sec:pf} below), the agent's utility directly enters the principal's payoff in a linear fashion. Explicitly, the agent's utility is $u = \phi(a)$, where $a \in \mathcal{A}$ is a policy variable and $\phi : \mathcal{A} \to [0,\infty)$ is surjective, and the principal's utility is $v = u - \psi(a)$ for some function $\psi : \mathcal{A} \to \R$, so
\begin{equation*}
	F^0(u)
	= \sup_{a \in \mathcal{A}} \left\{ u - \psi(a) : u=\phi(a) \right\}
	= u - \inf_{a \in \mathcal{A}} \left\{ \psi(a) : u=\phi(a) \right\}
\end{equation*}
for each $u \in [0,\infty)$. The first term is linear, so if the second term has almost no curvature, then $F^0$ is approximately affine.%
	\footnote{For example, if $\mathcal{A}$ is a convex subset of $\R$ and $\phi,\psi$ are twice continuously differentiable with $\phi' > 0 < \psi'$, then $F^0(u) = u - \psi( \phi^{-1}(u) )$ for each $u \in [0,\infty)$, so the curvature $\abs{F^{0\prime\prime} / F^{0\prime} }$ is small if the curvature difference $\abs{ \psi''/\psi' - \phi''/\phi' }$ is small.}

The utility level $u^\star$ (defined in §\ref{sec:model}) admits a simple description when $F^0$ is affine:
it is the unique $u \in \left[ 0, u^0 \right]$ at which the frontiers are furthest apart,%
	\footnote{$u^\star$ is a strict local maximum of the gap $F^1-F^0$, which is concave when $F^0$ is affine.}
as indicated in \Cref{fig:frontiers_affine}.
A \emph{deadline mechanism} is one in which the agent's utility absent disclosure
is at the efficient level $u^0$ before a deterministic deadline,
and at the inefficiently low level $u^\star$ afterwards:

\begin{definition}
	\label{definition:deadline}
	A mechanism $(x,X)$ is a \emph{deadline mechanism} iff
	\begin{equation*}
		x_t =
		\begin{cases}
			u^0		& \text{for $t \leq T$} \\
			u^\star	& \text{for $t > T$} 
		\end{cases}
		\qquad \text{for some $T \in [0,\infty]$.}
	\end{equation*}
\end{definition}

Deadline mechanisms are simple:
only two utility levels are used, with a single switch between them.
And they form a small class of mechanisms, parametrised by a single number: the deadline $T$.
(The utility levels $u^0$ and $u^\star$ are not free parameters, being pinned down by the technologies $F^0,F^1$.)

The agent's reward $X$ upon disclosure in a deadline mechanism
(equal to the present value of the remainder of the pre-disclosure flow $x$)
is decreasing until the deadline, then constant at $u^\star$:
\begin{equation*}
	X_t =
	\begin{cases}
		\left( 1 - e^{-r(T-t)} \right) u^0 + e^{-r(T-t)} u^\star
		& \text{for $t \leq T$} \\
		u^\star
		& \text{for $t > T$.} 
	\end{cases}
	\tag{$\diamondsuit$}
	\label{eq:deadline_X}
\end{equation*}

\subsection{Only deadline mechanisms are undominated}
\label{sec:deadline:theorem}

The affine case admits a sharp prediction:
no matter what the shapes of the new frontier $F^1$ and breakthrough distribution $G$,
the principal will choose a mechanism from the small and simple deadline class.

\begin{theorem}
	\label{theorem:deadline}
	If the old frontier $F^0$ is affine on $\left[ 0, u^0 \right]$,
	then any undominated mechanism is a deadline mechanism.
\end{theorem}

The welfare implications are stark: ex-post Pareto efficiency in case of an early breakthrough,
and surplus destruction otherwise.
In particular, absent a breakthrough, the old technology is operated Pareto-efficiently (i.e. on the downward-sloping part of $F^0$, specifically at $u^0$) before the deadline, and inefficiently (at $u^\star$) afterwards.
Once the new technology arrives,
it is deployed efficiently (on the downward-sloping part of $F^1$) if its arrival was early (while $X \geq u^1$).%
	\footnote{A detail: $X_t \geq u^1$ holds in early periods $t$ only if the deadline is sufficiently late. We show in the next section that this must be the case in undominated mechanisms.}
If its arrival was late, then $F^1$ is operated inefficiently if $u^\star < u^1$, and efficiently if $u^\star = u^1$.
These welfare implications, as well as the special role played by $u^\star$, are general properties that hold even outside of the affine case, so we postpone discussing them fully until §\ref{sec:opt:discussion} below.

We prove \Cref{theorem:deadline} in \cref{app:pf_theorem_deadline}.
Below, we give an intuitive sketch.

\begin{proof}[Sketch proof]
	Fix a non-deadline mechanism $(x,X)$ with $x \leq u^0$, and assume for simplicity that $x \geq u^\star$.
	We will show that $(x,X)$ is dominated by the deadline mechanism $\bigl( x^\dag, X^\dag \bigr)$ whose deadline $T$ satisfies
	\begin{equation*}
		\underbrace{ \left( 1 - e^{-rT} \right) u^0 + e^{-rT} u^\star }
		_{\textstyle = X_0^\dag \quad\text{by \eqref{eq:deadline_X}} }
		= X_0 .
	\end{equation*}
	This mechanism is a \emph{front-loading} of $(x,X)$:
	the pre-disclosure flow has the same present value $X_0 = r \int_0^\infty e^{-rt} x_t \dd t$,
	but is higher early and lower late,
	as depicted in \Cref{fig:deadline_sketch:U}.
	\Utrue													
	\Udltrue												
	\begin{figure}
		\begin{subfigure}{0.48\textwidth}
			\centering
			\Uhatfalse										
			\Uhatdlfalse									



\ifUdl
	\Ttrue
\fi
\ifUhatdl
	\Ttrue
\fi
\newif\ifdlblue
\dlbluetrue
\ifU\else
	\ifUhat\else
		\dlbluefalse
	\fi
\fi




\begin{tikzpicture}[scale=1, line cap=round,
	declare function={
		Uhat_orig(\t,\k,\eta,\alph,\bet,\cons)				
			= \k * 2.71828^(-\eta*\t)
			* sin( deg( \alph * \t - \bet ) )
			+ \cons ;
		U_orig(\t,\r,\k,\eta,\alph,\bet,\cons)				
			= ((\r+\eta)/\r) * \k
			* 2.71828^(-\eta*\t)
			* sin( deg( \alph * \t - \bet ) )
			- ((\alph*\k)/\r)
			* 2.71828^(-\eta*\t)
			* cos( deg( \alph*\t - \bet ) )
			+ \cons;
		Uhat_deadline(\t,\uzero,\ustar,\T,\r)				
			= ( 1 - 2.71828^(-\r*(\T-\t)) )
			* \uzero
			+ 2.71828^(-\r*(\T-\t))
			* \ustar;
		T_fn(\uzero,\ustar,\r,\Uhatzero)					
			= - (1/\r) 
			* ln( (\uzero-\Uhatzero)
			/ (\uzero-\ustar) );
		}]

\pgfmathsetmacro{\ticklength}{1/14};						
\pgfmathsetmacro{\samples}{100};							

\pgfmathsetmacro{\Ulabelpos}{1.3};							
\pgfmathsetmacro{\Udllabelpos}{1.0};						
\pgfmathsetmacro{\Uhatlabelpos}{3.2};						
\pgfmathsetmacro{\Uhatdllabelpos}{0.8};						

\pgfmathsetmacro{\uzero}{4.5};								
\pgfmathsetmacro{\ustar}{0.5};								
\pgfmathsetmacro{\r}{0.18};									
\pgfmathsetmacro{\k}{0.8};									
\pgfmathsetmacro{\eta}{1.0};								
\pgfmathsetmacro{\alph}{1.4};								
\pgfmathsetmacro{\bet}{-0.9};								
\pgfmathsetmacro{\cons}{0.14 * (\r+\eta)/\r};				

\pgfmathsetmacro{\Uhatzero}
	{Uhat_orig(0,\k,\eta,\alph,\bet,\cons)};
\pgfmathsetmacro{\T}
	{T_fn(\uzero,\ustar,\r,\Uhatzero)};

\draw[->] (0.0,0.0) -- (5.2,0.0);							
\draw (5.2,0.0) node[anchor=west] {$t$};					
\draw[->] (0.0,0.0) -- (0.0,5.2);							
\draw (0.0,5.2) node[anchor=south] {$u$};					

\ifUhat
	\draw													
		[domain=0.0:5.0, variable=\t,
		samples=\samples, very thick, opacity=0.4]
		plot ( {\t},
		{ Uhat_orig(\t,\k,\eta,\alph,\bet,\cons) } );
	\draw[opacity=0.5]										
		( { \Uhatlabelpos },
		{ Uhat_orig(%
		\Uhatlabelpos,\k,\eta,\alph,\bet,\cons) } )
		node[anchor=south]
		{$X$};
\fi

\ifUhatdl
	\ifdlblue
		\draw												
			[domain=0.0:\T, variable=\t,
			samples=\samples, very thick, blue,
			opacity=0.4]
			plot ( {\t},
			{ Uhat_deadline(\t,\uzero,\ustar,\T,\r) } );
		\draw[-, very thick, blue, opacity=0.4]				
			( \T, \ustar )
			-- ( 5.0, \ustar );
		\draw[blue,opacity=0.6]								
			( { \Uhatdllabelpos },
			{ Uhat_deadline(%
				\Uhatdllabelpos,\uzero,\ustar,\T,\r) - 0.1 } )
			node[anchor=north]
			{$X^\dag$};
	\else
		\draw												
			[domain=0.0:\T, variable=\t,
			samples=\samples, very thick, orange]
			plot ( {\t},
			{ Uhat_deadline(\t,\uzero,\ustar,\T,\r) } );
		\draw[-, very thick, orange]						
			( \T, \ustar )
			-- ( 5.0, \ustar );
		\draw[orange]										
			( { \Uhatdllabelpos },
			{ Uhat_deadline(%
				\Uhatdllabelpos,\uzero,\ustar,\T,\r) - 0.1 } )
			node[anchor=south west]
			{$X$};
	\fi
\fi

\ifU
	\draw													
		[domain=0.0:5.0, variable=\t,
		samples=\samples, very thick, dotted]
		plot ( {\t},
		{ U_orig(\t,\r,\k,\eta,\alph,\bet,\cons) } );
	\draw													
		( { \Ulabelpos },
		{ U_orig(%
		\Ulabelpos,\r,\k,\eta,\alph,\bet,\cons) } )
		node[anchor=west]
		{$x$};
\fi

\ifUdl
	\ifdlblue
		\draw[-, very thick, dotted, blue]					
			( 0, \uzero )
			-- ( \T, \uzero );
		\draw[-, very thick, dotted, blue]					
			( \T, \ustar )
			-- ( 5.0, \ustar );
		\draw[blue]											
			( { \Udllabelpos },
			{ \uzero } )
			node[anchor=south]
			{$x^\dag$};
	\else
		\draw[-, very thick, dotted]						
			( 0, \uzero )
			-- ( \T, \uzero );
		\draw[-, very thick, dotted]						
			( \T, \ustar )
			-- ( 5.0, \ustar );
		\draw												
			( { \Udllabelpos },
			{ \uzero } )
			node[anchor=south]
			{$x$};
	\fi
\fi

\ifT
	\draw[-] ( \T, - \ticklength )							
		-- ( \T, \ticklength );
	\draw ( \T, 0 )											
		node[anchor=north] {$\strut T$};
\else
	\draw ( \T, 0 )											
		node[anchor=north] {\phantom{$\strut T$}};
\fi

\draw[-] ( - \ticklength, \uzero )							
	-- ( \ticklength, \uzero );
\draw ( 0, \uzero )											
	node[anchor=east] {$\strut u^0$};
\draw[-] ( - \ticklength, \ustar )							
	-- ( \ticklength, \ustar );
\draw ( 0, \ustar )											
	node[anchor=east] {$\strut u^\star$};

\end{tikzpicture}

			\caption{$x^\dag$ is higher early and lower late.}
			\label{fig:deadline_sketch:U}
		\end{subfigure}
		\hfill
		\begin{subfigure}{0.48\textwidth}
			\centering
			\Uhattrue										
			\Uhatdltrue										



\ifUdl
	\Ttrue
\fi
\ifUhatdl
	\Ttrue
\fi
\newif\ifdlblue
\dlbluetrue
\ifU\else
	\ifUhat\else
		\dlbluefalse
	\fi
\fi




\begin{tikzpicture}[scale=1, line cap=round,
	declare function={
		Uhat_orig(\t,\k,\eta,\alph,\bet,\cons)				
			= \k * 2.71828^(-\eta*\t)
			* sin( deg( \alph * \t - \bet ) )
			+ \cons ;
		U_orig(\t,\r,\k,\eta,\alph,\bet,\cons)				
			= ((\r+\eta)/\r) * \k
			* 2.71828^(-\eta*\t)
			* sin( deg( \alph * \t - \bet ) )
			- ((\alph*\k)/\r)
			* 2.71828^(-\eta*\t)
			* cos( deg( \alph*\t - \bet ) )
			+ \cons;
		Uhat_deadline(\t,\uzero,\ustar,\T,\r)				
			= ( 1 - 2.71828^(-\r*(\T-\t)) )
			* \uzero
			+ 2.71828^(-\r*(\T-\t))
			* \ustar;
		T_fn(\uzero,\ustar,\r,\Uhatzero)					
			= - (1/\r) 
			* ln( (\uzero-\Uhatzero)
			/ (\uzero-\ustar) );
		}]

\pgfmathsetmacro{\ticklength}{1/14};						
\pgfmathsetmacro{\samples}{100};							

\pgfmathsetmacro{\Ulabelpos}{1.3};							
\pgfmathsetmacro{\Udllabelpos}{1.0};						
\pgfmathsetmacro{\Uhatlabelpos}{3.2};						
\pgfmathsetmacro{\Uhatdllabelpos}{0.8};						

\pgfmathsetmacro{\uzero}{4.5};								
\pgfmathsetmacro{\ustar}{0.5};								
\pgfmathsetmacro{\r}{0.18};									
\pgfmathsetmacro{\k}{0.8};									
\pgfmathsetmacro{\eta}{1.0};								
\pgfmathsetmacro{\alph}{1.4};								
\pgfmathsetmacro{\bet}{-0.9};								
\pgfmathsetmacro{\cons}{0.14 * (\r+\eta)/\r};				

\pgfmathsetmacro{\Uhatzero}
	{Uhat_orig(0,\k,\eta,\alph,\bet,\cons)};
\pgfmathsetmacro{\T}
	{T_fn(\uzero,\ustar,\r,\Uhatzero)};

\draw[->] (0.0,0.0) -- (5.2,0.0);							
\draw (5.2,0.0) node[anchor=west] {$t$};					
\draw[->] (0.0,0.0) -- (0.0,5.2);							
\draw (0.0,5.2) node[anchor=south] {$u$};					

\ifUhat
	\draw													
		[domain=0.0:5.0, variable=\t,
		samples=\samples, very thick, opacity=0.4]
		plot ( {\t},
		{ Uhat_orig(\t,\k,\eta,\alph,\bet,\cons) } );
	\draw[opacity=0.5]										
		( { \Uhatlabelpos },
		{ Uhat_orig(%
		\Uhatlabelpos,\k,\eta,\alph,\bet,\cons) } )
		node[anchor=south]
		{$X$};
\fi

\ifUhatdl
	\ifdlblue
		\draw												
			[domain=0.0:\T, variable=\t,
			samples=\samples, very thick, blue,
			opacity=0.4]
			plot ( {\t},
			{ Uhat_deadline(\t,\uzero,\ustar,\T,\r) } );
		\draw[-, very thick, blue, opacity=0.4]				
			( \T, \ustar )
			-- ( 5.0, \ustar );
		\draw[blue,opacity=0.6]								
			( { \Uhatdllabelpos },
			{ Uhat_deadline(%
				\Uhatdllabelpos,\uzero,\ustar,\T,\r) - 0.1 } )
			node[anchor=north]
			{$X^\dag$};
	\else
		\draw												
			[domain=0.0:\T, variable=\t,
			samples=\samples, very thick, orange]
			plot ( {\t},
			{ Uhat_deadline(\t,\uzero,\ustar,\T,\r) } );
		\draw[-, very thick, orange]						
			( \T, \ustar )
			-- ( 5.0, \ustar );
		\draw[orange]										
			( { \Uhatdllabelpos },
			{ Uhat_deadline(%
				\Uhatdllabelpos,\uzero,\ustar,\T,\r) - 0.1 } )
			node[anchor=south west]
			{$X$};
	\fi
\fi

\ifU
	\draw													
		[domain=0.0:5.0, variable=\t,
		samples=\samples, very thick, dotted]
		plot ( {\t},
		{ U_orig(\t,\r,\k,\eta,\alph,\bet,\cons) } );
	\draw													
		( { \Ulabelpos },
		{ U_orig(%
		\Ulabelpos,\r,\k,\eta,\alph,\bet,\cons) } )
		node[anchor=west]
		{$x$};
\fi

\ifUdl
	\ifdlblue
		\draw[-, very thick, dotted, blue]					
			( 0, \uzero )
			-- ( \T, \uzero );
		\draw[-, very thick, dotted, blue]					
			( \T, \ustar )
			-- ( 5.0, \ustar );
		\draw[blue]											
			( { \Udllabelpos },
			{ \uzero } )
			node[anchor=south]
			{$x^\dag$};
	\else
		\draw[-, very thick, dotted]						
			( 0, \uzero )
			-- ( \T, \uzero );
		\draw[-, very thick, dotted]						
			( \T, \ustar )
			-- ( 5.0, \ustar );
		\draw												
			( { \Udllabelpos },
			{ \uzero } )
			node[anchor=south]
			{$x$};
	\fi
\fi

\ifT
	\draw[-] ( \T, - \ticklength )							
		-- ( \T, \ticklength );
	\draw ( \T, 0 )											
		node[anchor=north] {$\strut T$};
\else
	\draw ( \T, 0 )											
		node[anchor=north] {\phantom{$\strut T$}};
\fi

\draw[-] ( - \ticklength, \uzero )							
	-- ( \ticklength, \uzero );
\draw ( 0, \uzero )											
	node[anchor=east] {$\strut u^0$};
\draw[-] ( - \ticklength, \ustar )							
	-- ( \ticklength, \ustar );
\draw ( 0, \ustar )											
	node[anchor=east] {$\strut u^\star$};

\end{tikzpicture}

			\caption{$X^\dag \leq X$, with equality at $0$.}
			\label{fig:deadline_sketch:Uhat}
		\end{subfigure}
		\caption{Sketch proof of \Cref{theorem:deadline}: front-loading by a deadline mechanism.}
		\label{fig:deadline_sketch}
	\end{figure}
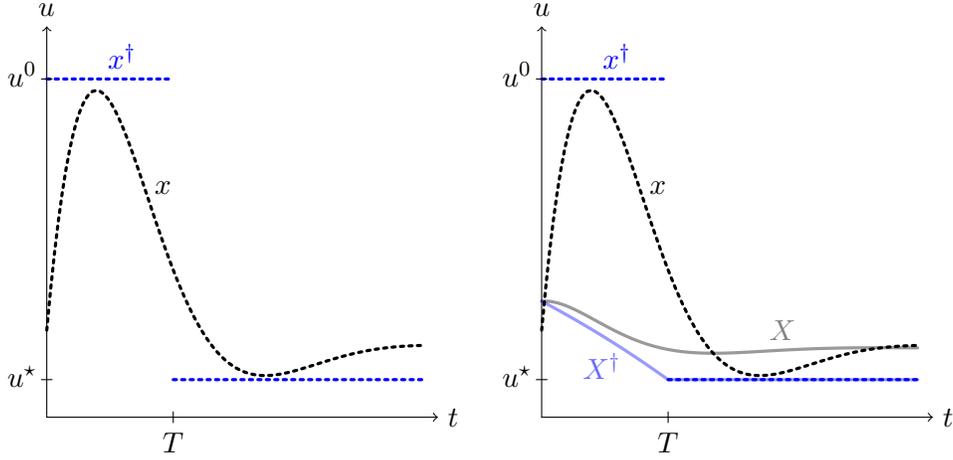%
	As time passes, the present value
	$X^\dag_t = r \int_t^\infty e^{-r(s-t)} x^\dag_s \dd s$
	of the remainder of the front-loaded flow $x^\dag$ rapidly diminishes,
	so that $X^\dag$ is weakly below $X$ in every period (see \Cref{fig:deadline_sketch:Uhat}).

	The principal's period-$t$ continuation payoff if the agent never discloses is
	\begin{equation*}
		Y_t
		\coloneqq r \int_t^\infty e^{-r(s-t)} F^0(x_s) \dd s
		= F^0\left( r \int_t^\infty e^{-r(s-t)} x_s \dd s \right)
 		= F^0( X_t ) ,
	\end{equation*}
	where the middle equality holds by the affineness of $F^0$. Her payoff is thus
	\begin{align*}
		\Pi_G(x,X)
		&= \E_G\left(
		Y_0 
		- e^{-r\tau} Y_\tau
		+ e^{-r\tau} F^1\left( X_\tau \right)
		\right)
		\\
		&= F^0(X_0)
		+ \E_G\left( e^{-r\tau} \left[ F^1 - F^0 \right](X_\tau) \right) .
		\tag*{}
		\label{eq:affine_payoff}
	\end{align*}
	Front-loading lowers $X$ toward $u^\star$, leaving $X_0$ unchanged.
	Since $F^1 - F^0$ is (strictly) decreasing on $\left[ u^\star, u^0 \right]$ by definition of $u^\star$,
	this improves the principal's payoff whatever the distribution $G$.
	The improvement is in fact strict for any full-support distribution.
	Thus $\bigl( x^\dag, X^\dag \bigr)$ dominates $(x,X)$.
\end{proof}

The key simplification in the above sketch is the assumption that $x \geq u^\star$. The proof in \cref{app:pf_theorem_deadline} dispenses with this assumption by choosing the deadline $T$ to satisfy $\bigl( 1 - e^{-rT} \bigr) u^0 + e^{-rT} u^\star = \max\{X_0,u^\star\}$, and showing (in a few extra steps) that this yields a dominating mechanism even if $x \ngeq u^\star$.

\Cref{theorem:deadline} provides a rationale for deadline mechanisms even when $F^0$ is not exactly affine: provided $F^0$ has only moderate curvature, the principal loses little by restricting attention to deadline mechanisms.

\subsection{Undominated deadlines}
\label{sec:deadline:undom_deadline}

\Cref{theorem:deadline} asserts that only deadline mechanisms are undominated when $F^0$ is affine, but does not adjudicate between deadlines.
In fact, not every deadline mechanism is undominated.
Consider a deadline $T$ so early that $X_0 < u^1$.
Since the disclosure reward $X$ decreases over time in a deadline mechanism, we have $X_\tau < u^1$ whatever the time $\tau$ of the breakthrough.

The principal can do better by using the later deadline $\underline{T}$ that satisfies $X_0 = u^1$, or explicitly (using \cref{eq:deadline_X} on \cpageref{eq:deadline_X})
\begin{equation*}
	\left( 1 - e^{-r\underline{T}} \right) u^0
	+ e^{-r\underline{T}} u^\star
	= u^1 .
\end{equation*}
This raises the agent's disclosure reward $X$ toward $u^1$,
improving the principal's post-disclosure payoff $F^1\left( X_\tau \right)$
whatever the breakthrough time $\tau$ (strictly if $\tau<\underline{T}$).
The principal also enjoys the high pre-disclosure flow $F^0\left( u^0 \right) > F^0\left( u^\star \right)$ for longer,
which is beneficial in case of a late breakthrough.

Undominatedness thus requires a deadline no earlier than $\underline{T}$.
This condition is not only necessary, but also sufficient:

\begin{proposition}
	\label{proposition:dl_charac}
	If the old frontier $F^0$ is affine on $\left[ 0, u^0 \right]$,
	then a mechanism is undominated iff
	it is a deadline mechanism with deadline
	$T \in \left[ \underline{T}, \infty \right]$.
\end{proposition}

The proof is in \cref{app:deadline_charac}.

\subsection{Optimal deadlines}
\label{sec:deadline:opt_deadline}

\Cref{proposition:dl_charac} narrows the search for an optimal mechanism
to deadline mechanisms with a sufficiently late deadline.
The optimal choice among these depends on the breakthrough distribution $G$.

A late deadline is beneficial if the breakthrough occurs late,
as the efficient high utility $u^0$ is then provided for a long time.
The cost is that in case of an early breakthrough,
the agent must be given a utility of $X > u^1$ forever.
A first-order condition balances this trade-off:

\begin{proposition}
	\label{proposition:opt_deadline}
	Assume that the old frontier $F^0$ is affine on $\left[ 0, u^0 \right]$,
	that the new frontier $F^1$ is differentiable on $\left( 0, u^0 \right)$,
	and that $u^\star > 0$.
	A mechanism is optimal for $G$ iff it is a deadline mechanism and satisfies $\E_G\left( F^{1\prime}\left( X_\tau \right) \right) = 0$.
\end{proposition}

In other words, the new technology should be operated optimally \emph{on average.}
This is a restriction on the deadline $T$ because $X$ is a function of it,
as described by \cref{eq:deadline_X} on \cpageref{eq:deadline_X}.

We prove \Cref{proposition:opt_deadline} in \cref{app:opt_deadline_pf}
by deriving a general first-order condition that is valid without any auxiliary assumptions,
then showing that it can be written as $\E_G\left( F^{1\prime}\left( X_\tau \right) \right) = 0$ when $F^0$ is affine, $F^1$ is differentiable and $u^\star$ is interior.

In the same \hyperref[app:opt_deadline_pf]{appendix}, we derive comparative statics for optimal deadlines:
they become later when the breakthrough distribution $G$ becomes later in the sense of first-order stochastic dominance.
This improves the agent's ex-ante payoff $X_0$, as can be seen from \cref{eq:deadline_X} on \cpageref{eq:deadline_X}.

\section{Optimal mechanisms in general}
\label{sec:opt}

In this section, we show that optimal mechanisms in the general (non-affine) case
exhibit a graduated deadline structure: absent disclosure, the agent's utility still declines from $u^0$ toward $u^\star$, but not necessarily abruptly.
Given the breakthrough distribution, we describe the optimal path.

\subsection{Qualitative features of optimal mechanisms}
\label{sec:opt:undominated}

Recall from §\ref{sec:model} that $u^\star$ denotes the greatest $u \in \left[ 0, u^0 \right]$
at which the old and new frontiers $F^0,F^1$ have equal slopes,
as depicted in \Cref{fig:frontiers} (\cpageref{fig:frontiers}).

\begin{theorem}
	\label{theorem:opt}
	Any mechanism $(x,X)$ that is optimal for some distribution $G$ with $G(0)=0$ and unbounded support
	has $x$ decreasing
	\begin{equation*}
		\text{from}\quad
		\lim_{t \to 0} x_t = u^0
		\quad \text{toward} \quad
		\lim_{t \to \infty} x_t = u^\star .%
			\footnote{Recall that a mechanism has multiple \emph{versions} (\cref{footnote:mech_defn_formal}, \cpageref{footnote:mech_defn_formal}).
			\Cref{theorem:opt} asserts that any optimal mechanism has a version with the stated properties.
			We focus on $\lim_{t \to 0} x_t$ rather than $x_0$
			because `$x_0=u^0$' is vacuous: \emph{any} mechanism has a version satisfying it.}
	\end{equation*}
\end{theorem}

That is, optimal mechanisms are just like deadline mechanisms, except that the transition from $u^0$ to $u^\star$ may be gradual.
This graduality follows directly from relaxing affineness:
when $F^0$ has a strictly concave shape,
by definition,
the principal prefers providing intermediate utility
to providing only the extreme utilities $u^\star,u^0$.
\Cref{theorem:opt} is the combination of this mechanical effect
with the front-loading insight expressed by \Cref{theorem:deadline}.

Formally, the proof in \cref{app:opt} relies on a form of \emph{local} front-loading to establish monotonicity.
Given monotonicity, it is immediate that $x_t$ converges as $t \to \infty$. We explain in the next section why the limit must be $u^\star$.

The role of monotonicity is \emph{not} to provide incentives: on the contrary, mechanisms of the form $(x,X)$ satisfy IC (with equality) by definition, whatever the pre-disclosure flow $x : \R_+ \to \left[0,u^0\right]$. Rather, what \Cref{theorem:opt} asserts is that if $x$ is not decreasing, then there is a better mechanism. This claim is non-trivial to prove.

Absent a breakthrough, efficiency deteriorates as we travel leftward along the upward-sloping part of the old frontier $F^0$.
Once the new technology becomes available,
it is operated efficiently (on the downward-sloping part of $F^1$) if its arrival was sufficiently early.%
	\footnote{\label{footnote:X1_u1}%
	We show in \cref{app:opt_transition} that $X_t > u^1$ holds in all sufficiently early periods $t$.}
If its arrival was late, then $F^1$ is operated inefficiently if $u^\star < u^1$, and efficiently if $u^\star = u^1$.

The distributional hypotheses are mild:
$G(0)=0$ means that the new technology is unavailable initially,
while unbounded support rules out an effectively finite horizon.
The former's role is as a sufficient condition for $\lim_{t \to 0} x_t = u^0$,
while the latter is required by our proof strategy.

\subsection{Discussion}
\label{sec:opt:discussion}

Two salient features of \Cref{theorem:opt,theorem:deadline} are the special role played by $u^\star$ and the possibility (in case of a late breakthrough) of perpetual surplus destruction. We now discuss these two properties.

For simplicity, assume that $F^0$ and $F^1$ are differentiable, and consider a mechanism that is eventually constant: $x = \widebar{u}$ on $(T,\infty)$, where $\widebar{u} \in \left(0,u^0\right)$ and $G(T)<1$.
Unless $\widebar{u} = u^\star$, the mechanism $(x,X)$ may be improved by a simple perturbation:
\begin{equation*}
	x^\eps =
	\begin{cases}
		x
		& \text{on $[0,T]$} \\
		\widebar{u}+\eps
		& \text{on $\left(T,T+\ln(2)/r\right]$} \\
		\widebar{u}-\eps
		& \text{on $\left[T+\ln(2)/r,\infty\right)$}
	\end{cases}
	\qquad\quad \text{where $\eps \neq 0$.}
\end{equation*}
If $\eps>0$, then this is a `front-loading', making the pre-disclosure flow $x$ higher early on (before $T+\ln(2)/r$) and lower later, while keeping $X^\eps = X$ on $[0,T]$.%
	\footnote{Because $X^\eps_T = X_T + \eps e^{rT} \bigl( \int_T^{T+\ln(2)/r} r e^{-rs} \dd s - \int_{T+\ln(2)/r}^\infty r e^{-rs} \dd s \bigr) = X_T$ for each $\eps$.}
Since $\frac{\dd}{\dd \eps} x^\eps \bigr|_{\eps=0} = \frac{\dd}{\dd \eps} X^\eps \bigr|_{\eps=0} = 0$ on $[0,T]$, perturbing $\eps$ away from zero changes the principal's payoff $\Pi_G\left( x^\eps, X^\eps \right)$ at rate
\begin{multline*}
	\left. \frac{\dd}{\dd \eps}
	\E_G\left(
	r \int_0^\tau e^{-rt}
	F^0\left( x^\eps_t \right)
	\dd t
	\right)
	\right|_{\eps=0}
	+ \left. \frac{\dd}{\dd \eps}
	\E_G\left( e^{-r\tau}
	F^1\left( X^\eps_\tau \right)
	\right) 
	\right|_{\eps=0}
	\\
	\begin{aligned}
		&= \E_G\left( r \int_0^\tau e^{-rt} \left. \tfrac{\dd}{\dd \eps} x^\eps_t \right|_{\eps=0} \dd t \right)
		\times F^{0\prime}\left( \widebar{u} \right)
		+ K_G \times F^{1\prime}\left(\widebar{u}\right)
		\\
		&= K_G \times \left[ F^{1\prime}\left( \widebar{u} \right) - F^{0\prime}\left( \widebar{u} \right) \right]
		\qquad\qquad \text{where $K_G \coloneqq \E_G\left( e^{-r\tau} \left. \tfrac{\dd}{\dd \eps} X^\eps_\tau \right|_{\eps=0} \right)$,}
	\end{aligned}
\end{multline*}
where the second equality holds since the big expectation equals $\E_G( \phi_\tau'(0) )$ where $\phi_\tau(\eps) \coloneqq r \int_0^\tau e^{-rt} x^\eps_t \dd t = X^\eps_0 - e^{-r\tau} X^\eps_\tau$.
Thus \emph{whatever} the breakthrough distribution $G$, the principal's payoff can be improved by perturbing $\eps$ except if $F^{0\prime}\left(\widebar{u}\right)=F^{1\prime}\left(\widebar{u}\right)$, or equivalently $\widebar{u} = u^\star$.

This accounts for the special role of $u^\star$. It also implies the optimality of perpetual surplus destruction in case of a late breakthrough (after $T$), since setting $\widebar{u} = u^\star < u^1$ yields $X = x < u^1$ on $(T,\infty)$.

Economically, the above argument boils down to a demonstration that $u^\star$ balances the cost and benefit of `front-loading', so that neither front-loading ($\eps>0$) nor `back-loading' ($\eps<0$) yields an improvement. As discussed in the proof of \Cref{theorem:deadline} (§\ref{sec:deadline:theorem} above),
the benefit of front-loading is that the pre-disclosure flow $x$ is experienced only before the breakthrough, so making it higher early and lower late is mechanically better.%
	\footnote{The principal prefers a higher pre-disclosure flow since $F^0$ is increasing on $[0,u^0]$.}
The cost of front-loading is that it lowers the disclosure reward $X$, thereby increasing the severity of perpetual surplus destruction in case of a late breakthrough.

\subsection{Optimal transition}
\label{sec:opt:opt_transition}

\Cref{theorem:opt} describes the distribution-free qualitative features of optimal mechanisms,
but does not specify the precise manner in which the agent's utility ought to decline from $u^0$ toward $u^\star$.
The optimal path, for a given breakthrough distribution, is characterised by an Euler equation:

\begin{proposition}
	\label{proposition:opt_transition}
	Assume that $u^\star > 0$
	and that the frontiers $F^0,F^1$ are differentiable on $\left( 0, u^0 \right)$.
	Then any mechanism $(x,X)$ that is optimal for a distribution $G$ with $G(0)=0$ and unbounded support
	satisfies
	the initial condition $\E_G\left( F^{1\prime}\left( X_\tau \right) \right) = 0$
	and the Euler equation
	\begin{equation*}
		F^{0\prime}(x_t) 
		\geq \E_G\left(
		F^{1\prime}\left( X_\tau \right)
		\middle| \tau > t \right)
		\quad \text{for each $t \in \R_+$, with equality if $x_t < u^0$.}%
			\footnote{Here $F^{j\prime}(0)$ ($F^{j\prime}(u^0)$) for $j \in \{0,1\}$ denotes the right-hand (left-hand) derivative.
			Recall that a mechanism has multiple \emph{versions} (\cref{footnote:mech_defn_formal}, \cpageref{footnote:mech_defn_formal}).
			In full, the proposition asserts that some (any) version satisfies the Euler equation for (almost) every $t \in \R_+$.}
	\end{equation*}
\end{proposition}

The initial condition $\E_G\left( F^{1\prime}\left( X_\tau \right) \right) = 0$ demands that the new technology be used optimally on average,
just like the first-order condition for an optimal deadline in the affine case (\Cref{proposition:opt_deadline}, \cpageref{proposition:opt_deadline}).
The special role of $u^\star$ (discussed in the previous section) can be deduced from the Euler equation: as $t \to \infty$, $X_t = r \int_t^\infty e^{-r(s-t)} x_s \dd s$ converges to $\underline{u} \coloneqq \lim_{t \to \infty} x_t$, so $F^{0\prime}\left( \underline{u} \right) = F^{1\prime}\left( \underline{u} \right)$, which is to say that $\underline{u} = u^\star$.

\begin{proof}[Sketch proof]
	A mechanism $(x,X)$ with $0 < x < u^0$ may be perturbed near an arbitrary period $t \in \R_+$ by adding $\eps$ to $x$ on $[t,t+\delta)$, where $\eps \neq 0$ and $\delta>0$ are small.
	This changes $X_s = r \int_s^\infty e^{-r(s'-s)} x_{s'} \dd s'$ for $s \leq t$ by $r e^{-r(t-s)} \delta \eps + \oo(\delta\eps)$, so changes the principal's payoff $\Pi_G(x,X)$ by
	\begin{equation*}
		r e^{-rt} F^{0\prime}(x_t) \delta \eps [1-G(t)]
		+ \int_{[0,t]} e^{-rs} F^{1\prime}(X_s) \left( r e^{-r(t-s)} \delta \eps \right) G(\dd s)
		+ \oo(\delta\eps) .
	\end{equation*}
	If $(x,X)$ is optimal, then it cannot be improved by such perturbations:
	\begin{equation*}
		F^{0\prime}(x_t) [1-G(t)]
		+ \int_{[0,t]} F^{1\prime}(X_s) G(\dd s) = 0 .
		\label{eq:euler_true}
		\tag{$\mathcal{E}_t$}
	\end{equation*}
	Letting $t \to \infty$ yields $\E_G\left( F^{1\prime}(X_\tau) \right) = 0$. Substituting this equality into \eqref{eq:euler_true} and dividing by $1-G(t)>0$ yields $F^{0\prime}(x_t) = \E_G\left( F^{1\prime}(X_\tau) \middle| \tau > t \right)$.
\end{proof}

To understand the Euler equation, differentiate it and rearrange to obtain
\begin{equation*}
	\dot{x}_t
	= - \underbrace{ \left( \frac{G'(t)}{1-G(t)} \right) }_{\text{hazard rate}}
	\frac{ F^{0\prime}(x_t) - F^{1\prime}(X_t) }{
	\underbrace{ - F^{0\prime\prime}(x_t) }_{\text{curvature}}
	} .%
		\footnote{This expression is valid
		under the additional assumptions that $G$ admits a continuous density
		and that $F^0$ possesses a continuous and strictly negative second derivative.}
\end{equation*}
Thus the agent's pre-disclosure utility declines in proportion to the hazard rate,
and in inverse proportion to the local curvature of the old frontier $F^0$.
As the latter would suggest, $x$ jumps over any affine segments ($F^{0\prime\prime} = 0$ and `$\dot{x} = \infty$'),
and pauses at kinks (`$F^{0\prime\prime} = -\infty$' and $\dot{x} = 0$).

Without the interiority ($u^\star > 0$) and differentiability hypotheses,
a superdifferential Euler equation characterises the optimal path.
We prove in \cref{app:opt_transition} that this equation is necessary for optimality, whence \Cref{proposition:opt_transition} follows,
and furthermore show that it is sufficient.

As for comparative statics, we show in \cref{suppl:mcs} that as the breakthrough distribution $G$ becomes later in the sense of monotone likelihood ratio,
the disclosure reward $X$ increases in every period.
(The pre-disclosure flow $x$ need not increase pointwise.)
It follows in particular that the agent's ex-ante payoff $X_0$ improves.

Although our focus is on general properties, there are special cases in which the Euler equation may be solved in closed form:

\begin{example}
	\label{example:exp_quad}
	Let the breakthrough arrive at constant rate $\lambda>0$, so that $G(t) = 1 - e^{-\lambda t}$ for every $t \in \R_+$. Fix $u^1 < u^0$ in $(0,\infty)$, and assume that
	\begin{equation*}
		F^j(u) \coloneqq a^j \left( u^j - \tfrac{1}{2} u \right) u + b^j
		\quad \text{for each $j \in \{0,1\}$ and every $u \in \left[ 0, u^0 \right]$,}
	\end{equation*}
	where $0 < a^0 < a^1 > a^0 u^0 / u^1$, and $b^1-b^0$ is large enough that $F^1 \geq F^0$.
	Solving the Euler equation yields the optimal mechanism $x$ given by
	\begin{equation*}
		x_t \coloneqq \left( u^0 - u^\star \right) e^{-\lambda kt} + u^\star
		\quad \text{for each $t \in \R_+$,}
	\end{equation*}
	where $u^\star = \displaystyle\frac{ a^1 u^1 - a^0 u^0 }{ a^1 - a^0 }$ and
	$k \coloneqq \displaystyle \frac{1 + r/\lambda}{2} \textstyle \left( \sqrt{ \frac{r/\lambda}{\left( \frac{ 1 + r/\lambda }{2} \right)^2} \left( \frac{a^1}{a^0} - 1 \right) + 1 } - 1 \right)$.
	In the special case $r=\lambda$, this simplifies to $k = \sqrt{a^1/a^0} - 1$.
\end{example}

\section{Application to unemployment insurance}
\label{sec:pf}

If unemployment benefits are generous but time-limited, then a worker who receives a job offer before her benefits run out may have an incentive to delay starting her new job, for example by arranging a deferred start date (or by simply waiting before accepting the offer).
Empirically, such strategic delay appears to be widespread.%
	\footnote{See \textcite{BooneVanours2012,DellavignaEtal2017,KyyraPesolaVerho2019}.}

In this section, we study the design of unemployment insurance (`UI') schemes when workers can exercise such strategic delay. We focus in particular on the merits of \emph{deadline benefit schemes,} in which the short-term unemployed receive a generous benefit, while those remaining unemployed past a deadline see their benefit reduced to a much lower level. Such schemes are used in many countries, including Germany, France and Sweden. We also study the optimal choice of deadline.

\paragraph{Related literature.}
The literature on optimal unemployment insurance has two main strands.
The first concerns the moral-hazard problem of incentivising job-search effort \parencite{ShavellWeiss1979,HopenhaynNicolini1997}.
We contribute to the second strand, which studies the adverse-selection problem arising from privately observed job offers \parencite{AtkesonLucas1995}.%
	\footnote{See also \textcite{ThomasWorrall1990,AtkesonLucas1992,HansenImrohoroglu1992,ShimerWerning2008}.}
(As shown in §\ref{sec:mh} below, our conclusions in this section would not change if we added moral hazard to the model.)
Within this second strand, our contribution is to characterise optimal UI under the assumption that workers can delay starting a new job, rather than having to start right away.

\subsection{Model}
\label{sec:pf:model}

A worker (agent) is unemployed.
At a random time $\tau \sim G$, she receives a job offer. If she accepts, then she chooses when to start.
The worker's ability to delay her start date is the distinguishing feature of our otherwise-standard model.
The state observes in real time whether the worker is employed, but cannot observe whether she has received a job offer. All jobs are permanent and pay the same wage, denoted $w>0$.

The worker's utility is $u = \phi(C) - \kappa(L)$,
where $C \geq 0$ is her consumption
and $L \geq 0$ her labour supply.
We assume that $\phi,\kappa : [0,\infty) \to [0,\infty)$ are respectively strictly concave and strictly convex, that they are differentiable on $(0,\infty)$ with strictly positive derivatives that satisfy
\begin{equation*}
	\lim_{C \to \infty} \phi'(C) = 0 , \quad
	\lim_{C \to 0} \phi'(C) = \infty
	\quad \text{and} \quad
	\lim_{L \to 0} \kappa'(L) = 0 ,
\end{equation*}
and that they are continuous with $\phi(0) = \kappa(0) = 0$ and $\lim_{C \to \infty} \phi(C) = \infty$.
We interpret $C=0$ as the lowest socially acceptable standard of living.
If the worker is unemployed, then $L=0$.

The state controls unemployment benefits and income taxes.
Following the literature, we impose no constraints on policy:%
	\footnote{This has been the standard approach since \textcite{HopenhaynNicolini1997}.}
income taxation after re-employment can be non-linear, for example, and can depend on the length of the preceding unemployment spell.
These policy instruments can implement any allocation $(C,L)$ which the worker prefers to autarky.%
	\footnote{An unemployed worker's consumption is simply her benefit.
	To get an employed worker to choose a bundle $(C,L)$ satisfying $u \coloneqq \phi(C) - \kappa(L) \geq 0$,
	use the income tax schedule $\theta(Y) = \min\left\{ Y, m Y + b \right\}$,
	with $m,b \in \R$ chosen so that
	the worker's income $L' \mapsto wL' - \theta(wL')$
	is tangent at $L$ to her indifference curve $L' \mapsto \phi^{-1}\left( \kappa(L') + u \right)$.}
We may therefore model the state as directly choosing $(C,L)$, subject to $u = \phi(C) - \kappa(L) \geq 0$.

The state's objective is social welfare $v = u + \lambda \times ( wL - C )$, where $u$ is the worker's welfare, $wL - C$ is net tax revenue, and $\lambda > 0$ is the shadow value of public funds.
The utility possibility frontiers for unemployed and employed workers are thus
\begin{align*}
	F^0(u)
	&\coloneqq \max_{C \geq 0} \left\{ u + \lambda (-C) : \phi(C) = u \right\}
	\\
	\text{and} \quad
	F^1(u)
	&\coloneqq \max_{C,L \geq 0} \left\{
	u + \lambda (w L - C) : \phi(C) - \kappa(L) = u
	\right\} ,
\end{align*}
respectively.
These frontiers satisfy our model assumptions (§\ref{sec:model}):

\begin{lemma}
	\label{lemma:UI_assns}
	In the application to unemployment insurance, the frontiers $F^0,F^1$
	are strictly concave and continuous,
	with unique peaks $u^0,u^1$ that satisfy $u^1 < u^0$.
	The gap $F^1 - F^0$ is strictly decreasing, so that $u^\star = 0$.
\end{lemma}

The conflict of interest $u^1 < u^0$
arises because the social first-best
requires employed workers to supply labour ($L>0$), which they dislike, without compensating them with extra consumption (first-best consumption is $C^0 \coloneqq (\phi')^{-1}(\lambda)$ regardless of employment status).
This is an instance of the fact, well-known in public finance since \textcite{Mirrlees1971,Mirrlees1974},%
	\footnote{See the third section of \textcite{Mirrlees1974}, as well as p.~201 of \textcite{Mirrlees1971}.}
that welfare-maximisation (absent incentive constraints) does not `reward merit': on the contrary, it dictates efficient production, meaning that more productive workers work harder.
The proof of \Cref{lemma:UI_assns} is elementary but tedious, so we omit it.%
	\footnote{It may be found in \textcite{omit}.}


We shall use the term `unemployment insurance (UI) scheme' for a mechanism.
By \Cref{proposition:indiff} (\cpageref{proposition:indiff}), undominated schemes keep the worker only just willing promptly to start a job, so have the form $(x,X)$.
Implicit in a UI scheme $(x,X)$
are the benefit $B_t$ paid to the time-$t$ unemployed (given by $x_t = \phi\left( B_t \right)$)
and the labour supply $L_t$ and tax bill $\theta_t = w L_t - C_t$ of a worker who started working at $t$
(which satisfy $X_t = \phi\left( w L_t - \theta_t \right) - \kappa( L_t )$).

\subsection{Optimal unemployment insurance}
\label{sec:pf:results}

Optimal UI schemes are described by \Cref{theorem:opt} (\cpageref{theorem:opt}):
unemployment benefits $B_t = \phi^{-1}(x_t)$ decrease over time,
from $C^0 \coloneqq \phi^{-1}\left( u^0 \right)$
toward $0 = \phi^{-1}\left( u^\star \right)$.
Thus workers enjoy socially optimal consumption at the beginning of an unemployment spell,
but see their benefits reduced over time,
with the long-term unemployed provided only with society's lowest acceptable standard of living (`zero consumption').

Employed workers are rewarded with a higher continuation utility $X_t$ the earlier they start a job.
This involves a mix of lower labour supply and more generous tax treatment of earnings (yielding higher consumption).

A \emph{deadline UI scheme} is one in which a generous benefit of $C^0$ is paid to the short-term unemployed,
while those remaining unemployed beyond a deadline receive a low benefit just sufficient to finance the minimum standard of living (`zero consumption').
Such schemes are widespread in practice, used in e.g. Germany, France and Sweden.

Our results speak to the desirability of such deadline schemes.
\Cref{theorem:deadline} (\cpageref{theorem:deadline}) implies that a deadline scheme is approximately optimal if $F^0$ is close to affine, a condition which is satisfied if the worker's consumption utility $\phi$ has limited curvature
or if the social value $\lambda$ of tax revenue is moderate.
Conversely, if neither assumption is close to being satisfied, then our results predict substantial welfare gains from more gradual tapering, as in Italy.

Given the prevalence of deadline schemes (whatever their merits),
the choice of deadline is an important policy problem.
Our analysis highlights labour-market prospects as a key consideration:
a worker with worse chances (a later job-finding distribution $G$, in the sense of first-order stochastic dominance) should be set a later deadline.%
	\footnote{In particular, the optimal deadline described by \Cref{proposition:opt_deadline} (\cpageref{proposition:opt_deadline}) is later when $G$ is, as noted at the end of §\ref{sec:deadline:opt_deadline} and proved in \cref{app:opt_deadline_pf} (\cpageref{proposition:opt_deadline_comp}).}
Two implications are that older workers ought to face later deadlines
and that extensions should be granted during recessions.
These recommendations are broadly followed in Germany and France, where workers older than about 50 face more lenient deadlines,
and all workers' deadlines were prolonged during the 2020 recession.

\section{Extension: moral hazard}
\label{sec:mh}

Many applications feature an element of moral hazard, whereby the agent can hasten the breakthrough by unobservably exerting effort at a cost.
An unemployed worker's search effort influences her job-finding rate, for example.
In this section, we show that our results remain valid when there is moral hazard.

We consider the simplest moral-hazard environment:
in each period $t \in \R_+$, the agent either exerts effort ($a_t=1$) or doesn't ($a_t=0$).
The agent incurs a flow cost of $c a_t$, where $c>0$,
and the breakthrough arrives at rate $\lambda a_t$, where $\lambda>0$.
The principal observes neither effort nor breakthroughs.
We rule out the uninteresting case in which the principal prefers not to incentivise effort
by assuming that the breakthrough is sufficiently valuable:%
	\footnote{Precisely, the (mild) assumption is that the principal finds it worthwhile to incentivise a \emph{single} period's effort, by rewarding disclosure enough that $\lambda \left( X^1 - X^0 \right) \geq r c$.}
\begin{equation*}
	F^1\left(u + rc/\lambda \right) > F^0(u)
	\quad\text{for every $u \in \left[0,u^0\right]$.}
	\tag{$\triangle$}
	\label{eq:moralhazard_assumption}
\end{equation*}

An \emph{effort schedule} is a measurable map $a : \R_+ \to \{0,1\}$; the interpretation of $a_t=1$ ($=0$) is that in case of no breakthrough by period $t$, the agent exerts (no) effort in that period.
Each effort schedule $a : \R_+ \to \{0,1\}$ induces
a breakthrough distribution, namely $G^a$ given by $G^a(t) \coloneqq 1 - \exp\bigl( -\lambda \int_0^t a \bigr)$.%
	\footnote{A detail: the distribution $G^a$ may be improper (specifically, if $\int_0^\infty a < \infty$).}

The definition of a mechanism is unchanged.
A \emph{triplet} $\left(x^0,X^1,a\right)$ comprises a mechanism $\left(x^0,X^1\right)$ and an effort schedule $a : \R_+ \to \{0,1\}$.
A triplet $\left(x^0,X^1,a\right)$ 
is called \emph{incentive-compatible (IC)} exactly if $\left(x^0,X^1\right)$ is IC,
and \emph{obedient} iff the agent is willing to exert effort according to $a$ (i.e. among all effort schedules, $a$ yields the highest expected payoff).
By a revelation principle, we may restrict attention to IC and obedient triplets.
The principal's payoff from an IC and obedient triplet $\left(x^0,X^1,a\right)$ is $\Pi_{G^a}\left(x^0,X^1\right)$.

\emph{Continual effort} is the effort schedule $a \equiv 1$.
Call a mechanism $\left(x^0,X^1\right)$ \emph{obedient} iff the triplet $\left(x^0,X^1,1\right)$ is obedient.
Continual effort is optimal:

\begin{lemma}
	\label{lemma:max_effort_obedient}
	There exists an IC and obedient mechanism $\bigl(x^{0\dag},X^{1\dag}\bigr)$ such that $\Pi_{G^1}\bigl(x^{0\dag},X^{1\dag}\bigr) \geq \Pi_{G^a}\left(x^0,X^1\right)$ for any IC and obedient triplet $\left(x^0,X^1,a\right)$.
\end{lemma}

In light of \Cref{lemma:max_effort_obedient},
the principal may restrict attention to IC and obedient mechanisms without loss of optimality,
and we may assume that when faced with such a mechanism, the agent exerts continual effort.
Call an IC and obedient mechanism $\left(x^0,X^1\right)$ \emph{o-undominated} iff it is not dominated by any IC and obedient mechanism.

\begin{namedthm}[\Cref*{proposition:indiff}$\boldsymbol{'}$.]
	\label{proposition:indiff_moralhazard}
	Any o-undominated IC and obedient mechanism $\left(x^0,X^1\right)$ satisfies $x^0 \leq u^0$ and $X^1 = X^0 + rc/\lambda$.
\end{namedthm}

This is exactly \Cref{lemma:lequ0} and \Cref{proposition:indiff},
except that the disclosure reward $X^1$ differs from $X^0$ by a constant instead of by zero (in order to incentivise effort).
That is the \emph{only} difference that moral hazard makes:
o-undominated and optimal mechanisms $\left(x^0,X^1\right)$
still have pre-disclosure flow $x^0$ as described by \Cref{theorem:deadline,theorem:opt} (verbatim).

To see why \Cref{theorem:deadline,theorem:opt} remain true,
consider the `translated' model in which the frontiers are $F^0$ and $u \mapsto F^1(u + rc/\lambda)$.
These frontiers satisfy our model assumptions
(in particular, the new frontier exceeds the old by assumption \eqref{eq:moralhazard_assumption}).
\hyperref[proposition:indiff_moralhazard]{\Cref*{proposition:indiff}$'$} provides that \Cref{lemma:lequ0} and \Cref{proposition:indiff} hold:
all mechanisms $\left(x^0,X^1\right)$ that the principal considers satisfy $x^0 \leq u^0$ and $X^1 = X^0$.
And the principal's payoff from such a mechanism is $\Pi_{G^1}\left(x^0,X^0\right)$.
Thus \Cref{theorem:deadline,theorem:opt} are applicable.

The proofs of \Cref{lemma:max_effort_obedient} and {\hyperref[proposition:indiff_moralhazard]{\Cref*{proposition:indiff}$'$}} are in \cref{suppl:moralhazard}.



\begin{appendices}

\renewcommand*{\thesubsection}{\Alph{subsection}}

\crefalias{section}{appsec}
\crefalias{subsection}{appsec}
\crefalias{subsubsection}{appsec}
\section*{Appendices}
\label{app}
\addcontentsline{toc}{section}{Appendices}

\subsection{Background and notation}
\label{app:background}

The Lebesgue integral is used throughout.
In particular, for $s<t$ in $\R_+$ and a function $\phi : \R_+ \to [-\infty,\infty]$,
$\int_s^t \phi$ denotes the Lebesgue integral $\int_{(s,t)} \phi \dd \lambda$, where $\lambda$ is the Lebesgue measure.

We rely on various facts about concave functions \parencite[see][esp. part~V]{Rockafellar1970}.
For $j \in \{0,1\}$, recall that $F^j : [0,\infty) \to [-\infty,\infty)$ is concave and upper semi-continuous.
Write $D^j \coloneqq \left\{ u \in [0,\infty) : F^j(u) > -\infty \right\}$ for its effective domain (a convex set).
We have $\left(0,u^0\right] \subseteq D^j$ by assumption.

The right- and left-hand derivatives of $F^j$ are denoted by $F^{j+}$ and $F^{j-}$, respectively.
The former (latter) is well-defined on $D^j \union \inf D^j$
(on $\left( D^j \union \sup D^j \right) \setminus \{0\}$), but may take infinite values on the boundary.
$F^{j+}$ is right-continuous, and $F^{j-}$ is left-continuous.
If the derivative $F^{j\prime}$ exists at $u \in \interior D^j$, then $F^{j\prime}(u) = F^{j+}(u) = F^{j-}(u)$, and $F^{j\prime}$ is continuous at $u$.

The directional derivatives $F^{j+},F^{j-}$ are decreasing, and satisfy
\begin{equation*}
	F^{j-}(u) \leq F^{j+}(u') \leq F^{j-}(u') \leq F^{j+}(u'')
	\quad \text{for any $u > u' > u''$ in $\interior D^j$.}
\end{equation*}
The first (last) inequality is strict iff $F^j$ is not affine on $[u',u]$ (on $[u'',u']$),
and the middle inequality is strict exactly if $F^j$ has a kink at $u'$.

A \emph{supergradient} of $F^j$ at $u \in \cl D^j$ is an $\eta \in [-\infty,\infty]$ such that
\begin{equation*}
	F^j(u') \leq F^j(u) + \eta (u'-u)
	\quad \text{for every $u' \neq u$ in $[0,\infty)$.}
\end{equation*}
(Note that $\infty$ and $-\infty$ can be supergradients.)
$F^j$ admits at least one supergradient at every $u \in \cl D^j$.
For $u \in \interior D^j$,
$\eta \in [-\infty,\infty]$
is a supergradient of $F^j$ at $u$ exactly if $F^{j+}(u) \leq \eta \leq F^{j-}(u)$,
while for $u = \inf D^j$ ($u = \sup D^j$) the former (latter) inequality by itself is necessary and sufficient.

\subsection{Proof of \texorpdfstring{\Cref{proposition:indiff}}{Proposition~\ref{proposition:indiff}} (p.~\pageref{proposition:indiff})}
\label{app:pf_theorem_indiff}

We shall follow the sketch proof, but with significant elaborations aimed at overcoming the two technical hurdles discussed at the end of §\ref{sec:indiff}.

For any mechanism $\left( x^0, X^1 \right)$,
let $h : \R_+ \to [-\infty,\infty]$ be given by $h(t) \coloneqq e^{-rt}(X^1_t - X^0_t)$ for each $t \in \R_+$.%
	\footnote{In case $X^1_t = X^0_t = \infty$, we let $h(t) \coloneqq 0$ by convention.}
\Cref{proposition:indiff} asserts precisely that undominated IC mechanisms have $h$ identically equal to zero.

\begin{observation}
	\label{observation:ic_charac}
	A mechanism $\left( x^0, X^1 \right)$ is incentive-compatible exactly if
	$h$ is (a)~decreasing and (b)~non-negative.
\end{observation}

\begin{proof}
	Part~\ref{item:ic:delay} (part~\ref{item:ic:never}) of the definition of incentive-compatibility on \cpageref{definition:ic} requires precisely that
	$h$
	be decreasing (non-negative).
\end{proof}

\begin{namedthm}[Continuity lemma.]
	\label{lemma:indiff_disc}
	Any undominated IC mechanism has $h$ continuous.
\end{namedthm}

\begin{proof}
	We prove the contrapositive. Fix an IC mechanism $\left( x^0, X^1 \right)$.

	Suppose that $h$ is discontinuous at some $t \in (0,\infty)$.
	Since $h$ is decreasing and $X^0$ is continuous, $\lim_{s \uparrow t} X^1_s$ and $\lim_{s \downarrow t} X^1_s$ exist and satisfy
	$\lim_{s \uparrow t} X^1_s \geq X^1_t \geq \lim_{s \downarrow t} X^1_s$,
	with one of the inequalities strict.
	We shall assume that
	\begin{equation*}
		\lim_{s \uparrow t} X^1_s = X^1_t > \lim_{s \downarrow t} X^1_s ,
	\end{equation*}
	omitting the similar arguments for the other two cases.
	If $\lim_{s \downarrow t} X^1_s < u^1$, then we may increase $X^1$ toward $u^1$ on a small interval $(t,t+\eps)$ while keeping $h$ decreasing.%
		\footnote{Choose an $\eps>0$ small enough that
		$X^1 + \eps < \min\{u^1,X^1_t\}$ on $(t,t+\eps)$.
		Let $X^{1\dag}_s \coloneqq X^1_s - (s-t) + \eps$ for $s \in (t,t+\eps)$
		and $X^{1\dag} \coloneqq X^1$ off $(t,t+\eps)$.
		Then $X^1 \leq X^{1\dag} \leq u^1$, with the first inequality strict on $(t,t+\eps)$.
		We have $h^\dag \geq h \geq 0$,
		and $h^\dag$ is clearly decreasing on $[0,t]$ and on $(t,\infty)$.
		At $t$, we have $h^\dag(t)-\lim_{s \downarrow t} h^\dag(s) = e^{-rt}(X^1_t-\lim_{s \downarrow t} X^1_s -\eps) \geq 0$.}
	If instead $\lim_{s \downarrow t} X^1_s \geq u^1$,
	then $\lim_{s \uparrow t} X^1_s = X^1_t > u^1$,
	so that we may decrease $X^1$ toward $u^1$ on a small interval $(t-\eps,t]$ while keeping $h$ decreasing.%
		\footnote{Choose an $\eps \in (0,1/r)$ small enough that
		$X^1 - \eps > \lim_{s \downarrow t} X^1_s$
		and $h > \eps$
		on $(t-\eps,t]$.
		Let $X^{1\dag}_s \coloneqq X^1_s + t-s-\eps$ for $s \in (t-\eps,t]$
		and $X^{1\dag} \coloneqq X^1$ off $(t-\eps,t]$.
		Then $u^1 \leq X^{1\dag} \leq X^1$, with the second inequality strict on $(t-\eps,t]$.
		Clearly $h^\dag$ is non-negative,
		and is decreasing on $[0,t-\eps]$ and on $(t,\infty)$.
		It is decreasing on $[t-\eps,t]$ since $h^\dag(s) - h(s) = e^{-rs}(t-s-\eps)$ is (by our choice of $\eps<1/r$).
		And at $t$, $h^\dag(t)-\lim_{s \downarrow t} h^\dag(s) = e^{-rt}( X^1_t - \eps - \lim_{s \downarrow t} X^1_s ) \geq 0$.}
	In either case,
	IC is preserved,
	and the principal's payoff $\Pi_G$ is (strictly) increased under any (full-support) distribution $G$.

	Suppose instead that $h$ is discontinuous at $t=0$;
	then $X^1_0 > \lim_{s \downarrow 0} X^1_s$ by IC and the continuity of $X^0$.
	The case $\lim_{s \downarrow 0} X^1_s < u^1$ may be dealt with as above.
	If $\lim_{s \downarrow 0} X^1_s \geq u^1$,
	then lowering $X^1_0$ toward $\lim_{s \downarrow 0} X^1_s$ preserves IC and (strictly) increases $\Pi_G$ for any distribution $G$ (with $G(0)>0$).
\end{proof}

\begin{proof}[Proof of \Cref{proposition:indiff}]
	Let $\left( x^0, X^1 \right)$ be an IC mechanism,
	so that $h$ is non-negative and decreasing,
	and suppose that $h$ is not identically zero.
	By the \hyperref[lemma:indiff_disc]{continuity lemma}, we may assume that $h$ (and thus $X^1$) is continuous.

	We consider three cases.
	(The first two concern slack `\ref{eq:delay_IC}':
	Case 1 [Case 2] corresponds to the sketch proof's
	\hyperlink{item:indiff_discrete_pf:case2}{case~(ii)}
	[cases~\hyperlink{item:indiff_discrete_pf:case1}{(i)} and \hyperlink{item:indiff_discrete_pf:case3}{(iii)}].
	Case 3 is where `\ref{eq:delay_IC}' binds, but `\ref{eq:nondisc_IC}' is slack.)
	In each case, we shall construct an incentive-compatible mechanism $\bigl( x^{0\dag}, X^{1\dag} \bigr)$
	such that
	\begin{equation*}
		\Pi_G\left( x^{0\dag}, X^{1\dag} \right)
		\geq \mathrel{(>)} \Pi_G\left( x^0, X^1 \right)
		\quad \text{for every (full-support) $G$.}
		\label{eq:strong_dominance}
		\tag{D}
	\end{equation*}
	Define
	$A \coloneqq \left\{
	t \in \R_+ :
	\text{$h$ is differentiable at $t$ and $h'(t) < 0$}
	\right\}$.

	\emph{Case 1: $\left\{ t \in A: x^0_t < u^0\right\}$
	is non-null.}
	Since $h>0$ on $A$,%
		\footnote{Since $h \geq 0$,
		$h(t)=0$ implies $\liminf_{t' \downarrow t} [h(t')-h(t)] / (t'-t) \geq 0$
		and thus $t \notin A$.}
	there is an $\eps>0$ for which the set
	\begin{equation*}
		A_\eps \coloneqq \left\{ t \in A :
		\text{$x^0_t + \eps < u^0$,
		$h(t) \geq \eps$ and
		$h'(t) + r\eps \leq 0$}
		\right\}
	\end{equation*}
	is non-null.%
		\footnote{$A_0 = \Union_{n \in \N} A_{1/n}$ is non-null,
		so continuity of measures (with $\lambda$ denoting the Lebesgue measure) yields
		$0
		< \lambda(A_0)
		= \lim_{n \to \infty} \lambda( A_{1/n} )$,
		whence $\lambda( A_{1/n} ) > 0$ for some $n \in \N$.}
	Define $x^{0\dag} \coloneqq x^0 + \eps \1_{A_\eps}$,
	and consider the mechanism $\bigl( x^{0\dag}, X^1 \bigr)$.
	Clearly $x^0 \leq x^{0\dag} \leq u^0$, and $x^{0\dag} \neq x^0$ on the non-null set $A_\eps$,
	so that \eqref{eq:strong_dominance} holds
	by the strict monotonicity of $F^0$ on $\left[ 0, u^0 \right]$.
	$h^\dag$ is decreasing since for any $t<t'$ in $\R_+$,
	\begin{align*}
		h^\dag(t') - h^\dag(t)
		&= h(t') - h(t)
		+ r\eps \int_t^{t'} e^{-rs} \1_{A_\eps}(s) \dd s
		\\
		&\leq \int_t^{t'} h' \1_{A_\eps}
		+ r\eps \int_t^{t'}e^{-rs} \1_{A_\eps}(s) \dd s
		\leq 0 ,
	\end{align*}
	where the first inequality holds since $h$ is decreasing,%
		\footnote{Recall the Lebesgue decomposition $h = h_a + h_s$
		where $h_a$ is decreasing and absolutely continuous
		and $h_s$ is decreasing with $h_s'=0$ a.e.
		\parencite[e.g.][p. 150]{SteinShakarchi2005}.}
	and the second holds by definition of $A_\eps$.
	As for non-negativity,
	we have $h^\dag = h \geq 0$ on $\left(\sup A_\eps,\infty\right)$,
	while $h^\dag \geq 0$ on $\left[0,\sup A_\eps\right)$
	since $h^\dag$ is decreasing
	and $h^\dag \geq h - \eps \geq 0$ on $A_\eps$ by definition of the latter.
	Thus $\bigl( x^{0\dag}, X^1 \bigr)$ is incentive-compatible.

	\emph{Case 2: There are $t' < t''$ in $\R_+$ such that $h(t') > h(t'')$ and $X^1 \ne u^1$ on $[t',t'']$.}
	Since $X^1$ is continuous,
	we have either $X^1 > u^1$ on $[t',t'']$
	or $X^1 < u^1$ on $[t',t'']$.
	We shall assume the former,
	omitting the similar argument for the latter case.
	Because $s \mapsto e^{rs}h(t'') + X^0_s$ is continuous and takes the value $X^1_{t''} > u^1$ at $s=t''$,
	\begin{equation*}
		t^\star
		\coloneqq \inf \left\{ t \in \left[ t', t'' \right] :
		\text{$e^{rs}h(t'') + X^0_s \geq u^1$
		for all $s \in [t,t'']$}
		\right\}
	\end{equation*}
	is well-defined and strictly smaller than $t''$.
	Define
	\begin{equation*}
		X^{1\dag}_t
		\coloneqq
		\begin{cases}
			e^{rt}h(t'') + X^0_t
			& \text{for $t \in \left[ t^\star, t'' \right)$}
			\\
			X^1_t
			& \text{for $t \notin \left[ t^\star, t'' \right)$,}
		\end{cases}
	\end{equation*}
	and consider the mechanism $\bigl( x^0, X^{1\dag} \bigr)$.
	This mechanism is IC since
	$h^\dag = h + [ h(t'')-h ] \1_{\left[ t^\star, t'' \right)}$
	is clearly decreasing and non-negative.

	It remains to show that $\bigl( x^0, X^{1\dag} \bigr)$ satisfies \eqref{eq:strong_dominance}.
	Since $X^1$ and $X^{1\dag}$ differ only on $\left[ t^\star, t'' \right)$
	and $F^1$ is strictly decreasing on $\left[ u^1, \infty \right)$,
	it suffices to prove that
	\begin{equation*}
		u^1 \leq X^{1\dag}_t \leq \mathrel{(<)} X^1_t
		\quad \text{for every (some) $t \in \left[ t^\star, t'' \right)$.\footnotemark}
	\end{equation*}
	\footnotetext{It is enough for the inequality to be strict at a single time $t \in [ t^\star, t'' )$, since it then holds strictly on a proper interval
	by the continuity of $X^1$ and $X^{1\dag}$ on $[ t^\star, t'' )$.}
	The first inequality holds by definition of $t^\star$.
	For the second, observe that
	\begin{equation*}
		X^{1\dag}_t - X^1_t
		= e^{rt}\left[h^\dag(t)-h(t)\right]
		= e^{rt} \left[ h(t'')-h(t) \right]
		\leq 0
		\quad \text{for $t \in \left[ t^\star, t'' \right)$}
	\end{equation*}
	since $h$ is decreasing.
	We claim that the inequality is strict at $t=t^\star$.
	If $t^\star = t'$, then this is true because $h(t') > h(t'')$.
	And if not, then $t^\star \in (t',t'')$,
	in which case $X^{1\dag}_{t^\star} = u^1 < X^1_{t^\star}$
	by continuity of $X^0$ and $X^1 > u^1$.

	\emph{Case 3: neither Case 1 nor Case 2.}
	Since $X^1$ is continuous,
	every $t \in \R_+$ belongs either to a maximal open interval on which $X^1 \neq u^1$
	or else to a maximal closed interval on which $X^1 = u^1$.
	$h$ is increasing on any interval of the former kind since we are not in Case 2.
	We shall show that $h$ is also increasing on each interval of the latter kind;
	then since $h$ is continuous, it is increasing and thus constant.

	So fix an interval $I$ of the latter kind.
	Since $h$ is decreasing, its derivative $h'(t) = re^{-rt}\left(x^0_t-u^1\right)$
	exists a.e. on $I$.
	As we are not in Case 1, we have for a.e. $t \in I$
	that either $h'(t)=0$ or $x^0_t = u^0$,
	and in the latter case $h'(t) = re^{-rt}(u^0-u^1) > 0$.
	Assuming wlog that $x^0 \leq u^0$,%
		\footnote{Otherwise the IC mechanism $(\min\{x^0,u^0\},X^1)$ would satisfy \eqref{eq:strong_dominance}.}
	the expression for $h'$ implies that $h$ is $ru^0$-Lipschitz on $I$.
	Thus $h$ is increasing on $I$,
	as desired.

	Since (by hypothesis) $h$ is not identically zero,
	it is constant at some $k>0$,
	so that $X^1_t = X^0_t + e^{rt} k$ for every $t \in \R_+$.
	Thus $X^{1\dag} \coloneqq \min\left\{ X^1, X^0 + u^1 \right\}$
	is strictly smaller than $X^1$ after some time $T>0$,
	so that $\bigl( x^0, X^{1\dag} \bigr)$ satisfies \eqref{eq:strong_dominance}.
	And it is incentive-compatible.%
		\footnote{We have $h^\dag(t) = e^{-rt} u^1 \in (0,h^\dag(T))$ for $t > T$, and this expression is decreasing.}
\end{proof}

\subsection{Proof of \texorpdfstring{\Cref{theorem:deadline}}{Theorem~\ref{theorem:deadline}} (p.~\pageref{theorem:deadline})}
\label{app:pf_theorem_deadline}

Fix a non-deadline mechanism $(x,X)$ with $x \leq u^0$ a.e.;%
	\footnote{IC mechanisms not of this form are dominated, by \Cref{lemma:lequ0} and \Cref{proposition:indiff}.}
we will show that it is dominated by the deadline mechanism $\bigl( x^\dag, X^\dag \bigr)$ whose deadline $T$ satisfies
\begin{equation*}
	\left( 1 - e^{-rT} \right) u^0 + e^{-rT} u^\star
	\equiv X^\dag_0
	= X_0 \join u^\star ,
\end{equation*}
where `$\join$' denotes the pointwise maximum.

\begin{namedthm}[Claim.]
	\label{claim:dl_leq_origjoin}
	$X^\dag \leq X \join u^\star$.
\end{namedthm}

\begin{proof}[Proof]%
	\renewcommand{\qedsymbol}{$\square$}
	For $t \geq T$, we have $X^\dag = u^\star \leq X \join u^\star$.
	For $t < T$, suppose first that $X^\dag_0 = X_0$;
	then since $x^\dag = u^0 \geq x$ on $[0,t] \subseteq [0,T]$, we have
	\begin{align*}
		e^{-rt} X^\dag_t
		&= X^\dag_0 - r \int_0^t e^{-rs} x_s^\dag \dd s
		\\
		&\leq X_0 - r \int_0^t e^{-rs} x_s \dd s
		= e^{-rt} X_t
		\leq e^{-rt} \left( X_t \join u^\star \right) .
	\end{align*}
	If instead $X^\dag_0 = u^\star$,
	then the fact that $x^\dag \geq u^\star$ yields
	\begin{align*}
		e^{-rt} X^\dag_t
		&= X^\dag_0 - r \int_0^t e^{-rs} x_s^\dag \dd s
		\\
		&\leq u^\star - r \int_0^t e^{-rs} u^\star \dd s
		= e^{-rt} u^\star
		\leq e^{-rt} \left( X_t \join u^\star \right) .
		\qedhere
	\end{align*}
\end{proof}%
\renewcommand{\qedsymbol}{$\blacksquare$}

The concave function $F^1 - F^0$ is uniquely maximised at $u^\star$,
so is strictly increasing on $\left[ 0, u^\star \right]$
and strictly decreasing on $\left[ u^\star, u^0 \right]$.
Since $u^\star \leq X^\dag \leq X \join u^\star$ by the \hyperref[claim:dl_leq_origjoin]{claim}, it follows that
\begin{equation}
	\left[ F^1 - F^0 \right]\left( X^\dag \right)
	\geq \left[ F^1 - F^0 \right]\left( X \join u^\star \right) .
	\label{eq:deadline_pf_ineq1}
\end{equation}
Since $X \join u^\star \geq X$,
and the two differ only when both are in $\left[ 0, u^\star \right]$, we have
\begin{equation}
	\left[ F^1 - F^0 \right]\left( X \join u^\star \right)
	\geq \left[ F^1 - F^0 \right]\left( X \right) ,
	\label{eq:deadline_pf_ineq2}
\end{equation}
which chained together with the preceding inequality yields
\begin{equation}
	\left[ F^1 - F^0 \right]\left( X^\dag \right)
	\geq \left[ F^1 - F^0 \right]\left( X \right) .
	\label{eq:deadline_pf_ineq3}
\end{equation}
The facts that $X^\dag_0 = X_0 \join u^\star \geq X_0$
and that $F^0$ is increasing on $\left[ 0, u^0 \right]$
together imply
\begin{equation}
	F^0\left( X^\dag_0 \right)
	\geq F^0\left( X_0 \right) .
	\label{eq:deadline_pf_ineq4}
\end{equation}
Thus for any distribution $G$, using the expression for the principal's payoff derived in the sketch proof (\cpageref{eq:affine_payoff}), we have
\begin{align*}
	\Pi_G\left( x^\dag, X^\dag \right)
	&= F^0\left( X_0^\dag \right)
	+ \E_G\left(
	e^{-r\tau} \left[ F^1 - F^0 \right]\left( X^\dag_\tau \right)
	\right)
	&&
	\\
	&\geq F^0\left( X^\dag_0 \right)
	+ \E_G\left(
	e^{-r\tau} \left[ F^1 - F^0 \right]\left( X_\tau \right)
	\right)
	&& \text{by \eqref{eq:deadline_pf_ineq3}}
	\\
	&\geq F^0\left( X_0 \right)
	+ \E_G\left(
	e^{-r\tau} \left[ F^1 - F^0 \right]\left( X_\tau \right)
	\right)
	&& \text{by \eqref{eq:deadline_pf_ineq4}}
	\\
	&= \Pi_G(x,X) .
	&&
\end{align*}

It remains show that $\bigl( x^\dag, X^\dag \bigr)$ delivers a \emph{strict} improvement for some distribution $G$.
We shall accomplish this by showing that the inequality \eqref{eq:deadline_pf_ineq3} holds strictly on a non-null set of times, so that the first inequality in the above display is strict for any distribution $G$ with full support.
Since $X^\dag \leq X \join u^\star$ by the \hyperref[claim:dl_leq_origjoin]{claim} and $X,X^\dag$ are continuous,
there are two cases: either
\hyperdest{item:deadline_pf_casea}(a)~$X^\dag < X \join u^\star$ on a non-null set of times, or
\hyperdest{item:deadline_pf_caseb}(b)~$X^\dag = X \join u^\star$.

\emph{Case \hyperlink{item:deadline_pf_casea}{(a)}:
$X^\dag < X \join u^\star$ on a non-null set $\mathcal{T}$.}
In this case, the inequality \eqref{eq:deadline_pf_ineq1} holds strictly on $\mathcal{T}$,
and thus so does \eqref{eq:deadline_pf_ineq3}.

\emph{Case \hyperlink{item:deadline_pf_caseb}{(b)}:
$X^\dag = X \join u^\star$.}
Since the original mechanism $(x,X)$ is not a deadline mechanism,
there must be a non-null set of times on which $x \neq x^\dag$,
and thus $X \neq X^\dag = X \join u^\star$ on some non-null set $\mathcal{T}$,
so that $X < X \join u^\star$ on $\mathcal{T}$.
Then \eqref{eq:deadline_pf_ineq2} is strict on $\mathcal{T}$,
and thus so is \eqref{eq:deadline_pf_ineq3}.
\qed

\subsection{Proof of \texorpdfstring{\Cref{proposition:dl_charac}}{Proposition~\ref{proposition:dl_charac}} (p.~\pageref{proposition:dl_charac})}
\label{app:deadline_charac}

Write $\bigl( x^T, X^T \bigr)$ for the deadline mechanism with deadline $T$,
and $\pi_G(T)$ for its payoff under a distribution $G$.
By \Cref{theorem:deadline}, any undominated mechanism is a deadline mechanism.
We showed in the text (§\ref{sec:deadline:undom_deadline}, \cpageref{sec:deadline:undom_deadline}) that those with deadline $T < \underline{T}$ are dominated,
so it remains only to show that those with deadline $T \geq \underline{T}$ are not.
By \Cref{theorem:deadline}, it suffices to prove that $( x^T, X^T )$ for $T \in \left[ \underline{T}, \infty \right]$
is not dominated by another deadline mechanism.%
	\footnote{Were $( x^T, X^T )$ dominated, it would be dominated by an undominated mechanism (\Cref{proposition:dominated_by_undominated}, \cref{suppl:undom_opt_properties}), which by \Cref{theorem:deadline} must be a deadline mechanism.}

\emph{Part~1: finite deadlines.}
Fix a deadline $T \in \left[ \underline{T}, \infty \right)$;
we shall identify a distribution $G$ under which the deadline $T$ yields a strictly higher payoff than any other deadline.
In particular, consider the point mass at $T - \underline{T}$.
The mechanism $\bigl( x^T, X^T \bigr)$
has $x = u^0$ on $\left[ 0, T-\underline{T} \right] \subseteq \left[ 0, T \right]$
and
\begin{equation*}
	X^T_{T-\underline{T}}
	= \left( 1 - e^{-r\underline{T}} \right) u^0
	+ e^{-r\underline{T}} u^\star
	= u^1
\end{equation*}
by \eqref{eq:deadline_X} on \cpageref{eq:deadline_X} and the definition of $\underline{T}$.
Thus $\bigl( x^T, X^T \bigr)$ provides flow payoff $F^0\left(u^0\right)$ before the breakthrough and $F^1\left(u^1\right)$ afterwards, which is the first-best.
Any other deadline $T'$ has $X^{T'}_{T-\underline{T}} \neq u^1$,
so provides a strictly lower post-disclosure payoff and a no higher pre-disclosure payoff.

\emph{Part~2: the infinite deadline.}
Fix an arbitrary finite deadline $T \in [0,\infty)$;
we must show that $\bigl( x^T, X^T \bigr)$
does not dominate $\left( x^\infty, X^\infty \right)$.
To that end, we shall identify a distribution $G$ under which the former mechanism is strictly worse.
In particular, let $G^t$ denote the point mass at some $t \geq T$.
Under this distribution, the payoff difference between the two mechanisms is
\begin{multline*}
	\pi_{G^t}(T)
	- \pi_{G^t}(\infty)
	= e^{-rt} \left\{
	\left[ F^1\left( u^\star \right) - F^1\left( u^0 \right) \right]
	- \left[ F^0\left( u^\star \right) - F^0\left( u^0 \right) \right]
	\right\}
	\\
	+ e^{-rT}
	\left[ F^0\left( u^\star \right) - F^0\left( u^0 \right) \right] .
\end{multline*}
The second term is strictly negative since $F^0$ is uniquely maximised at $u^0$ and $u^\star \leq u^1 < u^0$.
By choosing $t \geq T$ large enough, we can make the first term as small as we wish,
so that the payoff difference is strictly negative.
\qed

\subsection{Generalisation and proof of \texorpdfstring{\Cref{proposition:opt_deadline}}{Proposition~\ref{proposition:opt_deadline}} (p.~\pageref{proposition:opt_deadline})}
\label{app:opt_deadline_pf}

In this appendix,
we obtain a general characterisation of optimal deadlines which entails \Cref{proposition:opt_deadline} and which delivers comparative statics.
Write $\bigl( x^T, X^T \bigr)$ for the deadline mechanism with deadline $T \in [0,\infty]$,
and consider the first-order condition
\begin{align*}
	&\left[ 1 - G(T) \right] \alpha
	+ \int_{[0,T]} F^{1+}\left( X^T_t \right) G( \dd t )
	\nonumber
	\\
	\leq 0
	\leq{}
	&\left[ 1 - G(T-) \right] \alpha
	+ \int_{[0,T)} F^{1-}\left( X^T_t \right) G( \dd t ) ,
	\tag{$\partial$}
	\label{eq:foc_general}
\end{align*}
where $F^{1-}$ ($F^{1+}$) is the left-hand (right-hand) derivative of $F^1$,%
	\footnote{These are well-defined since $F^1$ is concave.}
\begin{equation*}
	\alpha
	\coloneqq \frac{ F^0\left( u^0 \right) - F^0\left( u^\star \right) }
	{ u^0 - u^\star } ,
\end{equation*}
and $G(T-) \coloneqq \lim_{t \uparrow T} G(t)$ for $T>0$,
$G(0-) \coloneqq G(0)$ and $G(\infty) \coloneqq 1$.

\begin{remark}
	\label{remark:foc_general_specific}
	If $F^1$ is differentiable on $\left( 0, u^0 \right)$,
	then \eqref{eq:foc_general} reads
	\begin{equation*}
		\left[ G(T) - G(T-) \right]
		\left[ F^{1\prime}\left( u^\star \right) - \alpha \right]
		\leq \left[ 1 - G(T) \right] \alpha
		+ \int_{[0,T]} F^{1\prime}\left( X^T_t \right) G( \dd t )
		\leq 0 .%
			\footnote{In case $u^\star = 0$, we write $F^{1\prime}(0) \coloneqq F^{1+}(0)$ and assume that the latter is finite.}
	\end{equation*}
	If in addition $F^0$ is affine on $\left[ 0, u^0 \right]$ and $u^\star$ strictly exceeds zero,
	then
	\begin{align*}
		\alpha
		= F^{0\prime}\left( u^\star \right)
		= F^{1\prime}\left( u^\star \right)
		= F^{1\prime}\left( X^T_t \right)
		\quad \text{for any $t \geq T$,}
	\end{align*}
	and thus \eqref{eq:foc_general} may be written $\E_G\left( F^{1\prime}\left( X^T_\tau \right) \right) = 0$, as in \Cref{proposition:opt_deadline}.
\end{remark}

Whether or not it is exactly optimal to use a deadline mechanism,
\eqref{eq:foc_general} is a necessary condition for optimal choice \emph{among deadline mechanisms:}

\begin{lemma}
	\label{lemma:opt_deadline_nec}
	Among deadline mechanisms with finite deadline,
	the best for $G$ satisfy \eqref{eq:foc_general}.
\end{lemma}

In the affine case, \eqref{eq:foc_general} is both necessary and sufficient:

\begin{namedthm}[\Cref*{proposition:opt_deadline}$\boldsymbol{'}$.]
	\label{proposition:opt_deadline_general}
	If the old frontier $F^0$ is affine on $\left[ 0, u^0 \right]$,
	then a mechanism is optimal for $G$
	iff it is a deadline mechanism with deadline satisfying \eqref{eq:foc_general}.
\end{namedthm}

In light of \Cref{remark:foc_general_specific}, this result immediately implies \Cref{proposition:opt_deadline}.
Finally, optimal deadlines are monotone in the distribution $G$:

\begin{proposition}[comparative statics]
	\label{proposition:opt_deadline_comp}
	If the old frontier $F^0$ is affine on $\left[0,u^0\right]$
	and $G$ first-order stochastically dominates $G^\dag$,
	then $T \geq T^\dag$
	for some deadlines $T$ and $T^\dag$
	that are optimal for $G$ and $G^\dag$, respectively.
\end{proposition}

To prove the above results,
we rely on two observations:

\begin{observation}
	\label{observation:undom_not_optimal}
	A deadline $T \in \left[0,\infty\right]$ satisfies \eqref{eq:foc_general} for \emph{some} distribution $G$ exactly if it belongs to $\left[ \underline{T}, \infty \right)$.%
		\footnote{Each deadline $T \in \left[ \underline{T}, \infty \right)$
		satisfies \eqref{eq:foc_general} when $G$ is the point mass at $T-\underline{T}$.
		Conversely, any $T < \underline{T}$ violates the first inequality in \eqref{eq:foc_general} since then $X^T < u^1$ and thus $F^{1+}( X^T ) > 0$,
		while $T=\infty$ violates the second inequality because $F^{1-}( X^\infty ) \equiv F^{1-}(u^0) < 0$.}
\end{observation}

\begin{observation}
	\label{observation:piG_derivs}
	Write $\pi_G(T)$ for the principal's payoff under $G$ from deadline $T$.
	Letting `$\meet$' denote the minimum, $\pi_G(T)$ is equal to
	\begin{equation*}
		\int_{\R_+} \left[
		r \int_0^{t \meet T} e^{-rs} F^0\left( u^0 \right) \dd s
		+ r \int_{t \meet T}^t e^{-rs} F^0\left( u^\star \right) \dd s
		+ e^{-rt} F^1\left( X^T_t \right)
		\right] G( \dd t ) .
	\end{equation*}
	Its right- and left-hand derivatives are (for a constant $K>0$)
	\begin{align*}
		\pi_G^+(T)
		&= e^{-rT} K \left( \left[ 1 - G(T) \right]
		\alpha
		+ \int_{[0,T]} F^{1+}\left( X^T_t \right) G( \dd t )
		\right)
		\quad &&\text{for $T \in [0,\infty)$}
		\\
		\pi_G^-(T)
		&=
		e^{-rT} K \left(
		\left[ 1 - G(T-) \right]
		\alpha
		+ \int_{[0,T)} F^{1-}\left( X^T_t \right) G( \dd t )
		\right)
		&&\text{for $T \in (0,\infty)$.}
	\end{align*}
\end{observation}

\begin{proof}[Proof of \Cref{lemma:opt_deadline_nec}]
	$\pi_G^+(T) \leq 0$ is necessary for $T \in [0,\infty)$ to be best,
	and this rules out $T=0$ since $\pi_G^+(0) > 0$.
	Furthermore, $\pi_G^-(T) \geq 0$ is necessary for $T \in (0,\infty)$ to be best.
	So any best $T<\infty$ satisfies \eqref{eq:foc_general}.
\end{proof}

\begin{proof}[Proof of {\hyperref[proposition:opt_deadline_general]{\Cref*{proposition:opt_deadline}$\xslantmath{'}$}}]
	All optimal mechanisms are deadline mechanisms by \Cref{theorem:deadline} (\cpageref{theorem:deadline}),
	and their deadlines satisfy \eqref{eq:foc_general} if finite by \Cref{lemma:opt_deadline_nec}.
	To rule out the infinite deadline,
	define $\smash{\phi_T \coloneqq F^{1-}\bigl( X^T \bigr) \1_{[0,T)}}$
	for each $T \in \R_+$,
	and note that $\phi_T \to F^{1-}\left(u^0\right)$ pointwise as $T \to \infty$
	(since
	$X^T \uparrow u^0$ pointwise and
	$F^{1-}$ is left-continuous)
	and that $(\phi_T)_{T \in \R_+}$
	is uniformly bounded above by $\alpha$.%
		\footnote{Since $X^T > u^\star$ on $[0,T)$,
		we need only show that $F^{1-} \leq \alpha$ on $(u^\star,u^0)$.
		If $u^\star=0$, then $F^{1-} < \alpha$ on $(0,u^0)$ by definition of $u^\star$.
		And if $u^\star>0$, then $F^{1-} \leq F^{1+}(u^\star) \leq \alpha$ on $(u^\star,u^0)$
		since $F^1$ is concave
		and $\alpha$ is a supergradient of $F^1$ at $u^\star$ (by definition of $u^\star$).}
	Thus by Fatou's lemma,
	\begin{equation*}
		\limsup_{T \to \infty}
		\int_{[0,T)} F^{1-}\left( X^T_t \right) G( \dd t )
		= \limsup_{T \to \infty}
		\int_{\R_+} \phi_T \dd G
		\leq F^{1-}\left( u^0 \right)
		< 0 ,
	\end{equation*}
	so that $\pi_G^-(T) < 0$ for all sufficiently large $T \in \R_+$.
	Hence $\pi_G$ is eventually strictly decreasing,
	so that any sufficiently late deadline $T<\infty$
	is strictly better than $\infty$:
	namely, $\pi_G(T)
	> \lim_{T' \to \infty} \pi_G(T')
	= \pi_G(\infty)$.

	For the converse, consider a deadline mechanism $\bigl( x^T, X^T \bigr)$ that satisfies \eqref{eq:foc_general}.
	Then $T \geq \underline{T}$ by \Cref{observation:undom_not_optimal},
	so that $\bigl( x^T, X^T \bigr)$ is undominated by \Cref{proposition:dl_charac} (\cpageref{proposition:dl_charac}).
	It remains to show that $\bigl( x^T, X^T \bigr)$ maximises the principal's payoff under $G$,
	for which it suffices that $T$ maximise $\pi_G$.%
		\footnote{Then $( x^T, X^T )$ is better under $G$ than any other deadline mechanism.
		And it is better than any \emph{non-}deadline mechanism $(x,X)$
		because any such is dominated by some undominated mechanism (by \Cref{proposition:dominated_by_undominated} in \cref{suppl:undom_opt_properties}),
		which by \Cref{theorem:deadline} must be a deadline mechanism $\smash{( x^{T'}, X^{T'} )}$,
		so that $\smash{\Pi_G( x^T, X^T ) \geq \Pi_G( x^{T'}, X^{T'} ) \geq \Pi_G(x,X)}$.}

	Since $T$ satisfies \eqref{eq:foc_general},
	we need only show that this first-order condition is sufficient for maximisation of $\pi_G$,
	by establishing that $\pi_G^+$ and $\pi_G^-$ are down-crossing.%
		\footnote{I.e. that $\pi_G^+(T) \leq \mathrel{(<)} 0$
		implies $\pi_G^+(T') \leq \mathrel{(<)} 0$ for any $T<T'$,
		and similarly for $\pi_G^-$.}
	It suffices to show that $T \mapsto e^{rT} \pi_G^+(T)$ and $T \mapsto e^{rT} \pi_G^-(T)$ are decreasing.
	For the former, take $T < T'$ and compute
	\begin{multline*}
		\frac{ e^{rT} \pi_G^+(T') - e^{rT} \pi_G^+(T) }{ K }
		= \left[ -G(T') - G(T) \right] \alpha
		+ \int_{(T,T']} F^{1+}\left( X^{T'}_t \right) G( \dd t )
		\\
		\qquad\qquad + \int_{[0,T]} \left[
		F^{1+}\left( X^{T'}_t \right)
		- F^{1+}\left( X^T_t \right)
		\right] G( \dd t )
		\\
		= \int_{(T,T']} \left[
		F^{1+}\left( X^{T'}_t \right) - \alpha
		\right] G( \dd t )
		+ \int_{[0,T]} \left[
		F^{1+}\left( X^{T'}_t \right)
		- F^{1+}\left( X^T_t \right)
		\right] G(\dd t) .
	\end{multline*}
	The first term is non-positive since $F^{1+} \leq \alpha$ on $\left[ u^\star, u^0 \right] \ni X^{T'}$, and the second is non-positive since $F^{1+}$ is decreasing and $\smash{X^{T'} \geq X^T}$.
	Similarly,
	\begin{multline*}
		\frac{ e^{rT} \pi_G^-(T') - e^{rT} \pi_G^-(T) }{ K }
		\\
		= \int_{[T,T')} \left[
		F^{1-}\left( X^{T'}_t \right) - \alpha
		\right] G( \dd t )
		+ \int_{[0,T)} \left[
		F^{1-}\left( X^{T'}_t \right)
		- F^{1-}\left( X^T_t \right)
		\right] G(\dd t) ,
	\end{multline*}
	where the second term is non-positive since $F^{1-}$ is decreasing.
	The first term is also non-positive because
	$\smash{F^{1-}\bigl( X^{T'}_t \bigr)
	\leq F^{1+}\left( u^\star \right)
	\leq \alpha}$ for each $t \in [T,T')$,
	where the first inequality holds since $F^1$ is concave and $\smash{X^{T'}_t > u^\star}$ for every $t<T'$,
	and the second holds by definition of $u^\star$.
\end{proof}

\begin{proof}[Proof of \Cref{proposition:opt_deadline_comp}]
	By Topkis's theorem,%
		\footnote{See e.g. Theorem~2.8.1 in \textcite[p. 76]{Topkis1998}.}
	it suffices to show that $\pi_G^+ \geq \pi_{G^\dag}^+$ and $\pi_G^- \geq \pi_{G^\dag}^-$ (`increasing differences').
	We have for any $T \in \R_+$ that
	\begin{align*}
		\frac{ e^{rT} \pi_G^+(T) }{ K }
		&= \E_G\left(
		\1_{[0,T]}(\tau) \times F^{1+}\left( X^T_\tau \right)
		+ \1_{(T,\infty)}(\tau) \times \alpha
		\right)
		\\
		&\geq \E_{G^\dag}\left(
		\1_{[0,T]}(\tau) \times F^{1+}\left( X^T_\tau \right)
		+ \1_{(T,\infty)}(\tau) \times \alpha
		\right)
		= \frac{ e^{rT} \pi_{G^\dag}^+(T) }{ K } ,
	\end{align*}
	where the equalities hold by \Cref{observation:piG_derivs},
	and the inequality holds because
	$G$ first-order stochastically dominates $G^\dag$
	and the map
	\begin{equation*}
		t
		\mapsto \1_{[0,T]}(t) \times F^{1+}\left( X^T_t \right)
		+ \1_{(T,\infty)}(t) \times \alpha
	\end{equation*}
	is increasing since $F^{1+}$ and $X^T$ are decreasing
	and we have $F^{1+} \leq \alpha$ on $\left[u^\star,u^0\right] \ni X^T$.
	A similar argument shows that $\pi_G^- \geq \pi_{G^\dag}^-$: the map is
	\begin{equation*}
		t
		\mapsto \1_{[0,T)}(t) \times F^{1-}\left( X^T_t \right)
		+ \1_{[T,\infty)}(t) \times \alpha ,
	\end{equation*}
	which is increasing since $F^{1+}$ and $X^T$ are decreasing
	and $F^{1-}\bigl( X^T_t \bigr) \leq F^{1+}\left( u^\star \right) \leq \alpha$ for $t<T$,
	where the first inequality holds since $F^1$ is concave and $X^T_t > u^\star$ for $t<T$,
	and the second follows from the definition of $u^\star$.
\end{proof}

\subsection{A superdifferential Euler equation}
\label{app:Euler}

In this appendix, we argue that optimal mechanisms are described by an Euler equation, and that this equation admits a solution with nice properties.

Formalising these claims poses two technical challenges:
(a)~the frontiers $F^0,F^1$ are (concave but) not necessarily differentiable, requiring us to use the \emph{superdifferential} calculus, and (b)~the frontiers' slopes (supergradients) may be unbounded.
Some work is required to overcome these challenges, both in this appendix and in the two that follow it. We view this work as worthwhile, since the alternative would be to assume away challenges (a) and (b), thereby ruling out many applications (e.g. \Cref{fig:projects} on \cpageref{fig:projects}). We shall rely heavily on the convex-analysis concepts
reviewed in \cref{app:background} (\cpageref{app:background}).

In the present appendix, our task is to define a superdifferential Euler equation for the principal's problem and relate it to optimality (§\ref{app:Euler:Euler_lemma}) and to construct a solution (§\ref{app:Euler:construction}).
These tools will be used in the next two appendices to prove \Cref{theorem:opt} and \Cref{proposition:opt_transition} (\cpageref{theorem:opt,proposition:opt_transition}).

\subsubsection{The Euler equation and optimality}
\label{app:Euler:Euler_lemma}

\begin{definition}
	\label{definition:Euler_general}
	Given a distribution $G$,
	a mechanism $(x,X)$ \emph{satisfies the Euler equation (for $G$)} iff
	there is a measurable $\phi^0 : \R_+ \to [0,\infty]$ and a $G$-integrable $\phi^1 : \R_+ \to [-\infty,\infty]$ such that
	$\phi^0(t)$ is a supergradient of $F^0$ at $x_t$ for almost all $t \in \R_+$ such that $G(t) < 1$,
	$\phi^1(t)$ is a supergradient of $F^1$ at $X_t$ for $G$-almost all $t \in \R_+$,
	and
	\begin{equation*}
		\left[ 1 - G(t) \right] \phi^0(t)
		+ \int_{[0,t]} \phi^1 \dd G = 0 \quad \text{for every $t \in \R_+$.}
		\label{eq:euler}
		\tag{E}
	\end{equation*}
\end{definition}

Let $\mathcal{X}$ be the set of all measurable maps $\R_+ \to \left[0,u^0\right]$.
For a given breakthrough distribution $G$, define $\pi_G : \mathcal{X} \to [-\infty,\infty)$ by
\begin{equation*}
	\pi_G(x)
	\coloneqq \Pi_G\left(x,X\right)
	= \E_G\left( r \int_0^\tau e^{-rs} F^0(x_s) \dd s
	+ e^{-r\tau} F^1\left(X_\tau\right)\right) .
\end{equation*}
This is the principal's payoff under $G$ from the mechanism $(x,X)$.

\begin{namedthm}[Euler lemma.]
	\label{lemma:euler}
	Let $G$ be any distribution,
	and suppose that a mechanism $(x,X)$ with $x \in \mathcal{X}$ satisfies the Euler equation (with some $\phi^0,\phi^1$).
	Then $x \in \argmax_{\mathcal{X}} \pi_G$.
	Moreover, any mechanism $\bigl( x^\dag, X^\dag \bigr)$
	with $x^\dag \in \argmax_{\mathcal{X}} \pi_G$ satisfies the Euler equation with (the same) $\phi^0,\phi^1$.
\end{namedthm}

The proof is in \cref{suppl:Euler_lemma_pf}.

To interpret the Euler equation, recall the sketch proof of \Cref{proposition:opt_transition} (\cpageref{proposition:opt_transition}), and note that if the the frontiers $F^0,F^1$ are differentiable on $\left( 0, u^0 \right)$, then a mechanism $(x,X)$ with $0 < x < u^0$ satisfies the Euler equation exactly if \eqref{eq:euler_true} on \cpageref{eq:euler_true} holds for a.e. $t \in \R_+$.%
	\footnote{For a detailed proof of this equivalence, see \Cref{observation:euler_differentiable} in §\ref{app:Euler:construction} below.}

For bounded $\phi^0$, the backward-looking integral equation \eqref{eq:euler}
is equivalent to a forward-looking integral equation plus an initial condition:

\begin{observation}
	\label{observation:euler_forward}
	For a distribution $G$,
	a bounded and measurable $\phi^0 : \R_+ \to \R$
	and a $G$-integrable $\phi^1 : \R_+ \to [-\infty,\infty]$,
	equation \eqref{eq:euler} holds iff
	$\E_G\left(\phi^1(\tau)\right) = 0$ and
	\begin{equation}
		\phi^0(t)
		= \E_G\left(\phi^1(\tau) \middle| \tau > t\right)
		\quad \text{for every $t \in \R_+$ such that $G(t) < 1$.}
		\label{eq:euler_forward}
	\end{equation}
\end{observation}

\begin{proof}
	For any $t \in \R_+$,
	$\int_{(t,\infty)} \phi^1 \dd G$ is finite since $\phi^1$ is $G$-integrable,
	so we may add and subtract it to obtain
	\begin{multline*}
		[1-G(t)] \phi^0(t)
		+ \int_{[0,t]} \phi^1 \dd G
		\\
		=
		\begin{cases}
			[1-G(t)] \left[ \phi^0(t)
			- \E_G\left( \phi^1(\tau) \middle| \tau>t \right) \right]
			+ \E_G\left( \phi^1(\tau) \right)
			& \text{if $G(t)<1$}
			\\
			\E_G\left( \phi^1(\tau) \right)
			& \text{if $G(t)=1$.}
		\end{cases}
	\end{multline*}
	Thus $\E_G\left(\phi^1(\tau)\right) = 0$ and \eqref{eq:euler_forward} imply \eqref{eq:euler}.
	Conversely, if \eqref{eq:euler} holds,
	then letting $t \to \infty$
	and using the boundedness of $\phi^0$ yields
	\begin{equation*}
		0
		= - \lim_{t \to \infty} [1-G(t)] \phi^0(t)
		= \lim_{t \to \infty} \int_{\R_+} \phi^1 \1_{[0,t]} \dd G
		= \int_{\R_+} \phi^1 \dd G
		= \E_G\left( \phi^1(\tau) \right) ,
	\end{equation*}
	where the third equality holds by dominated convergence;
	thus \eqref{eq:euler_forward} holds.
\end{proof}

\subsubsection{Constructing a solution of the Euler equation}
\label{app:Euler:construction}

\begin{definition}
	\label{definition:simple_Fs}
	$F^0,F^1$ are \emph{simple} if they are strictly concave and possess bounded derivatives,
	$F^{1\prime}$ is Lipschitz continuous on $\left[u^\star,u^0\right]$,
	and $u^\star > 0$.
\end{definition}

\begin{observation}
	\label{observation:euler_differentiable}
	If $F^0,F^1$ are simple,
	then a mechanism $(x,X)$ with $u^\star \leq x \leq u^0$ satisfies the Euler equation iff
	it satisfies \eqref{eq:euler_true} on \cpageref{eq:euler_true} for a.e. $t \in \R_+$,
	or equivalently (by \Cref{observation:euler_forward} above)
	$\E_G\left( F^{1\prime}(X_\tau) \right) = 0$ and
	\begin{equation*}
		F^{0\prime}(x_t)
		= \E_G\left( F^{1\prime}(X_\tau) \middle| \tau > t\right)
		\quad \text{for a.e. $t \in \R_+$ such that $G(t) < 1$.}
	\end{equation*}
\end{observation}

\begin{proof}
	Fix $(x,X)$ with $u^\star \leq x \leq u^0$.
	If $(x,X)$ satisfies the Euler equation with $\phi^0,\phi^1$,
	then $\phi^1(s) = F^{1\prime}(X_s)$ for $G$-a.e. $s \in \R_+$,
	so that \eqref{eq:euler} reads
	\begin{equation*}
		\left[ 1 - G(t) \right] \phi^0(t)
		+ \int_{[0,t]} F^{1\prime}(X_s) G( \dd s ) = 0
		\quad \text{for every $t \in \R_+$,}
	\end{equation*}
	and thus \eqref{eq:euler_true} holds for a.e. $t \in \R_+$ since $\phi^0(t) = F^{0\prime}(x_t)$ for a.e. $t \in \R_+$ with $G(t)<1$.

	Suppose instead that $(x,X)$ satisfies \eqref{eq:euler_true} for a.e. $t \in \R_+$.
	Let $T \coloneqq \inf\{ t \in \R_+ : G(t)=1 \}$ with the convention that $\inf \varnothing \coloneqq \infty$,
	and define $\phi^0,\phi^1 : \R_+ \to \R$ by
	\begin{equation*}
		\phi^0(t)
		\coloneqq
		\begin{cases}
			-\frac{1}{1-G(t)} \int_{[0,t]} F^{1\prime}(X_s) G( \dd s )
			& \text{for $t < T$} \\
			F^{0\prime}(x_t)
			& \text{for $t \geq T$}
		\end{cases}
	\end{equation*}
	and $\phi^1(t) \coloneqq F^{1\prime}(X_t)$ for every $t \in \R_+$.
	Then $\phi^0(t) = F^{0\prime}(x_t)$ for a.e. $t \in \R_+$ by \eqref{eq:euler_true},
	and $\phi^0,\phi^1$ satisfy \eqref{eq:euler}.
\end{proof}

Let $\mathcal{X}'$ be the set of all decreasing maps $\R_+ \to \left[u^\star,u^0\right]$, endowed with the topology of pointwise convergence.
Given a sequence of technologies $\left(F^0_n,F^1_n\right)_{n \in \N}$ satisfying our model assumptions, write $u^0_n$, $u^\star_n$, and $\mathcal{X}'_n$ for the analogues of $u^0$, $u^\star$ and $\mathcal{X}'$, respectively.

\begin{observation}
	\label{observation:thmC_pf_sequential_compactness}
	For technologies $F^0,F^1$
	and $\left(F^0_n,F^1_n\right)_{n \in \N}$ such that $u^\star_n \to u^\star$ and $u^0_n \uparrow u^0$,
	any sequence $(x^n)_{n \in \N}$ with $x^n \in \mathcal{X}'_n$ for each $n \in \N$
	admits a convergent subsequence with limit in $\mathcal{X}'$.
	(Thus $\mathcal{X}'$ is sequentially compact.)
\end{observation}

\begin{proof}
	The sequence $(x^n)_{n \in \N}$ lives in $\left[0,u^0\right]$
	since $x^n \leq u^0_n \leq u^0$ for each $n \in \N$.
	Thus by the Helly selection theorem \parencite[e.g.][p. 167]{Rudin1976},
	$(x^n)_{n \in \N}$ admits a subsequence
	along which it converges pointwise to some decreasing $x : \R_+ \to \left[0,u^0\right]$.
	We have $x \geq u^\star$
	since $x^n \geq u^\star_n$ for each $n \in \N$
	and $u^\star_n \to u^\star$.
	By considering the constant sequence $\left(F^0_n,F^1_n\right) \equiv \left(F^0,F^1\right)$,
	we see that $\mathcal{X}'$ is sequentially compact.
\end{proof}

The following three lemmata construct a solution of the Euler equation.
Their (tedious) proofs are relegated to \cref{suppl:construction}.

\begin{lemma}
	\label{lemma:existence_simple}
	If $F^0,F^1$ are simple and $G$ has finite support,
	then there exists an $x \in \mathcal{X}'$ such that $(x,X)$ satisfies the Euler equation.
\end{lemma}

\begin{lemma}
	\label{lemma:approx_G}
	Let $F^0,F^1$ be simple,
	and let $(G_n)_{n \in \N}$ be a sequence of finite-support CDFs converging pointwise to a CDF $G$.
	Let $(x^n)_{n \in \N}$ be a sequence in $\mathcal{X}'$ such that $\left(x^n,X^n\right)$ satisfies the Euler equation for $\left(F^0,F^1,G_n\right)$ for each $n \in \N$,
	and suppose that $(x^n)_{n \in \N}$ converges pointwise to some $x \in \mathcal{X}'$.
	Then $(x,X)$ satisfies the Euler equation for $\left(F^0,F^1,G\right)$.
\end{lemma}

\begin{lemma}
	\label{lemma:approx_F}
	Given $F^0,F^1$, there exists a sequence $\left(F^0_n,F^1_n\right)_{n \in \N}$ of simple technologies
	such that $u^0_n \uparrow u^0$ and $u^\star_n \to u^\star$ as $n \to \infty$
	and, for any CDF $G$ with unbounded support and any mechanism $(x,X)$,
	if $x$ is the pointwise limit of a sequence $(x^n)_{n \in \N}$
	along which $(x^n,X^n)$ satisfies the Euler equation for $\left(F^0_n,F^1_n,G\right)$
	and $x^n \in \mathcal{X}'_n$ for each $n \in \N$,
	then $(x,X)$ satisfies the Euler equation for $\left(F^0,F^1,G\right)$ with some increasing $\phi^0,\phi^1$.
\end{lemma}

The following will be used in the next appendix to prove \Cref{theorem:opt}.

\begin{namedthm}[Existence corollary.]
	\label{corollary:existence}
	For any distribution $G$ with unbounded support,
	there is a mechanism $(x,X)$ with $x \in \mathcal{X}'$
	which satisfies the Euler equation for $G$ with some increasing $\phi^0,\phi^1 : \R_+ \to \R$.
\end{namedthm}

\begin{proof}
	Let $\left(F^0_n,F^1_n\right)_{n \in \N}$ be the simple technologies delivered by \Cref{lemma:approx_F}.
	Choose a sequence $(G_m)_{m \in \N}$ of finite-support distributions
	converging pointwise to $G$.

	Fix an arbitrary $n \in \N$.
	For every $m \in \N$,
	\Cref{lemma:existence_simple} assures us of the existence of an $x^{nm} \in \mathcal{X}'_n$ such that $(x^{nm},X^{nm})$
	satisfies the Euler equation for $\left(F^0_n,F^1_n,G_m\right)$.
	Since $\mathcal{X}'_n$ is sequentially compact by \Cref{observation:thmC_pf_sequential_compactness},
	we may assume (passing to a subsequence if necessary) that
	$x^{nm}$ converges pointwise as $m \to \infty$ to some $x^n \in \mathcal{X}'_n$.
	Since $u^0_n \to u^0$ and $u^\star_n \to u^\star$ as $n \to \infty$,
	\Cref{observation:thmC_pf_sequential_compactness} permits us to assume (again passing to a subsequence if required) that
	$x^n$ converges pointwise to some $x \in \mathcal{X}'$ as $n \to \infty$.

	By \Cref{lemma:approx_G}, $(x^n,X^n)$ satisfies the Euler equation for $\left(F^0_n,F^1_n,G\right)$ for each $n \in \N$.
	Hence by \Cref{lemma:approx_F}, $(x,X)$ satisfies the Euler equation for $\left(F^0,F^1,G\right)$ with some increasing $\phi^0,\phi^1$.
\end{proof}

\subsection{Proof of \texorpdfstring{\Cref{theorem:opt}}{Theorem~\ref{theorem:opt}} (p.~\pageref{theorem:opt})}
\label{app:opt}

We shall argue as follows. Fix an optimal mechanism $(x,X)$. We first show that if $x$ is decreasing, then $\lim_{t \to 0} x_t = u^0$ and $\lim_{t \to \infty} x_t = u^\star$ (\Cref{lemma:limits} below). To establish monotonicity, we rely on the Euler equation \eqref{eq:euler} (\cref{app:Euler}, \cpageref{eq:euler}), which $(x,X)$ must satisfy by the \hyperref[lemma:euler]{Euler lemma} (\cref{app:Euler}). We first show that the Euler equation implies that $x$ decreases in periods $t$ that lie in the support of $G$ and have $x_t > u^\star$ (\hyperref[claim:pf_opt:1]{claim 1} below). We then show, using a `local' version of the front-loading logic from the proof of \Cref{theorem:deadline}, that $x$ must decrease also in periods $t$ outside the support of $G$ (\hyperref[claim:pf_opt:2]{claim 2}) and in periods $t$ with $x_t = u^\star$ (\hyperref[claim:pf_opt:3]{claim 3}).

Recall from §\ref{app:Euler:construction} the definition of $\mathcal{X}'$.

\begin{lemma}
	\label{lemma:limits}
	Suppose that $G$ satisfies $G(0)=0$ and has unbounded support.
	Let $(x,X)$ with $x \in \mathcal{X}'$ satisfy the Euler equation with some $\phi^0,\phi^1$ such that $\phi^0$ is increasing.
	Then $\lim_{t \to 0} x_t = u^0$ and $\lim_{t \to \infty} x_t = u^\star$.
\end{lemma}

\begin{proof}
	Since $x$ is decreasing
	with $u^\star \leq x \leq u^0$, the limits
	\begin{equation*}
		\widebar{u} \coloneqq \lim_{t \to 0} x_t
		\quad \text{and} \quad
		\underline{u} \coloneqq \lim_{t \to \infty} x_t
	\end{equation*}
	exist and satisfy $u^\star \leq \underline{u} \leq \widebar{u} \leq u^0$.
	As $G$ has unbounded support,
	$\phi^0$ is a supergradient of $F^0$ at $x_t$ for a.e. $t \in \R_+$.

	To show that $\widebar{u} \geq u^0$,
	note that for a.e. $t \in \R_+$,
	$\phi^0(t)$ is a supergradient at $x_t \leq \widebar{u}$ of the concave function $F^0$,
	so that $\phi^0(t) \geq F^{0+}(x_t) \geq F^{0+}\left(\widebar{u}\right)$.
	Thus $\phi^0 \geq F^{0+}\left(\widebar{u}\right)$ on $(0,\infty)$ since $\phi^0$ is increasing.
	Letting $t \to 0$ in \eqref{eq:euler} (\cpageref{eq:euler}) then yields
	\begin{equation*}
		0
		= \lim_{t \to 0} \phi^0(t)
		\geq F^{0+}\left( \widebar{u} \right) ,
	\end{equation*}
	which implies that $\widebar{u} \geq u^0$ since $F^{0+} > 0$ on $\left[0,u^0\right)$
	by definition of $u^0$.

	To show that $\underline{u} \leq u^\star$, assume without loss that $\underline{u} > 0$.
	Then $\phi^0$ is bounded, since it is increasing and $\phi^0(t)$ is a supergradient of the concave function $F^0$ at $x_t$ for a.e. $t \in \R_+$.
	Hence, we may use \Cref{observation:euler_forward} (\cpageref{observation:euler_forward}) to obtain
	\begin{equation*}
		F^{0+}(x_t)
		\leq \phi^0(t)
		= \E_G\left( \phi^1(\tau) \middle| \tau > t \right)
		\leq F^{1-}\left( \underline{u} \right)
		\quad \text{for a.e. $t \in \R_+$,}
	\end{equation*}
	where the first (second) inequality holds since $\phi^0(t)$ ($\phi^1(s)$) is a supergradient of the concave function $F^0$ at $x_t$ for a.e. $t \in \R_+$ (of $F^1$ at $X_s \geq \underline{u}$ for $G$-a.e. $s \in \R_+$).
	Then $F^{0+}(x_t) \leq F^{1-}\left(\underline{u}\right)$ for \emph{every} $t \in \R_+$
	since $F^{0+}$ and $x$ are decreasing.
	Since $F^{0+}$ is right-continuous,
	letting $t \to \infty$ yields $F^{0+}\left(\underline{u}\right) \leq F^{1-}\left(\underline{u}\right)$,
	which implies that $\underline{u} \leq u^\star$ by definition of the latter.
\end{proof}

Recall from \cref{app:Euler} the definitions of $\mathcal{X}$ and $\pi_G$,
the \hyperref[lemma:euler]{Euler lemma},
and the \hyperref[corollary:existence]{existence corollary}.

\begin{proof}[Proof of \Cref{theorem:opt}]
	Let $G$ be a distribution with $G(0)=0$ and unbounded support.
	By the \hyperref[corollary:existence]{existence corollary}, there is a mechanism $\bigl(x^\dag,X^\dag\bigr)$ with $x^\dag \in \mathcal{X}'$
	which satisfies the Euler equation for $G$ with some increasing $\phi^0,\phi^1$.
	By the \hyperref[lemma:euler]{Euler lemma}, $x^\dag$ belongs to $\argmax_{\mathcal{X}} \pi_G$.

	Let $( x, X )$ be optimal for $G$;
	we must show that it has the properties asserted by \Cref{theorem:opt}.
	By \Cref{lemma:lequ0} (\cpageref{lemma:lequ0}), it must be that $x \in \mathcal{X}$.
	Thus $x$ belongs to $\argmax_{\mathcal{X}} \pi_G$,
	so by the \hyperref[lemma:euler]{Euler lemma} again,
	$( x, X )$ satisfies the Euler equation with (the above increasing) $\phi^0,\phi^1$.

	It suffices to show that some version%
		\footnote{Recall from \cref{footnote:mech_defn_formal} (\cpageref{footnote:mech_defn_formal})
		that $\widetilde{x}$ a \emph{version} of $x$ exactly if $\widetilde{x} = x$ a.e.}
	of $x$ is decreasing and $\geq u^\star$,
	since it then belongs to $\mathcal{X}'$,
	so that the remaining properties $\lim_{t \to 0} x_t = u^0$ and $\lim_{t \to \infty} x_t = u^\star$ follow by \Cref{lemma:limits}.

	Adopt the convention that $F^{0-}(0) \coloneqq \infty$.

	\begin{namedthm}[Claim 0.]
		\label{claim:pf_opt:0}
		$\phi^0 \leq F^{0-}\left(u^\star\right)$, strictly on a neighbourhood of $t=0$.
	\end{namedthm}

	\begin{proof}%
		\renewcommand{\qedsymbol}{$\square$}
		The result is immediate if $u^\star=0$, so suppose that $u^\star>0$.
		Since $\bigl(x^\dag,X^\dag\bigr)$ satisfies the Euler equation with $\phi^0,\phi^1$ and $G$ has unbounded support, $\phi^0(t)$ is a supergradient of $F^0$ at $\smash{x^\dag_t}$ for a.e. $t \in \R_+$.
		Thus since $F^0$ is concave and $x^\dag \geq u^\star$ (because $x^\dag \in \mathcal{X}'$), we have
		$\smash{\phi^0(t) \leq F^{0-}\bigl( x^\dag_t \bigr) \leq F^{0-}\left(u^\star\right)}$ for a.e. $t \in \R_+$.
		Hence $\phi^0 \leq F^{0-}\left(u^\star\right)$ since $\phi^0$ is increasing.

		Letting $t \to 0$ in \eqref{eq:euler} (\cref{app:Euler}, \cpageref{eq:euler})
		yields $\lim_{t \to 0} \phi^0(t) = 0 < F^{0-}\left(u^\star\right)$,
		so that $\phi^0(0) \leq \phi^0(t) < F^{0-}\left(u^\star\right)$ for all sufficiently small $t>0$.
	\end{proof}%
	\renewcommand{\qedsymbol}{$\blacksquare$}

	Write $T$ for the (possibly infinite) time
	at which $\phi^0$ hits $F^{0-}\left(u^\star\right)$:
	\begin{equation*}
		T
		\coloneqq \inf\left\{ t \in \R_+ :
		\phi^0(t) \geq F^{0-}\left(u^\star\right) \right\} ,
	\end{equation*}
	with the convention that $\inf \varnothing \coloneqq \infty$.
	$T$ is strictly positive by \hyperref[claim:pf_opt:0]{claim 0}.

	The (increasing) function $\phi^0$ is called \emph{non-constant at $t \in \R_+$}
	if $\phi^0(s) \neq \phi^0(t)$ for every $s \neq t$,
	and \emph{constant at $t$} otherwise.
	Clearly if $\phi^0$ is constant at $t$, then it is constant on a proper interval containing $t$.%
		\footnote{But $t$ need not be in the \emph{interior} of such an interval.}

	\begin{namedthm}[Claim 1.]
		\label{claim:pf_opt:1}
		$x = x^\dag$ a.e. on $\left\{ t \in \R_+ : \text{$\phi^0$ is non-constant at $t$} \right\}$.
	\end{namedthm}

	A set of times is \emph{prior to $T$} iff
	its intersection with $(T,\infty)$ is empty.
	(The set of times in \hyperref[claim:pf_opt:1]{claim 1} is prior to $T$,
	by \hyperref[claim:pf_opt:0]{claim 0} and the definition of $T$.)

	\begin{namedthm}[Claim 2.]
		\label{claim:pf_opt:2}
		On any proper interval of $\R_+$ prior to $T$ on which $\phi^0$ is constant,
		some version of $x$ is decreasing.
	\end{namedthm}

	\begin{namedthm}[Claim 3.]
		\label{claim:pf_opt:3}
		If $T<\infty$, then on $[T,\infty)$,
		some version of $x$ is decreasing and bounded below by $u^\star$.
	\end{namedthm}

	For each maximal proper interval of $\R_+$ prior to $T$
	on which $\phi^0$ is constant at some $\alpha \in \R$,
	\hyperref[claim:pf_opt:2]{claim 2} delivers a version $x^\alpha$ of $x$ that is decreasing on this interval.
	If $T<\infty$, then \hyperref[claim:pf_opt:3]{claim 3} provides a version $x^\star$ of $x$ that is decreasing on $[T,\infty)$ and bounded below by $u^\star$.
	Define $\widetilde{x} : \R_+ \to \R$ by
	\begin{equation*}
		\widetilde{x}_t
		\coloneqq
		\begin{cases}
			x^{\phi^0(t)}_t	& \text{if $t<T$ and $\phi^0$ is constant at $t$} \\
			x^\dag_t		& \text{if ($t<T$ and) $\phi^0$ is non-constant at $t$} \\
			x^\star_t		& \text{if $t \geq T$.}
		\end{cases}
	\end{equation*}
	We have $\widetilde{x} = x^\dag = x$ a.e. on $\left\{ t \in \R_+ : \text{$\phi^0$ is non-constant at $t$} \right\}$ by \hyperref[claim:pf_opt:1]{claim 1}.
	Thus $\widetilde{x}$ is a version of $x$.%
		\footnote{Since $\phi^0$ is increasing,
		it is constant on at most countably many intervals.
		So the definition of $\widetilde{x}$ has at most countably many cases,
		in each of which $\widetilde{x}$ equals a version of $x$.
		}

	Let $\mathcal{T}$ be the set of times $t \in \R_+$ at which $\phi^0(t)$ is a supergradient of $F^0$ at $\widetilde{x}_t$.
	Its complement $\R_+ \setminus \mathcal{T}$ is null
	since $\widetilde{x}$ is a version of $x$,
	$(x,X)$ satisfies the Euler equation with $\phi^0,\phi^1$,
	and $G$ has unbounded support.
	It therefore suffices to show that $\widetilde{x}$ is decreasing and bounded below by $u^\star$ on $\mathcal{T}$.%
		\footnote{Then define $\widebar{x}_t \coloneqq \sup_{[t,\infty) \intersect \mathcal{T}} \widetilde{x}$ for each $t \in \R_+$.
		This $\widebar{x}$ is a version of $\widetilde{x}$ (and thus of $x$),
		and is (everywhere) decreasing and bounded below by $u^\star$.}

	To see that $\widetilde{x}$ is decreasing on $\mathcal{T}$,
	fix any $s<t$ in $\mathcal{T}$; we must show that $\widetilde{x}_s \geq \widetilde{x}_t$.
	If $\phi^0(s) \neq \phi^0(t)$, then $\phi^0(s) < \phi^0(t)$ since $\phi^0$ is increasing. Since $s,t$ belong to $\mathcal{T}$ and $F^0$ is concave, it follows that
	\begin{equation*}
		F^{0+}\left(\widetilde{x}_s\right)
		\leq \phi^0(s)
		< \phi^0(t)
		\leq F^{0-}\left(\widetilde{x}_t\right) ,
	\end{equation*}
	which implies that $\widetilde{x}_s \geq \widetilde{x}_t$ since $F^0$ is concave.
	If instead $\phi^0(s) = \phi^0(t)$, then we have either $s,t < T$ or $s,t \geq T$.%
		\footnote{$s,t$ must be on the same side of $T$
		since $\phi^0$ (being increasing) is constant on $[s,t]$,
		whereas $\phi^0(T-\eps) < \phi^0(T)$ for any $\eps \in [0,T)$ by (\hyperref[claim:pf_opt:0]{claim 0} and) the definition of $T$.}
	In the former (latter) case, $\widetilde{x}$ equals the decreasing function $x^{\phi^0(t)}$ (the decreasing function $x^\star$) on $[s,t]$.

	It remains to show that $\widetilde{x} \geq u^\star$ on $\mathcal{T}$.
	If $T<\infty$, then this holds because $\widetilde{x}$ is decreasing and $\widetilde{x} = x^\star \geq u^\star$ on $[T,\infty)$.
	If instead $T=\infty$, then
	\begin{equation*}
		F^{0+}\left( \widetilde{x}_t \right)
		\leq \phi^0(t)
		< F^{0-}\left( u^\star \right)
		\quad \text{for every $t \in \mathcal{T}$}
	\end{equation*}
	by definition of $\mathcal{T}$ and the concavity of $F^0$ (weak inequality)
	and by definition of $T$ (strict inequality).
	Since $F^0$ is concave, this implies that $\widetilde{x} \geq u^\star$ on $\mathcal{T}$.

	The rest of the proof is devoted to establishing claims \hyperref[claim:pf_opt:1]{1}, \hyperref[claim:pf_opt:2]{2} and \hyperref[claim:pf_opt:3]{3}.
	The argument for the first is straightforward,
	while those for the latter two are (local) `front-loading' arguments similar to the proof of \Cref{theorem:deadline} (\cpageref{theorem:deadline}).

	\begin{proof}[Proof of {\hyperref[claim:pf_opt:1]{claim 1}}]%
		\renewcommand{\qedsymbol}{$\square$}
		Write $I \coloneqq \left\{ t \in \R_+ : \text{$\phi^0$ is non-constant at $t$} \right\}$;
		we must show that $x = x^\dag$ a.e. on $I$.
		By definition, $\phi^0$ is strictly increasing on $I$.

		Let $A$ be the set of all $\alpha \in \R$
		that are supergradients of $F^0$
		at more than one $u \in \left[0,u^0\right]$.
		$A$ is at most countable since $F^0$ is concave.
		Thus $I' \coloneqq \left\{ t \in I : \phi^0(t) \in A \right\}$ is null since $\phi^0$ is strictly increasing on $I$.

		Since $(x,X)$ and $\bigl(x^\dag,X^\dag\bigr)$ satisfy the Euler equation with $\phi^0,\phi^1$ and $G$ has unbounded support,
		$\phi^0(t)$ is a supergradient of $F^0$ at both $x_t$ and $\smash{x^\dag_t}$
		for a.e. $t \in \R_+$.
		The same therefore holds for a.e. $t \in I \setminus I'$,
		and $\smash{x_t = x^\dag_t}$ at each such $t$ by definition of $I'$ (and $A$).
		Thus $x=x^\dag$ a.e. on $I$ since $I'$ is null.
	\end{proof}%
	\renewcommand{\qedsymbol}{$\blacksquare$}

	To prove claims \hyperref[claim:pf_opt:2]{2} and \hyperref[claim:pf_opt:3]{3},
	we shall utilise a forward-looking variant of Euler equation.%
		\footnote{Similar to \Cref{observation:euler_forward} in \cref{app:Euler} (\cpageref{observation:euler_forward}), but without boundedness of $\phi^0$.}
	For any $t \in \R_+$,
	$\smash{\int_{(t,\infty)} \phi^1 \dd G}$ is finite since $\phi^1$ is $G$-integrable,
	and $G(t)<1$ since $G$ has unbounded support.
	We may therefore add and subtract $\smash{\int_{(t,\infty)} \phi^1 \dd G}$ in \eqref{eq:euler} (\cpageref{eq:euler}) and divide by $1-G(t)$ to obtain
	\begin{equation*}
		\phi^0(t)
		= \E_G\left(\phi^1(\tau) \middle| \tau > t\right)
		- \frac{\E_G\left(\phi^1(\tau)\right)}{1-G(t)}
		\quad \text{for all $t \in \R_+$.}
		\label{eq:euler_forward_general}
		\tag{E$'$}
	\end{equation*}
	Moreover, \eqref{eq:euler} and the non-negativity of $\phi^0$ imply that $\smash{\int_{[0,t]} \phi^1 \dd G \leq 0}$ for every $t \in \R_+$,
	so letting $t \to \infty$ and using dominated convergence yields%
		\footnote{In detail, $0
		\geq \lim_{t \to \infty} \int_{\R_+} \phi^1 \1_{[0,t]} \dd G
		= \int_{\R_+} \phi^1 \dd G
		= \E_G\left( \phi^1(\tau) \right)$.}
	\begin{equation*}
		\E_G\left(\phi^1(\tau)\right) \leq 0
		\label{eq:euler_average_general}
		\tag{$\infty$}
	\end{equation*}

	We next prove \hyperref[claim:pf_opt:3]{claim 3}.
	This requires a supporting claim:

	\begin{namedthm}[Claim 4.]
		\label{claim:pf_opt:4}
		If $T < \infty$,
		then $X_T \geq u^\star$.
	\end{namedthm}

	\begin{proof}%
		\renewcommand{\qedsymbol}{$\square$}
		The result is trivial if $u^\star=0$, so suppose $u^\star>0$.
		Fix any $\eps \in (0,T)$.
		(Recall that $T>0$, by \hyperref[claim:pf_opt:0]{claim 0}.)
		$\phi^0$ is not constant on $[T-\eps,T+\eps]$,
		and thus \eqref{eq:euler} (\cref{app:Euler}, \cpageref{eq:euler}) requires that $G(T-\eps) < G(T+\eps)$ since $G$ has unbounded support.
		Then since $(x,X)$ satisfies the Euler equation with $\phi^0,\phi^1$,
		it must be that $\phi^1(t)$ is a supergradient of $F^1$ at $X_t$ for some $t \in (T-\eps,T+\eps]$.

		Fix any $u \in \left[0,u^\star\right)$.
		Since $u^\star$ is a strict local maximum of $F^1-F^0$, $F^1-F^0$ is not decreasing on $[u,u^\star]$,
		and thus there is a $u' \in [u,u^\star)$ at which $F^{1+}(u') > F^{0-}\left(u^\star\right)$.%
			\footnote{If not, then $F^1-F^0$ would be decreasing on $[u,u^\star]$ since
			$( F^1 - F^0 )^+
			= F^{1+} - F^{0+}
			\leq F^{1+} - F^{0-}(u^\star)
			\leq 0$ on $[u,u^\star)$,
			where the first inequality holds since $F^0$ is concave.}
		Then
		\begin{equation*}
			F^{1+}(u')
			> F^{0-}\left(u^\star\right)
			\geq \phi^0(t)
			=
			\E_G\left(\phi^1(\tau) \middle| \tau > t\right)
			- \frac{\E_G\left(\phi^1(\tau)\right)}{1-G(t)}
			\geq \phi^1(t)
		\end{equation*}
		where the second inequality holds by \hyperref[claim:pf_opt:0]{claim 0},
		the equality is \eqref{eq:euler_forward_general},
		and the last inequality holds by \eqref{eq:euler_average_general} and the fact that $\phi^1$ is increasing.
		Thus $X_t > u' \geq u$
		since $\phi^1(t)$ is a supergradient at $X_t$ of the concave function $F^1$.

		Since $\eps \in (0,T)$ and $u \in [0,u^\star)$ were arbitrary and $X$ is continuous,
		it follows that $X_T \geq u^\star$.
	\end{proof}%
	\renewcommand{\qedsymbol}{$\blacksquare$}

	Write $\pi_t \coloneqq \pi_{G^t}$ for $t>0$, where $G^t$ denotes the point mass at $t$.

	\begin{proof}[Proof of {\hyperref[claim:pf_opt:3]{claim 3}}]%
		\renewcommand{\qedsymbol}{$\square$}
		Let $u' \in \left[u^\star,u^0\right]$ be the largest $u \in \R_+$
		at which $F^0$ admits $F^{0-}\left(u^\star\right)$ as a supergradient.
		We may have $u'=u^\star$;
		if not, then $F^0$ is affine on the interval $\left[u^\star,u'\right]$,
		with slope $F^{0-}\left(u^\star\right)$.

		We have $\phi^0 = F^{0-}\left(u^\star\right)$ on $(T,\infty)$
		by \hyperref[claim:pf_opt:0]{claim 0},
		the definition of $T$
		and the fact that $\phi^0$ is increasing.
		Then since $(x,X)$ satisfies the Euler equation with $\phi^0,\phi^1$ and $G$ has unbounded support,
		we must have $x \leq u'$ a.e. on $(T,\infty)$.
		It follows that $X \leq u'$ on $[T,\infty)$.

		On the other hand, we have $X_T \geq u^\star$ by \hyperref[claim:pf_opt:4]{claim 4}.
		If $u'=u^\star$, then we are done: $X=u^\star$ on $[T,\infty)$,
		and thus $x = u^\star$ a.e. on $(T,\infty)$,
		which obviously has a version that is decreasing and bounded below by $u^\star$ on $[T,\infty)$.

		It remains to consider the case in which $u'>u^\star$,
		meaning that $F^0$ has an affine segment with slope $F^{0-}\left(u^\star\right)$ extending from $u^\star$ to $u'$.
		We shall front-load the mechanism $(x,X)$ over this affine segment,
		much as in the proof of \Cref{theorem:deadline} (\cpageref{theorem:deadline}).
		In particular, given a deadline $T' \in [T,\infty]$, consider
		\begin{equation*}
			x^\star_t
			=
			\begin{cases}
				x_t		& \text{for $t \in [0,T)$} \\
				u'		& \text{for $t \in \left[T,T'\right)$} \\
				u^\star	& \text{for $t \in \left[T',\infty\right)$.}
			\end{cases}
		\end{equation*}
		Since $u^\star \leq X_T \leq u'$,
		we may choose the deadline $T'$ so that $X^\star_T = X_T$.

		We will show that the front-loaded mechanism $\left(x^\star,X^\star\right)$
		dominates $(x,X)$
		unless $X^\star=X$.
		This suffices because $(x,X)$ is undominated (being optimal for $G$),
		so that we must have $X = X^\star$
		and thus $x = x^\star$ a.e.;
		and $x^\star$ is decreasing and bounded below by $u^\star$ on $[T,\infty)$.

		Clearly $\pi_t\left(x^\star\right) = \pi_t(x)$ for all $t \leq T$;
		we will show that for each $t > T$,
		we have $\pi_t\left(x^\star\right) \geq \pi_t(x)$,
		with equality only if $X^\star_t = X_t$.
		Define
		\begin{equation*}
			\widehat{F}^0(u)
			\coloneqq F^0\left(u'\right)
			- \left(u'-u\right) F^{0-}\left(u'\right)
			\quad \text{for each $u \in \R_+$.}
		\end{equation*}
		We have $\widehat{F}^0 \geq F^0$ (with equality on $\left[u^\star,u'\right]$)
		since $F^{0-}\left(u'\right) = F^{0-}\left(u^\star\right)$ is a supergradient of $F^0$ at $u'$
		(at every $u \in \left[u^\star,u'\right]$).
		Thus for any $t > T$, we have
		\begin{align*}
			\pi_t\left(x\right)-\pi_t\left(x^\star\right)
			&= r \int_T^t e^{-rs}
			\left[ F^0\left(x_s\right) - F^0(x^\star_s) \right] \dd s
			+ e^{-rt} \left[ F^1\left(X_t\right)
			- F^1\left(X^\star_t\right) \right]
			\\
			&\leq r \int_T^t e^{-rs}
			\left[\widehat{F}^0\left(x_s\right)
			- \widehat{F}^0(x^\star_s) \right] \dd s
			+ e^{-rt}\left[ F^1\left(X_t\right)
			- F^1\left(X^\star_t\right) \right]
			\\
			&= e^{-rt} \left[ \left(F^1-\widehat{F}^0\right)\left(X_t\right)
			- \left(F^1-\widehat{F}^0\right)(X^\star_t) \right] ,
		\end{align*}
		where the first equality holds since $x = x^\star$ on $[0,T]$,
		the inequality holds since $F^0 \leq \widehat{F}^0$
		with equality on $\left[u^\star,u'\right] \ni x^\star$,
		and the final equality holds since $\widehat{F}^0$ is affine on $[0,u']$ and $X_T = X^\star_T$.

		Since $u^\star$ is a strict local maximum of $F^1-F^0$
		and $F^1-\widehat{F}^0 \leq F^1-F^0$ with equality at $u^\star$,
		it must be that $u^\star$ is a strict local maximum of $F^1-\widehat{F}^0$.
		Thus since $F^1-\widehat{F}^0$ is concave,
		it is strictly increasing on $\left[0,u^\star\right]$ and strictly decreasing on $\left[u^\star,u'\right]$.
		It thus suffices to show that $X^\star$ lies between $X$ and $u^\star$.
		And this holds because $X \geq X^\star \geq u^\star$ on $(T,T')$,%
			\footnote{\label{footnote:FL_lower_X}%
			The first inequality holds because for any $t \in (T,T')$,
			\begin{equation*}
				\textstyle
				X_t
				= e^{r(t-T)} X_T - r \int_T^t e^{-r(s-t)} x_s \dd s
				\geq e^{r(t-T)} X_T - r \int_T^t e^{-r(s-t)} u' \dd s
				= X^\star_t ,
			\end{equation*}
			where the inequality holds since $x \leq u'$ a.e. on $(T,\infty)$,
			and the last equality holds because $X_T = X^\star_T$ and $x^\star = u'$ on $[T,T')$.}
		while $X^\star = u^\star$ on $[T',\infty)$.
	\end{proof}%
	\renewcommand{\qedsymbol}{$\blacksquare$}

	It remains only to prove \hyperref[claim:pf_opt:2]{claim 2}.

	\begin{proof}[Proof of {\hyperref[claim:pf_opt:2]{claim 2}}]%
		\renewcommand{\qedsymbol}{$\square$}
		Fix a maximal proper interval $J$ of $\R_+$ prior to $T$
		on which $\phi^0$ is constant,
		and let $\alpha \in \R$ be the value that $\phi^0$ takes on $J$.

		Since $F^0$ is concave,
		the set of $u \in \left[0,u^0\right]$
		at which $\alpha$ is a supergradient of $F^0$
		is an interval $\left[u',u''\right]$,
		where $u^\star \leq u' \leq u'' \leq u^0$.%
			\footnote{$u^\star \leq u'$ obtains since
			$F^0$ is concave and
			$F^{0+}(u') = \alpha = \phi^0 < F^{0-}(u^\star)$ on $J$,
			where the inequality holds since $J$ is prior to $T$.
			As for $u'' \leq u^0$, letting $t \to 0$ in \eqref{eq:euler} (\cref{app:Euler}, \cpageref{eq:euler})
			yields $\lim_{t \to 0} \phi^0(t) = 0$.
			Since $\phi^0$ is increasing and $J$ is a proper interval,
			it follows that $\alpha \geq 0$.
			Thus $F^0$ is increasing on $\left[u',u''\right]$,
			so that $u'' \leq u^0$ by definition of the latter.}
		Since $(x,X)$ satisfies the Euler equation with $\phi^0,\phi^1$
		and $G$ has unbounded support,
		$\alpha$ is a supergradient of $F^0$ at $x_t$ for a.e. $t \in J$.
		This implies that $u' \leq x \leq u''$ a.e. on $J$.

		If $u'=u''$, then we are done:
		$x$ is a.e. constant at $u''=u' \geq u^\star$ on $J$,
		so obviously admits a version that is decreasing.

		Suppose instead that $u'<u''$,
		meaning that $F^0$ has an affine segment with slope $\alpha$
		extending from $u'$ to $u''$.
		We shall front-load the mechanism $(x,X)$ over this affine segment,
		imitating the proof of \Cref{theorem:deadline} (\cpageref{theorem:deadline}).
		In particular, given a deadline $T' \in \cl J$, define
		\begin{equation*}
			x^\star_t \coloneqq
			\begin{cases}
				x_t		& \text{for $t \notin J$} \\
				u''		& \text{for $t \leq T'$ in $J$} \\
				u'		& \text{for $t > T'$ in $J$.}
			\end{cases}
		\end{equation*}
		Since $u' \leq x \leq u''$ a.e. on $J$,
		we may choose the deadline $T' \in \cl J$ so that $X^\star_{\inf J} = X_{\inf J}$.

		We shall show that the front-loaded mechanism $\left( x^\star, X^\star \right)$ dominates $(x,X)$ unless $X^\star = X$.
		This is sufficient because $(x,X)$ is undominated (being optimal for $G$),
		so must then satisfy $x = x^\star$ a.e.;
		and $x^\star$ is decreasing on $J$.

		We have $\pi_t\left(x^\star\right) = \pi_t(x)$ for every $t \notin J$
		since $F^0$ is affine on $[u',u'']$ and $X^\star = X$ off $J$.%
			\footnote{Replicate the payoff-rewriting exercise
			in the sketch proof of \Cref{theorem:deadline} (\cpageref{theorem:deadline}).}
		It remains to show that $\pi_t\left(x^\star\right) \geq \pi_t(x)$ for every $t \in J$, with equality only if $X^\star_t = X_t$.
		Define
		\begin{equation*}
			\psi(u) \coloneqq F^1(u) - \alpha u
			\quad \text{for each $u \in \R_+$.}
		\end{equation*}
		Since $F^0$ is affine with slope $\alpha$ on $[u',u'']$
		and $X^\star_{\inf J} = X_{\inf J}$, we have
		\begin{equation*}
			\pi_t\left(x^\star\right) - \pi_t\left(x\right)
			= e^{-rt}
			\bigl[ \psi\left(X^\star_t\right)
			- \psi\left(X_t\right) \bigr]
			\quad \text{for each $t \in J$.}
		\end{equation*}
		Since $X^\star_{\inf J} = X_{\inf J}$ and $u' \leq x \leq u''$ a.e. on $J$,
		we have $X^\star \leq X$ on $J$.%
			\footnote{The idea is that front-loading lowers $X$ pointwise;
			we saw this in the sketch proof of \Cref{theorem:deadline} (\cpageref{theorem:deadline}) and in \cref{footnote:FL_lower_X}.}
		It therefore suffices to show that $\psi$ is strictly decreasing on $\left[ \inf_J X^\star, \infty \right)$.

		Suppose that $X \geq u'$ on $J$.
		Then $X_{\sup J} \geq u'$ since $X$ is continuous,
		so that $X^\star \geq u'$ on $J$ as well.
		We need thus only show that $\psi$ is strictly decreasing on $[u',\infty)$.
		It is strictly decreasing on $[u',u'']$
		since there we have $\psi = \left(F^1-F^0\right) + k$ for a constant $k \in \R$,
		and $F^1-F^0$ is strictly decreasing on $\left[u^\star,u^0\right] \supseteq [u',u'']$ by definition of $u^\star$.
		Since $\psi$ is concave, it must then be strictly decreasing on all of $[u',\infty)$.

		It remains to consider the case in which $X_s < u'$ for some $s \in J$.
		Write
		\begin{equation*}
			t' \coloneqq \inf J
			\quad \text{and} \quad
			t'' \coloneqq \sup J ,
		\end{equation*}
		noting that $t' < t''$ since $J$ is a proper interval.
		It must be that $t'' < \infty$,
		since otherwise we would have $x \geq u'$ a.e. on $(t',\infty)$
		and thus $X \geq u'$ on $J$.
		Since $X_{t''} \leq X$ on $[s,t'']$,%
			\footnote{If $X_t = \min_{[s,t'']} X$ for $t \in [s,t'']$, then
			since $x \geq u' > X_s \geq X_t$ a.e. on $[s,t'']$, we have
			\vspace{-4pt}%
			\begin{equation*}
				X_t
				= \textstyle r \int_t^{t''} e^{-r(z-t)} x_z \dd z + e^{-r(t''-t)} X_{t''}
				\geq (1-e^{-r(t''-t)}) X_t + e^{-r(t''-t)} X_{t''} .
			\end{equation*}\vspace{-12pt}}
		it suffices to show that $\psi$ is strictly decreasing on $\left[ X_{t''}, \infty \right)$.
		And for this, it is enough that $t \mapsto F^{1+}(X_t) - \alpha$
		be strictly negative at, or arbitrarily close to, $t''$.%
			\footnote{Since then $F^{1+} - \alpha < 0$ on $( X_{t''}, \infty )$,
			as $F^{1+}$ is decreasing.}

		Remark that since $\phi^0 = \alpha$ on $(t',t'')$,
		letting $t \uparrow t''$ in \eqref{eq:euler_forward_general} on \cpageref{eq:euler_forward_general} yields
		\begin{equation*}
			\E_G\left(\phi^1(\tau) \middle| \tau \geq t''\right)
			- \frac{\E_G\left(\phi^1(\tau)\right)}{1-\lim_{t \uparrow t''}G(t)}
			= \alpha .
			\label{eq:euler_forward_limit}
			\tag{$\uparrow$}
		\end{equation*}

		Suppose first that $G$ has an atom at $t''$.
		Then $\phi^1(t'')$ is a supergradient of $F^1$ at $X_{t''}$ since $(x,X)$ satisfies the Euler equation with $\phi^0,\phi^1$.
		Since $F^{1+}(X_{t''}) \leq \phi^1(t'')$ (as $F^1$ is concave),
		it suffices to show that $\phi^1(t'') < \alpha$.
		So suppose toward a contradiction that $\phi^1(t'') \geq \alpha$.
		Then
		\begin{equation*}
			\alpha
			\leq \phi^1(t'')
			\leq \E_G\left(\phi^1(\tau) \middle| \tau \geq t''\right)
			= \alpha
			+ \frac{\E_G\left(\phi^1(\tau)\right)}{1-\lim_{t \uparrow t''}G(t)}
			\leq \alpha
		\end{equation*}
		since $\phi^1$ is increasing (second inequality),
		by \eqref{eq:euler_forward_limit} (the equality)
		and by \eqref{eq:euler_average_general} on \cpageref{eq:euler_average_general} (final inequality).
		It follows that $\phi^1(t'') = \E_G\left( \phi^1(\tau) \middle| \tau \geq t'' \right) = \alpha$,
		so that $\phi^1 = \alpha$ $G$-a.e. on $[t'',\infty)$ since $\phi^1$ is increasing.
		But then $\phi^0 = \alpha$ on $(t',\infty)$ by \eqref{eq:euler_forward_general} on \cpageref{eq:euler_forward_general},
		which contradicts the fact that $t''<\infty$.

		Suppose instead that $G$ has no atom at $t''$.
		Then $t''$ belongs to $J$ since
		\begin{align*}
			\phi^0(t'')
			&= \E_G\left(\phi^1(\tau) \middle| \tau > t''\right) - \frac{\E_G\left(\phi^1(\tau)\right)}{1-G(t'')}
			\\
			&=
			\E_G\left(\phi^1(\tau) \middle| \tau \geq t''\right) - \frac{\E_G\left(\phi^1(\tau)\right)}{1-\lim_{t \uparrow t''}G(t)}
			= \alpha
		\end{align*}
		by \eqref{eq:euler_forward_general} on \cpageref{eq:euler_forward_general} (first equality)
		and \eqref{eq:euler_forward_limit} (last equality).
		Fix any $\eps>0$.
		Since $J$ is a maximal interval of constancy of $\phi^0$
		and $t''$ belongs to $J$,
		$\phi^0$ is not constant on $[t'',t''+\eps)$,
		and thus $[t'',t''+\eps)$ is $G$-non-null by \eqref{eq:euler_forward_general}.
		Since $(x,X)$ satisfies the Euler equation with $\phi^0,\phi^1$,
		it follows that $\phi^1(t)$ is a supergradient of $F^1$ at $X_t$
		for some $t \in [t'',t''+\eps)$.

		Now, since $G$ has no atom at $t''$ and $\phi^1$ is increasing, we must have
		\begin{equation*}
			\lim_{t \downarrow t''} \phi^1(t)
			< \E_G\left( \phi^1(\tau) \middle| \tau \geq t'' \right) ,
		\end{equation*}
		as otherwise $\phi^1$ would be $G$-a.e. constant on $(t'',\infty)$,
		which would contradict $t''<\infty$ by the argument above.
		Thus for $\eps>0$ sufficiently small, we have
		\begin{multline*}
			F^{1+}(X_t) \leq \phi^1(t) \leq \phi^1(t''+\eps)
			\\
			< \E_G\left( \phi^1(\tau) \middle| \tau \geq t'' \right)
			= \alpha
			+ \frac{\E_G\left(\phi^1(\tau)\right)}{1-\lim_{t \uparrow t''}G(t)}
			\leq \alpha
		\end{multline*}
		by the concavity of $F^1$ (first inequality),
		the monotonicity of $\phi^1$ (second inequality),
		\eqref{eq:euler_forward_limit} above (the equality)
		and \eqref{eq:euler_average_general} on \cpageref{eq:euler_average_general} (final inequality).

		Since $\eps>0$ may be chosen arbitrarily small and $t$ belongs to $[t'',t''+\eps)$,
		it follows that $F^{1+}(X_t) - \alpha < 0$ for arbitrarily small $t \geq t''$, as desired.
	\end{proof}%
	\renewcommand{\qedsymbol}{$\blacksquare$}

	With all three claims now established,
	the proof is complete.
\end{proof}

\subsection{Generalisation and proof of \texorpdfstring{\Cref{proposition:opt_transition}}{Proposition~\ref{proposition:opt_transition}} (p.~\pageref{proposition:opt_transition})}
\label{app:opt_transition}

Recall the (superdifferential) Euler equation
defined in \cref{app:Euler} (\cpageref{definition:Euler_general}).

\begin{namedthm}[\Cref*{proposition:opt_transition}$\boldsymbol{'}$.]
	\label{proposition:opt_transition_general}
	Let $G$ be a distribution with unbounded support.
	Any mechanism that is optimal for $G$ satisfies the Euler equation for $G$.
	Any undominated mechanism that satisfies the Euler equation for $G$ is optimal for $G$.
\end{namedthm}

This result refines \Cref{proposition:opt_transition} in two ways:
it provides that the Euler equation is necessary absent any auxiliary assumptions,
and furthermore asserts sufficiency.
To prove it, we shall rely on the \hyperref[lemma:euler]{Euler lemma}
and the \hyperref[corollary:existence]{existence corollary} in \cref{app:Euler} (\cpageref{lemma:euler,corollary:existence}).

\begin{proof}[Proof of {\hyperref[proposition:opt_transition_general]{\Cref*{proposition:opt_transition}$\xslantmath{'}$}}]
	Fix a distribution $G$.
	By \Cref{lemma:lequ0} and \Cref{proposition:indiff} (\cpageref{lemma:lequ0,proposition:indiff}),
	any undominated mechanism has the form $(x,X)$ with $x \in \mathcal{X}$.
	If $(x,X)$ is undominated and satisfies the Euler equation for $G$,
	then it maximises the principal's payoff under $G$ by (the first part of) the \hyperref[lemma:euler]{Euler lemma},
	so is optimal for $G$.
	Conversely, suppose that $(x,X)$ is optimal for $G$.
	By the \hyperref[corollary:existence]{existence corollary},
	there is a(nother) mechanism that satisfies the Euler equation for $G$.
	So by (the second part of) the \hyperref[lemma:euler]{Euler lemma},
	$(x,X)$ satisfies the Euler equation.
\end{proof}

\begin{proof}[Proof of \Cref{proposition:opt_transition}]
	Assume that
	$u^\star > 0$
	and that $F^0,F^1$ are differentiable on $\left( 0, u^0 \right)$,
	and let $(x,X)$ be optimal for a distribution $G$ with $G(0)=0$ and unbounded support.
	Then $x$ is decreasing with $0 < u^\star \leq X \leq x \leq u^0$
	and $\lim_{t \to \infty} x_t = u^\star \leq u^1 < u^0$
	by \Cref{theorem:opt} (\cpageref{theorem:opt}),
	and $(x,X)$ satisfies the Euler equation by \hyperref[proposition:opt_transition_general]{\Cref*{proposition:opt_transition}$'$}.

	Thus $0 < X < u^0$, so that $F^1$ is differentiable at $X_t$ for every $t \in \R_+$.
	Similarly, $F^0$ is differentiable at $x_t$ for every $t \in \R_+$ at which $x_t < u^0$.
	Hence by \Cref{observation:euler_forward} in \cref{app:Euler} (\cpageref{observation:euler_forward}), the Euler equation implies that
	$\E_G\left( F^{1\prime}\left( X_\tau \right) \right) = 0$ and
	\begin{equation*}
		F^{0\prime}(x_t)
		= \E_G\left( F^{1\prime}(X_\tau) \middle| \tau > t \right)
		\quad \text{for a.e. $t \in \R_+$ with $x_t < u^0$,}
	\end{equation*}
	and furthermore that
	\begin{equation*}
		F^{0-}(x_t)
		\geq \E_G\left( F^{1\prime}(X_\tau) \middle| \tau > t \right)
		\quad \text{for a.e. $t \in \R_+$ with $x_t = u^0$,}
	\end{equation*}
	since the left-hand derivative $F^{0-}\left( u^0 \right)$
	is the largest supergradient at $u^0$ of the concave function $F^0$.
	For any right-continuous version of $x$,%
		\footnote{E.g. $\widetilde{x}$ given by $\widetilde{x}_t = \sup_{s > t} x_s$ for each $t \in \R_+$.}
	the above (in)equalities must hold for \emph{every} $t \in \R_+$,
	since then both sides are right-continuous in $t$.%
		\footnote{The right-hand side is right-continuous in $t$ because
		$G$ is right-continuous and $\phi^1(s) \coloneqq F^{1\prime}(X_s)$ is $G$-integrable,
		so that for $t_n \downarrow t$ we have
		$G(t_n) \to G(t)$
		and (by dominated convergence) $\vphantom{\int_{\R_+}}\smash{\int_{\R_+} \phi^1 \1_{(t_n,\infty)} \dd G
		\to \int_{\R_+} \phi^1 \1_{(t,\infty)} \dd G}$.}
\end{proof}

\hyperref[proposition:opt_transition_general]{\Cref*{proposition:opt_transition}$'$}
implies the assertion made in \cref{footnote:X1_u1} on \cpageref{footnote:X1_u1}:

\begin{corollary}
	\label{corollary:opt_X1_u1}
	If $(x,X)$ is optimal for a distribution $G$ with $G(0)=0$ and unbounded support,
	then $X_0 > u^1$.
\end{corollary}

\begin{proof}
	If $u^\star = u^1$,
	then $X_0 > u^\star = u^1$
	by \Cref{theorem:opt} (\cpageref{theorem:opt}).
	Assume for the remainder that $u^\star < u^1$,
	and suppose toward a contradiction that $X_0 \leq u^1$.
	Then $X_t < u^1$ for all $t>0$
	since $X$ is decreasing with $\lim_{t \to \infty} X_t = u^\star < u^1$
	by \Cref{theorem:opt} (\cpageref{theorem:opt}),
	and thus $X < u^1$ $G$-a.e. since $G(0)=0$.
	Since $F^1$ is strictly increasing and concave on $\left[0,u^1\right]$,
	it follows that $F^{1+}(X) > 0$ $G$-a.e.

	$(x,X)$ satisfies the Euler equation with some $\phi^0,\phi^1$ by \hyperref[proposition:opt_transition_general]{\Cref*{proposition:opt_transition}$'$},
	so $\phi^1(t)$ is a supergradient of $F^1$ at $X_t$ for $G$-a.e. $t \in \R_+$,
	equation \eqref{eq:euler} (\cpageref{eq:euler}) holds,
	and $\phi^0$ is non-negative.
	Thus for any $t \in \R_+$ with $G(t)>0$, we have
	\begin{equation*}
		0
		< \int_{[0,t]} F^{1+}(X_s) G( \dd s )
		\leq \int_{[0,t]} \phi^1 \dd G
		= - [1-G(t)] \phi^0(t)
		\leq 0 ,
	\end{equation*}
	which is absurd.
\end{proof}

\crefalias{section}{supplsec}
\crefalias{subsection}{supplsec}
\crefalias{subsubsection}{supplsec}
\section*{Supplemental appendices}
\label{suppl}
\addcontentsline{toc}{section}{Supplemental appendices}

\subsection{Details for the discussion of assumptions (§\ref{sec:model:discussion})}
\label{suppl:ext}

In this appendix, we provide the details underlying the discussion in §\ref{sec:model:discussion} of our model assumptions.

\subsubsection{If \texorpdfstring{$\boldsymbol{u^\star}$}{u*} is not a strict local maximum}
\label{suppl:ext:drop_F1-F0_strict}

Our assumption that $u^\star$ is a strict local maximum of $F^1-F^0$
requires merely that $u^\star$ be a strict local maximum on $\left[0,u^\star\right]$,
as the same is true on $\left[u^\star,u^0\right]$ by definition of $u^\star$.
This holds vacuously if $u^\star=0$,
while if $u^\star>0$ it amounts essentially to ruling out a saddle point.%
	\footnote{\label{footnote:trichotomy}%
	Precisely, $u^\star$ must be either a local maximum, a saddle point, or a point at which \emph{both} $F^0$ \emph{and} $F^1$ have a kink.
	We omit the details; see \textcite{omit}.}

In fact, nothing changes if we weaken our assumption that $u^\star$ is a strict local maximum of $F^1-F^0$ to demand only that there be no proper interval $\left[ u_\star, u^\star \right] \subseteq \left[0,u^0\right]$ on which $F^0,F^1$ are affine with equal slopes.
Dropping this weaker assumption merely generates some uninteresting multiplicity.
For concreteness, consider the case in which $F^0$ is affine, so that $F^1-F^0$ is concave and thus attains its maximum over $\left[ 0, u^0 \right]$ on an interval $\left[ u_\star, u^\star \right]$.

\begin{definition}
	\label{definition:interval_deadline}
	A mechanism $(x,X)$ is an \emph{interval deadline mechanism} iff
	for some $T \in [0,\infty]$, we have
	$x_t = u^0$ for $t \leq T$
	and $x_t \in \left[ u_\star, u^\star \right]$ for $t > T$.
\end{definition}

With small alterations, the proof of \Cref{theorem:deadline} in \cref{app:pf_theorem_deadline} delivers

\begin{namedthm}[\Cref*{theorem:deadline}$\boldsymbol{'}$.]
	\label{namedthm:deadline_interval}
	If the old frontier $F^0$ is affine on $\left[ 0, u^0 \right]$,
	then any undominated mechanism is an interval deadline mechanism.
\end{namedthm}

\subsubsection{If some agent utility levels are infeasible}
\label{suppl:ext:diff_domain}

Our model does not require that every agent utility level $u \in [0,\infty)$ be feasible.
Concretely, suppose that technology $j \in \{0,1\}$ can only provide the agent with utility in an interval $I^j \subseteq [0,\infty)$.%
	\footnote{$I^j$ is necessarily an interval
	because any convex combination of feasible utility levels can be attained by rapidly switching back and forth (or randomising).}

The frontier $F^j$ is a concave and upper semi-continuous function $I^j \to \R$.
(Recall that these assumptions are without loss.)
It is innocuous to extend $F^j$ continuously to $\cl I^j$.%
	\footnote{The principal can anyway attain utility arbitrarily close to $\lim_{u \downarrow \inf I^j} F^j(u)$ by choosing $u > \inf I^j$ small, and similarly for $\sup I^j$.}
(Note that $F^j$ may take the value $-\infty$ off $I^j$.)
Assume that $F^0$ has a unique peak $u^0 \in \cl I^0$.
Assume without loss of generality that
(i)~$\left[0,u^0\right] \subseteq \cl I^0$,%
	\footnote{Any mechanism $(x^0,X^1)$ satisfies $x^0 \geq \inf I^0$
	since utilities $<\inf I^0$ cannot be reached using the old technology.
	Thus IC mechanisms $(x^0,X^1)$ have $X^1 \geq X^0 \geq \inf I^0$.
	So without loss, we may consider the translated model with agent utility $\widetilde{u} \coloneqq u - \inf I^0 \in [0,\infty)$.}
(so that $F^0$ is finite on $\left(0,u^0\right]$,)
and
(ii)~$I^0 \subseteq I^1$.%
	\footnote{The new technology expands the set of available physical allocations,
	so any agent utility feasible before the breakthrough remains feasible afterwards.}

We now impose the remaining model assumptions.
First, $u^0>0$.
Secondly, $F^1$ has a unique peak $u^1 \in \cl I^1$,
which satisfies $u^1 < u^0$.
Thirdly, $F^1 \geq F^0$ (without loss, recall).
Finally, $u^\star$ is a strict local maximum of $F^1-F^0$.

Extend $F^j$ to all of $[0,\infty)$
by letting $F^j \coloneqq -\infty$ off $\cl I^j$.
Then $F^0,F^1$ satisfy our model assumptions.
Since utility levels at which $F^j=-\infty$ are never chosen when using technology $j$,
it is as if they were not feasible.

\subsubsection{Participation constraint instead of non-negativity}
\label{suppl:ext:participation}

Suppose that the agent's utility can take any value $u \in [-K,\infty)$, where $K>0$ is (arbitrarily) large.%
	\footnote{The lower bound does not bind.
	We impose it merely to avoid integrability issues.}
The agent can quit anytime, earning a continuation payoff worth zero (a normalisation).
We focus on the interesting case in which the principal prefers for the agent never to quit, and therefore chooses among participation-inducing IC mechanisms.

The frontiers $F^0,F^1$ are now defined on $[-K,\infty)$.
As in the text, $u^\star$ denotes the largest $u \in \left[ 0, u^0 \right]$ at which $F^0,F^1$ have equal slopes, with $u^\star \coloneqq 0$ if there is no such $u$.
Note well that $u^\star$ is non-negative by definition.

\begin{namedthm}[Claim.]
	\label{claim:participation}
	All of our results remain valid
	(with $u^\star$ defined as above).
\end{namedthm}

\begin{proof}
	Consider the formally equivalent model in which the agent's utility is $\widetilde{u} \coloneqq u+K \in [0,\infty)$,
	with frontiers $\widetilde{F}^j\left( \widetilde{u} \right) \coloneqq F^j\left( \widetilde{u} - K \right)$
	peaking at $\widetilde{u}^j \coloneqq u^j + K$.
	Let $\widetilde{u}^\star$ be the largest $\widetilde{u} \in \left[ 0, \widetilde{u}^0 \right]$ at which $\widetilde{F}^0,\widetilde{F}^1$ have equal slopes,
	with $\widetilde{u}^\star \coloneqq 0$ if there is no such $\widetilde{u}$.
	It need \emph{not} be that $\widetilde{u}^\star = u^\star + K$:
	rather, this holds iff $\widetilde{u}^\star \geq K$.%
		\footnote{If $\widetilde{u}^\star < K$, then $\widetilde{u}^\star < K \leq u^\star + K$.
		Conversely, if $\widetilde{u}^\star \geq K$, then $\widetilde{u}^\star$ is a fortiori the largest $\widetilde{u} \in [ K, \widetilde{u}^0 ]$ at which $\widetilde{F}^0,\widetilde{F}^1$ have equal slopes,
		which is to say that $\widetilde{u}^\star-K$ is the largest $u \in [ 0, u^0 ]$ at which $F^0,F^1$ have equal slopes,
		which is the definition of $u^\star$.}
	We next argue that this may be assumed without loss of generality.

	The participation constraints read
	\begin{equation*}
		\widetilde{X}^1_t \geq K
		\quad \text{and} \quad
		\widetilde{X}^0_t
		+ \E_G\left( e^{-r(\tau-t)} \left( \widetilde{X}^1_\tau - \widetilde{X}^0_\tau \right)
		\middle| \tau > t \right) \geq K
		\quad \text{for all $t \in \R_+$.}
		\label{eq:participation}
		\tag*{}
	\end{equation*}
	Due to the first constraint, it is immaterial what values the new frontier $\widetilde{F}^1$ takes on $[0,K)$.
	So assume without loss that it equals the concave upper envelope of
	$\1_{[0,K)} \widetilde{F}^0 + \1_{[K,\infty)} \widetilde{F}^1$.
	Then $\widetilde{F}^1$ is weakly steeper than $\widetilde{F}^0$ on $[0,K)$,%
		\footnote{The greatest supergradient $\widetilde{F}^1$ weakly exceeds that of $\widetilde{F}^0$, and likewise for the smallest.}
	so that $\widetilde{u}^\star \geq K$
	and thus $\widetilde{u}^\star = u^\star+K$.

	The principal's problem is as in the text, except that she must respect the \hyperref[eq:participation]{participation constraints}.
	We now show that these do not bind.

	First, when $F^0$ is affine on $\left[0,u^0\right]$,
	any undominated mechanism $( \widetilde{x}^0, \widetilde{X}^1 )$ in the relaxed problem that ignores the \hyperref[eq:participation]{participation constraints} (i.e. the problem in the text)
	satisfies $\widetilde{X}^1 = \widetilde{X}^0 \geq \widetilde{u}^\star \geq K$
	by \Cref{proposition:indiff,theorem:deadline} (\cpageref{proposition:indiff,theorem:deadline}).
	This implies the \hyperref[eq:participation]{participation constraints}.
	Thus undominated (optimal) mechanisms are characterised, in $\widetilde{u}$ units, by \Cref{theorem:deadline} and \Cref{proposition:dl_charac} (by \Cref{proposition:opt_deadline}).

	Similarly,
	ignoring \hyperref[eq:participation]{participation},
	any mechanism $( \widetilde{x}^0, \widetilde{X}^1 )$ that is optimal for a distribution $G$ with $G(0)=0$ and unbounded support
	satisfies $\widetilde{X}^1 = \widetilde{X}^0 \geq \widetilde{u}^\star \geq K$
	by \Cref{proposition:indiff} and \Cref{theorem:opt} (\cpageref{proposition:indiff,theorem:opt}),
	so that the \hyperref[eq:participation]{participation constraints} hold.
	Thus \Cref{theorem:opt} and \Cref{proposition:opt_transition} characterise optimal mechanisms, in $\widetilde{u}$ units.

	These characterisations translate straightforwardly back to $u$ units, except for one wrinkle:
	the long-run utility level appearing in \Cref{theorem:deadline,theorem:opt} is $\widetilde{u}^\star - K$, not $u^\star$.
	We showed, however, that these two are equal.
\end{proof}

\subsubsection{Uncertain technology}
\label{suppl:ext:random_F1}

In our model, the new technology $F^1$ is known in advance---only its date of arrival is uncertain. In this appendix, we show that all of our results remain valid if the new technology is uncertain, provided the agent is not privately informed about its realisation.

Let $\mathcal{F}$ be a finite set of of concave and upper semi-continuous functions $[0,\infty) \to [-\infty,\infty)$ with unique peaks.
The new frontier $\boldsymbol{F}$ is a random element of $\mathcal{F}$, drawn independently of the breakthrough time $\tau$.
Write $U^1(F)$ for the unique peak of $F \in \mathcal{F}$,
and $u^1 \coloneqq \E\left( U^1(\boldsymbol{F}) \right)$ for its expectation.
We assume that there is a conflict of interest: $u^1 < u^0$.

The agent privately observes when the breakthrough occurs, but she does not learn the realised value of the new technology $\boldsymbol{F}$.
This means that the agent cannot easily determine the payoff consequences for the principal of the new technology, which is natural in many (but not all) applications.

A mechanism specifies, for each period $t$,
the agent's utility $x^0_t$ if she has not already disclosed,
as well as the continuation utility $\widehat{X}_t(F)$ with which she is rewarded for disclosing at time $t$ if the realised new technology is $F \in \mathcal{F}$.
Since the agent does not know $F$ prior to disclosure, only the expectation
$X^1_t \coloneqq \E\bigl( \widehat{X}_t(\boldsymbol{F}) \bigr)$
matters for her incentives.

For a given value $X^1_t = u$ of this expectation,
the principal chooses $\widehat{X}_t : \mathcal{F} \to [0,\infty)$ to maximise
$\E\bigl( \boldsymbol{F}\bigl( \widehat{X}_t(\boldsymbol{F}) \bigr) \bigr)$
subject to $\E\bigl( \widehat{X}_t(\boldsymbol{F}) \bigr) = u$.
We write $F^1(u)$ for the value of this problem.%
	\footnote{A maximum exists (so that $F^1$ is well-defined)
	because the constraint set is compact in the pointwise topology (being a closed and bounded subset of the Euclidean space $[0,\infty)^{\abs*{\mathcal{F}}}$)
	and the maximand is upper semi-continuous since every element of $\mathcal{F}$ is.}

To characterise the pre-disclosure flow $x^0$ and expected disclosure reward $X^1$ in undominated mechanisms, we may study the deterministic model in which the new technology is $F^1$.
(The technology-contingent disclosure reward $\widehat{X}$ may be backed out from the above maximisation problem.)
This deterministic model satisfies our model assumptions:

\begin{lemma}
	\label{lemma:randomF1_assns}
	$F^1$ is concave and upper semi-continuous, with unique peak at $u^1 = \E\left( U^1\left( \boldsymbol{F} \right) \right)$.
\end{lemma}

Our results therefore remain valid, characterising the $x^0$ and $X^1$ of undominated mechanisms in the uncertain-technology model.
We omit the (straightforward) proof of \Cref{lemma:randomF1_assns}
\parencite[see][]{omit}.

\subsection{Revelation principle}
\label{suppl:revelation_principle}

A revelation principle for our environment must account for the verifiability of the agent's disclosures.
A \emph{direct mechanism} is one which solicits a cheap-talk report of the breakthrough's arrival, then instructs the agent when to deliver her hard evidence (her verifiable disclosure).
The standard revelation principle \parencite[Proposition 2]{Myerson1982} permits us to restrict attention to incentive-compatible direct mechanisms, meaning those in which the agent is willing to report promptly and to deliver her evidence at the appointed time.

It remains only to show that among such mechanisms, we may further restrict our attention to those involving prompt delivery of the evidence.
Modulo differences in detail, this follows from Bull and Watson's (\citeyear{BullWatson2007}) revelation principle (their Theorem~2).
The key requirement for their result, the `normality' of evidence,
is satisfied in our model: for each type of the agent (i.e. breakthrough time), there is a most-informative manner of verifiably disclosing: namely, disclosing promptly.

\subsection{Existence of undominated and optimal mechanisms}
\label{suppl:undom_opt_properties}

In this appendix, we prove that any dominated mechanism is dominated by an undominated mechanism (\Cref{proposition:dominated_by_undominated} below). This claim is used directly in some of our proofs, and it implies the existence of undominated and optimal mechanisms (\Cref{corollary:optimal_existence} below). Although the claim is economically straightforward, its proof involves some topological tedium.

Observe first that the proof of \Cref{lemma:lequ0} (\cpageref{lemma:lequ0}), which relies on no existence claim, implies the following slightly stronger result:

\begin{namedthm}[\Cref{lemma:lequ0}$\boldsymbol{'}$.]
	\label{namedthm:lequ0}
	Any dominated IC mechanism $\left(x^0,X^1\right)$ is dominated by an IC mechanism $\left(x^{0\star},X^{1\star}\right)$ with $x^{0\star} \leq u^0$.
\end{namedthm}

Besides this lemma, we shall rely on some standard results which may be found in \textcite{Warga1972}.
Given $\ell \in \N$ and a compact set $K \subseteq \R^\ell$, let $\mathcal{P}_K$ be the space of all probability measures on $K$, endowed with the topology of weak convergence.
Let $\mathcal{F}_K$ be the set of all maps $f : \R_+ \times K \to \R$ such that $f(\cdot,z)$ is Borel measurable for each $z \in K$, $f(t,\cdot)$ is continuous for each $t \in \R_+$, and $\int_{\R_+} \max_{z \in K}|f(t,z)| \dd t < \infty$.
Endow $\mathcal{P}_K$ with the Borel $\sigma$-algebra.
The following is immediate from Theorem~IV.1.6 in \textcite{Warga1972}.

\begin{lemma}
	\label{lemma:S_K_equivalent_defns}
	For any map $\nu : \R_+ \to \mathcal{P}_K$, the following are equivalent:
	(a)~$\nu$ is Borel measurable.
	(b)~For any continuous $\phi : K \to \R$, $t \mapsto \int_K \phi \dd \nu_t$ is Borel measurable.
	(c)~For any $f \in \mathcal{F}_K$, $t \mapsto \int_K f(t,z) \nu_t(\dd z)$ is (Borel) integrable.
\end{lemma}

Let $\mathcal{S}_K$ be space of Lebesgue measurable maps $\nu : \R_+ \to \mathcal{P}_K$, where we identify any two maps that are a.e. equal. The following is immediate from Theorem~IV.2.1 in \textcite{Warga1972}.

\begin{lemma}
	\label{lemma:SK_seq_compact}
	$\mathcal{S}_K$ is sequentially compact in the weak topology.
	That is, for any sequence $(\nu^n)_{n \in \N}$ in $\mathcal{S}_K$, there is a $\nu \in \mathcal{S}_K$
	such that
	\begin{equation*}
		\int_{\R_+} \int_K f(t,z)\nu^n_t(\dd z) \dd t
		\to
		\int_{\R_+} \int_K f(t,z)\nu_t(\dd z) \dd t
		\quad \text{for every $f \in \mathcal{F}_K$}
		\label{eq:seq_compact_convergence}
		\tag{$\dag$}
	\end{equation*}
	along some subsequence of $(\nu^n)_{n \in \N}$.
\end{lemma}

\Cref{lemma:S_K_equivalent_defns,lemma:SK_seq_compact} imply, by standard arguments, that the space of IC mechanisms is sequentially compact and that the principal's payoff $\Pi_G$ is upper semi-continuous:%
	\footnote{For sequential compactness, Helly's selection theorem \parencite[e.g.][p. 167]{Rudin1976} and a diagonalisation argument yield a subsequence along which $(X^{1n})_{n \in \N}$ converges pointwise, and \Cref{lemma:SK_seq_compact} yields a sub-subsequence along which $(x^{0n})_{n \in \N}$ converges weakly. Upper semi-continuity is established using monotone convergence and Fatou's lemma.}

\begin{lemma}
	\label{lemma:Pi_usc}
	For any sequence $\left(x^{0n},X^{1n}\right)_{n \in \N}$ of IC mechanisms, there exists an IC mechanism $\left(x^0,X^1\right)$ such that
	$X^{0n} \to X^0$ and $X^{1n} \to X^1$ pointwise along some subsequence, and
	\begin{equation*}
		\Pi_G\left(x^0,X^1\right) \geq \limsup_{n \to \infty} \Pi_G\left(x^{0n},X^{1n}\right) \quad \text{for any distribution $G$.}
	\end{equation*}
\end{lemma}

Now, fix an arbitrary dominated IC mechanism $\left(x^{0\star},X^{1\star}\right)$, and let $\mathcal{X}^\uparrow$ be the set comprising $\left(x^{0\star},X^{1\star}\right)$ and all IC mechanisms that dominate it. Let $G$ be some full-support distribution with $G(0)>0$ such that $\sup_{\mathcal{X}^\uparrow} \Pi_G$ is finite. By \Cref{lemma:Pi_usc} applied to a sequence $\left(x^{0n},X^{1n}\right)_{n \in \N}$ in $\mathcal{X}^\uparrow$ along which $\Pi_G\left(x^{0n},X^{1n}\right) \to \sup_{\mathcal{X}^\uparrow}\Pi_G$, there is an $\left(x^0,X^1\right) \in \mathcal{X}^\uparrow$ that maximises $\Pi_G$ on $\mathcal{X}^\uparrow$. Given that $G$ has full support, $\left(x^0,X^1\right)$ is obviously undominated. We have shown:

\begin{proposition}
	\label{proposition:dominated_by_undominated}
	Any dominated IC mechanism is dominated by an undominated IC mechanism.
\end{proposition}

\begin{corollary}
	\label{corollary:optimal_existence}
	For any distribution $G$, an optimal mechanism exists.
\end{corollary}

\begin{proof}
	Fix a distribution $G$, and let $\mathcal{Y}$ denote the set of all IC mechanisms.
	Applying \Cref{lemma:Pi_usc} to a sequence $\left(x^{0n},X^{1n}\right)_{n \in \N}$ in $\mathcal{Y}$ along which $\Pi_G\left(x^{0n},X^{1n}\right) \to \sup_{\mathcal{Y}}\Pi_G$ yields the existence of an $\left(x^0,X^1\right) \in \mathcal{Y}$ that maximises $\Pi_G$ on $\mathcal{Y}$.
	By \Cref{proposition:dominated_by_undominated}, we may choose $\left(x^0,X^1\right)$ to be undominated.
	Hence $\left(x^0,X^1\right)$ is optimal for $G$.
\end{proof}

\subsection{Proof of \texorpdfstring{\Cref{lemma:max_effort_obedient}}{Lemma~\ref{lemma:max_effort_obedient}} and \texorpdfstring{\hyperref[proposition:indiff_moralhazard]{\Cref*{proposition:indiff}$\boldsymbol{'}$}}{Theorem~\ref{proposition:indiff}'} (§\ref{sec:mh})}
\label{suppl:moralhazard}

In §\ref{app:moralhazard:lemmata}, we state some lemmata and use them to prove \Cref{lemma:max_effort_obedient} and \hyperref[proposition:indiff_moralhazard]{\Cref*{proposition:indiff}$'$}.
The proof of one of the lemmata is deferred to §\ref{app:moralhazard:lemma_pf}.

\subsubsection{Proof using lemmata}
\label{app:moralhazard:lemmata}

An IC triplet $\left(x^0,X^1,a\right)$
is obedient iff the effort schedule $a$ maximises the agent's payoff
\begin{align*}
	a^\dag
	&\mapsto \int_0^\infty e^{-rt}
	\left( \lambda a^\dag_t e^{-\lambda \int_0^t a^\dag} X^1_t
	+ e^{-\lambda \int_0^t a^\dag}
	r \left[ x^0_t - c a^\dag_t \right] \right)
	\dd t 
	\\
	&= \int_0^\infty e^{-rt -\lambda \int_0^t a^\dag}
	\left( a^\dag_t
	\left[ \lambda X^1_t - rc \right]
	+ r x^0_t \right)
	\dd t .
\end{align*}
This requires precisely that for a.e. $t \in \R_+$,
$\lambda \left(X^1_t - v_t\right) \geq \mathrel{(\leq)} rc$ whenever $a_t=1$ ($a_t=0$), where $v$ denotes the agent's optimal value function under $\left(x^0,X^1\right)$.%
	\footnote{The HJB equation reads
	$r v_t = \dot v_t + rx^0_t
	+ \max_{\alpha \in \{0,1\}}
	\alpha \left[\lambda \left(X^1_t - v_t\right) - rc \right]$.}
Say that a triplet $\left(x^0,X^1,a\right)$ is \emph{quasi-obedient} iff 
\begin{equation*}
	\lambda \left(X^1_t - X^0_t\right) \geq rc \quad \text{for a.e. $t \in \R_+$ such that $a_t = 1$.}
\end{equation*}
Obedience implies quasi-obedience, since $v \geq X^0$.

Quasi-obedience is a useful way of relaxing obedience for two reasons. Firstly, it sometimes captures all of the obedience constrains that \emph{matter} (\Cref{observation:obedience_X1} below). Secondly, it is tractable, since it does not depend on the optimal value function $v$. This tractability permits us to prove that continual effort is optimal in the relaxation of the principal's problem in which obedience is weakened to quasi-obedience (\Cref{lemma:max_effort_quasiobedient} below). It also allows some arguments to be recycled (\hyperref[observation:quasiobedience_u0]{\Cref*{lemma:lequ0}q} and \hyperref[lemma:quasiobedience_X1]{\Cref*{proposition:indiff}q} below).

Call a mechanism $\left(x^0,X^1\right)$ \emph{quasi-obedient} iff the triplet $\left(x^0,X^1,1\right)$ is quasi-obedient.
An IC and quasi-obedient mechanism $\left(x^0,X^1\right)$ is \emph{q-un\-domi\-na\-ted} iff it is not dominated by any IC and quasi-obedient mechanism.

\begin{namedthm}[\Cref*{lemma:lequ0}q.]
	\label{observation:quasiobedience_u0}
	Any q-undominated, IC and quasi-obedient mechanism $\left(x^0,X^1\right)$ satisfies $x^0 \leq u^0$.	
\end{namedthm}

This follows immediately from the logic of \Cref{lemma:lequ0},
noting that quasi-obedience (like IC) is preserved when $x^0$ is lowered pointwise.

\begin{namedthm}[\Cref*{proposition:indiff}q.]
	\label{lemma:quasiobedience_X1}
	Any q-undominated, IC and quasi-obedient mechanism $\left(x^0,X^1\right)$ satisfies $X^1 = X^0 + rc/\lambda$.	
\end{namedthm}

The proof of \hyperref[lemma:quasiobedience_X1]{\Cref*{proposition:indiff}q} follows that of \Cref{proposition:indiff}; we omit it.

\begin{observation}
	\label{observation:obedience_X1}
	If a mechanism $\left(x^0,X^0+rc/\lambda\right)$
	is IC and has $x^0$ bounded,
	then it is obedient.
\end{observation}

\begin{proof}
	An IC mechanism $\left( x^0, X^1 \right)$ is obedient iff $\lambda \left(X^1 - v\right) \geq rc$.
	It therefore suffices to show that an IC mechanism $\left(x^0,X^0+rc/\lambda\right)$ with $x^0$ bounded
	has $v=X^0$,
	since then quasi-obedience implies obedience.

	Given any IC mechanism $\left( x^0, X^1 \right)$,
	a standard verification theorem
	provides that if a bounded and absolutely continuous function $w : \R_+ \to \R$
	satisfies the HJB equation
	\begin{equation*}
		r w_t = \dot w_t + rx^0_t
		+ \max_{\alpha \in \{0,1\}}
		\alpha \left[\lambda \left(X^1_t - w_t\right) - rc\right]
		\quad \text{for a.e. $t \in \R_+$,}
	\end{equation*}
	then $w=v$.
	For the IC mechanism $\left(x^0,X^0+rc/\lambda\right)$,
	the function $X^0$ is absolutely continuous and bounded since $x^0$ is bounded,
	and it satisfies the HJB equation;
	so $v=X^0$.
\end{proof}

\begin{lemma}	
	\label{lemma:max_effort_quasiobedient}
	There exists an IC and quasi-obedient mechanism $\bigl(x^{0\dag},X^{1\dag}\bigr)$ such that $\Pi_{G^1}\bigl(x^{0\dag},X^{1\dag}\bigr) \geq \Pi_{G^a}\left(x^0,X^1\right)$ for any IC and quasi-obedient triplet $\left(x^0,X^1,a\right)$.
\end{lemma}

In other words, continual effort is optimal in the relaxation of the principal's problem in which obedience is replaced by quasi-obedience. The (long) proof of \Cref{lemma:max_effort_quasiobedient} is given §\ref{app:moralhazard:lemma_pf} below.

\begin{proof}[Proof of \Cref{lemma:max_effort_obedient}]
	By \Cref{lemma:max_effort_quasiobedient}, there exists an IC and quasi-obedient mechanism $\bigl(x^{0\dag},X^{1\dag}\bigr)$ such that $\Pi_{G^1}\bigl(x^{0\dag},X^{1\dag}\bigr) \geq \Pi_{G^a}\left(x^0,X^1\right)$ for any IC and quasi-obedient triplet $\left(x^0,X^1,a\right)$.
	In particular, the inequality holds for any obedient triplet $\left(x^0,X^1,a\right)$.
	By a straightforward adaptation of \Cref{proposition:dominated_by_undominated} in \cref{suppl:undom_opt_properties},
	we may choose $\bigl( x^{0\dag}, X^{1\dag} \bigr)$ to be q-undominated.
	We have that $x^{0\dag}$ is bounded by \hyperref[observation:quasiobedience_u0]{\Cref*{lemma:lequ0}q} and that $X^{1\dag} = X^{0\dag} + rc/\lambda$ by \hyperref[lemma:quasiobedience_X1]{\Cref*{proposition:indiff}q};
	thus $\bigl( x^{0\dag}, X^{1\dag} \bigr)$ is obedient by \Cref{observation:obedience_X1}.
\end{proof}

\begin{proof}[Proof of {\hyperref[proposition:indiff_moralhazard]{\Cref*{proposition:indiff}$\xslantmath{'}$}}]
	We prove the contrapositive:
	we fix an IC and obedient mechanism $\left( x^0, X^1 \right)$
	which violates either $x^0 \leq u^0$ or $X^1 = X^0 + rc/\lambda$,
	and show that it must be dominated by some IC and obedient mechanism $\bigl( x^{0\dag}, X^{1\dag} \bigr)$.
	Since $\left( x^0, X^1 \right)$ is IC and obedient, it is quasi-obedient.
	So by \hyperref[observation:quasiobedience_u0]{\Cref*{lemma:lequ0}q} and \hyperref[lemma:quasiobedience_X1]{\Cref*{proposition:indiff}q}, $\left( x^0, X^1 \right)$ is dominated by some IC and quasi-obedient mechanism $\bigl( x^{0\dag}, X^{1\dag} \bigr)$,
	and (by a straightforward adaptation of the argument for \Cref{proposition:dominated_by_undominated} in \cref{suppl:undom_opt_properties})
	we may choose $\bigl( x^{0\dag}, X^{1\dag} \bigr)$ to be q-undominated.
	We have that $x^{0\dag}$ is bounded by \hyperref[observation:quasiobedience_u0]{\Cref*{lemma:lequ0}q} and that $X^{1\dag} = X^{0\dag} + rc/\lambda$ by \hyperref[lemma:quasiobedience_X1]{\Cref*{proposition:indiff}q};
	thus $\bigl( x^{0\dag}, X^{1\dag} \bigr)$ is obedient by \Cref{observation:obedience_X1}.
\end{proof}

\subsubsection{Proof of \texorpdfstring{\Cref{lemma:max_effort_quasiobedient}}{Lemma~\ref{lemma:max_effort_quasiobedient}}}
\label{app:moralhazard:lemma_pf}

A \emph{mixed} effort schedule is a measurable map $a : \R_+ \to [0,1]$.
Extend the definition of $G^a$ to mixed effort schedules via $G^a(t) \coloneqq 1 - \exp\bigl(-\lambda \int_0^t a \bigr)$.
For a mechanism $\left(x^0,X^1\right)$ and a mixed effort schedule $a$, the \emph{continuation payoff} of $\left(x^0,X^1,a\right)$ is the map $C : \R_+ \to [-\infty,\infty)$ such that, for each $t \in \R_+$, $C_t$ is the principal's continuation payoff from period $t$ conditional on no breakthrough in $[0,t]$, given that the distribution of the arrival time is $G^a$ and the agent discloses promptly.

\begin{observation}
	\label{observation:improvement}
	Fix a mechanism $\left(x^0,X^1\right)$ and mixed effort schedules $a$ and $a^\dag$ such that $\Pi_{G^a}\left(x^0,X^1\right) \in \R$, and let $C$ be the continuation payoff of $\left(x^0,X^1,a\right)$.
	If $\bigl(a^\dag-a\bigr)\left[F^1\left(X^1\right)-C \right]$
	is non-negative a.e. (and strictly positive on a non-null set of times), 
	then $\Pi_{G^{a^\dag}}\left(x^0,X^1\right) \geq \mathrel{(>)} \Pi_{G^a}\left(x^0,X^1\right)$. 
\end{observation}

\begin{lemma}
	\label{lemma:maximal}
	There exists an IC and quasi-obedient triplet $\bigl(x^{0\dag},X^{1\dag},a^\dag\bigr)$ such that $x^{0\dag} \leq u^0$ and $\Pi_{G^{a^\dag}}\bigl(x^{0\dag},X^{1\dag}\bigr) \geq \Pi_{G^a}\left(x^0,X^1\right)$ for any IC and quasi-obedient triplet $\left(x^0,X^1,a\right)$.
\end{lemma}

What \Cref{lemma:maximal} asserts is firstly that the restriction $x^{0\dag} \leq u^0$ is without loss of optimality, and secondly that there exists a \emph{best} IC and quasi-obedient triplet with this property. The former claim follows from the logic of \Cref{lemma:lequ0} (\cpageref{lemma:lequ0}). The latter claim is merely technical; it is routine but tedious to verify, using \Cref{lemma:S_K_equivalent_defns,lemma:SK_seq_compact} from \cref{suppl:undom_opt_properties} together with \Cref{observation:improvement}. We omit the details.

\begin{proof}[Proof of \Cref{lemma:max_effort_quasiobedient}]
	By \Cref{lemma:maximal}, there is an IC and quasi-obedient triplet $\left(x^0,X^1,a\right)$ such that $x^0 \leq u^0$ and $\Pi_{G^a}\left(x^0,X^1\right) \geq \Pi_{G^{a^\star}}\left(x^{0\star},X^{1\star}\right)$ for any IC and quasi-obedient triplet $\left(x^{0\star},X^{1\star},a^\star\right)$.
	It suffices to exhibit an IC and quasi-obedient $\bigl(x^{0\dag},X^{1\dag}\bigr)$ such that $\Pi_{G^1}\bigl(x^{0\dag},X^{1\dag}\bigr) \geq \Pi_{G^a}\left(x^0,X^1\right)$.
	Let 
	\begin{equation*}
		A \coloneqq \left\{s \in (0,\infty) : \text{$a=0$ a.e. on $(t,T)$ for some $t \in (0,s)$ and $T \in (s,\infty)$}\right\} .
	\end{equation*}
	Since $A$ is open, it equals the union of a collection $\mathcal{I}$ of disjoint open intervals.
	For any $t<T$ in $(0,\infty)$, write
	$\widebar{x}^{t,T} \coloneqq \int_t^T re^{-r(s-t)}x^0_s \dd s \bigm/ \int_t^T re^{-r(s-t)}\dd s$.

	\begin{namedthm}[Claim.]
		\label{claim:max_effort_quasiobedient}
		There exists an IC mechanism $\left(x^{0\star},X^{1\star}\right)$ such that
		$x^{0\star} \leq u^0$,
		the triplet $\left(x^{0\star},X^{1\star},a\right)$ is quasi-obedient, $\Pi_{G^a}\left(x^{0\star},X^{1\star}\right) \geq \Pi_{G^a}\left(x^0,X^1\right)$,
		and $x^{0\star} = \widebar{x}^{t,T}$ on $(t,T)$ for each $(t,T) \in \mathcal{I}$.
	\end{namedthm}
	
	\begin{proof}[Proof]%
	\renewcommand{\qedsymbol}{$\square$}
		Let $x^{0\star} \coloneqq \widebar{x}^{t,T}$ on $(t,T)$ for each $(t,T) \in \mathcal{I}$, $x^{0^\star} \coloneqq x^0$ on $\R_+ \setminus A$, and
		\begin{equation*}
			X^{1\star}_s \coloneqq
			\begin{cases}
				X^1_s & \text{if $s \notin A$}\\
				\left(1-e^{-r(T-s)}\right)\widebar{x}^{t,T} + e^{-r(T-s)}X^1_T & \text{if $s \in (t,T)$ where $(t,T) \in \mathcal{I}$,}
			\end{cases}
		\end{equation*}
		and let $X^1_\infty \coloneqq 0$.
		We have $\Pi_{G^a}\left(x^{0\star},X^{1\star}\right) \geq \Pi_{G^a}\left(x^0,X^1\right)$ since $F^0$ is concave. Moreover, $X^{1\star} = X^1$ and $X^{0\star} = X^0$ on $\R_+ \setminus A$.
		Thus since the triplet $\left(x^0,X^1,a\right)$ is quasi-obedient, so is the triplet $\left(x^{0\star},X^{1\star},a\right)$.
		
		It remains to show that $\left(x^{0\star},X^{1\star}\right)$ is IC; equivalently, that $s \mapsto h^\star(s) \coloneqq e^{-rs}\left(X^{1\star}_s - X^{0\star}_s\right)$ is non-negative and decreasing.
		Since $X^{1\star} = X^1$ and $X^{0\star} = X^0$ on $\R_+ \setminus A$, $h^\star$ matches $s \mapsto h(s) \coloneqq e^{-rs}\left(X^1_s - X^0_s\right)$ on $\R_+ \setminus A$.
		Moreover, $h^\star$ is constant on $[t,T)$ for any $(t,T) \in \mathcal{I}$, with $h^\star(t) = h(T) \in [0,h(t)]$ if $T < \infty$, and $h^\star(t) = 0 \leq h(t)$ if $T = \infty$.
		Since $\left(x^0,X^1,a\right)$ is IC, $h$ is non-negative and decreasing, and thus so is $h^\star$.
	\end{proof}%
	\renewcommand{\qedsymbol}{$\blacksquare$}

	By the \hyperref[claim:max_effort_quasiobedient]{claim}, we may assume without loss of generality that $x^0 = \widebar{x}^{t,T}$ on $(t,T)$ for each $(t,T) \in \mathcal{I}$.

	If $\Pi_{G^a}\left(x^0,X^1\right) \leq F^0\left(u^0\right)$, then $\bigl(x^{0\dag},X^{1\dag}\bigr) \coloneqq \left(u^0,u^0+rc/\lambda\right)$ is clearly IC and quasi-obedient, and we have $\Pi_{G^1}\bigl(x^{0\dag},X^{1\dag}\bigr) \geq F^0\left(u^0\right)$ by assumption~\eqref{eq:moralhazard_assumption} on \cpageref{eq:moralhazard_assumption}.
	Assume for the remainder that $\Pi_{G^a}\left(x^0,X^1\right) > F^0\left(u^0\right)$.

	Let $x^{0\dag} \coloneqq x^0$.
	To define $X^{1\dag}$, let 
	\begin{equation*}
		S \coloneqq \left\{t \in \R_+: a_t = 1 \text{ and } F^1\left(X^1_t\right) \geq C_t\right\} ,
	\end{equation*}
	and note that for any effort schedule $a^\dag \leq a$, the triplet $\bigl(x^0,X^1,a^\dag\bigr)$ is IC and quasi-obedient since $\left(x^0,X^1,a\right)$ is.
	Then by \Cref{observation:improvement} with $a^\dag \coloneqq \1_S$, we have $a = 0$ a.e. on $\R_+ \setminus S$.
	Hence $S \ne \varnothing$, for otherwise $\Pi_{G^a}\left(x^0,X^1\right) = \int_0^\infty re^{-rs}F^0\left(x^0_s\right) \dd s \leq F^0\left(u^0\right)$.

	Note also that since the triplet $\left(x^0,X^1,a\right)$ is IC and quasi-obedient, so is the triplet $\left(x^{0\star},X^{1\star},a^\star\right)$ defined by
	\begin{equation*}
		\left( x^{0\star}_t , X^{1\star}_t, a^\star_t \right)
		\coloneqq \left( x^0_{t+\inf S}, X^1_{t+\inf S}, a_{t+\inf S} \right)
		\quad \text{for each $t \in \R_+$.}
	\end{equation*}
	Hence $\inf S = 0$, since otherwise
	\begin{align*}
		\Pi_{G^a}\left(x^0,X^1\right)
		&\leq \left(1-e^{-r \inf S}\right)F^0\left(u^0\right)
		+ e^{-r \inf S} \Pi_{G^{a^\star}}\left(x^{0\star},X^{1\star}\right)
		\\
		&< \left(1-e^{-r \inf S}\right)\Pi_{G^a}\left(x^0,X^1\right)
		+ e^{-r \inf S} \Pi_{G^a}\left(x^0,X^1\right)
		\\
		&= \Pi_{G^a}\left(x^0,X^1\right) 
	\end{align*}
	by $F^0\left(u^0\right) < \Pi_{G^a}\left(x^0,X^1\right)$ and $\Pi_{G^{a^\star}}\left(x^{0\star},X^{1\star}\right) \leq \Pi_{G^a}\left(x^0,X^1\right)$.

	Since $\inf S = 0$, $\R_+ \setminus \cl S$ equals the union of a collection $\mathcal{J}$ of disjoint open intervals.
	Choose a measurable $X^{1\dag} : \cl S \to [0,\infty]$ such that $X^{1\dag} \coloneqq X^1$ on $S$ and, for each $t \in \cl S \setminus S$, $X^{1\dag}_t = \lim_{n \to \infty} X^1_{t^n}$ for a sequence $(t^n)_{n \in \N}$ in $S$ that converges to $t$.%
		\footnote{For example, set $X^{1\dag}_t \coloneqq \limsup_{s \to t} \1_S(s)X^1_s$ for each $t \in \cl S \setminus S$. Note that $X^{1\dag}$ is measurable as it is upper semi-continuous.}
	Extend $X^{1\dag}$ to $\R_+$ by letting, for each $s \in (t,T) \in \mathcal{J}$,
	\begin{equation*}
		X^{1\dag}_s \coloneqq 
		\begin{cases}
			\frac{1-e^{-r(T-s)}}{1-e^{-r(T-t)}}X^{1\dag}_t
			+\frac{e^{-r(T-s)}-e^{-r(T-t)}}{1-e^{-r(T-t)}} X^{1\dag}_T
			& \text{if $T < \infty$}\\
			X^0_t + rc/\lambda
			& \text{if $T = \infty$.}		
		\end{cases}
	\end{equation*}

	To show that $\bigl(x^{0\dag},X^{1\dag}\bigr)$ is quasi-obedient, note first that given any $t \in \cl S$, choosing a sequence $(t^n)_{n \in \N}$ in $S$ such that $t^n \to t$ and $X^1_{t^n} \to \smash{X^{1\dag}_t}$,
	\begin{equation*}
		X^{1\dag}_t = \lim_{n \to \infty} X^1_{t^n} \geq \lim_{n \to \infty} X^0_{t^n} + \textstyle rc/\lambda = X^{0\dag}_t + rc/\lambda ,
	\end{equation*}
	where the inequality holds since $X^1 \geq X^0 + rc/\lambda$ on $S$ as $\left(x^0,X^1,a\right)$ is quasi-obedient, and the last equality follows from $X^{0\dag} = X^0$ and the continuity of $X^0$.
	It remains only to show that $X^{1\dag} \geq X^{0\dag} + rc/\lambda$ on $(t,T)$ for each $(t,T) \in \mathcal{J}$.
	So fix an arbitrary $(t,T) \in \mathcal{J}$.
	If $T = \infty$, then $X^{1\dag} = X^{0\dag} - rc/\lambda$ on $(t,T)$.
	Suppose instead that $T<\infty$.
	Note that
	\begin{equation}
		X^{0\dag}_s = \widebar{x}^{t,T} + e^{-r(T-s)}\left(X^{0\dag}_T-\widebar{x}^{t,T}\right)\quad \text{for each $s \in [t,T]$,}
		\label{eq:X0_interval}
	\end{equation}
	so that $X^{1\dag}-X^{0\dag}$ is monotone on $[t,T]$, being an affine transformation of $s \mapsto e^{rs}$.
	This together with the fact that $X^{1\dag} \geq X^{0\dag} + rc/\lambda$ on $\{t,T\} \subseteq \cl S$ implies that $X^{1\dag} \geq X^{0\dag} + rc/\lambda$ on $[t,T]$, as desired.

	To prove that $\bigl(x^{0\dag},X^{1\dag}\bigr)$ is IC, it is necessary and sufficient to show that $t \mapsto h^\dag(t) \coloneqq \smash{e^{-rt}\bigl(X^{1\dag}_t - X^{0\dag}_t\bigr)}$ is non-negative and decreasing.
	Non-negativity follows from quasi-obedience.
	For monotonicity, it suffices to show that $h^\dag$ is decreasing on $\cl S$ and on the closure of each element of $\mathcal{J}$.
	For the former, $h^\dag = h$ on $S$, which is decreasing since $\left(x^0,X^1\right)$ is IC.
	For the latter, fix $(t,T) \in \mathcal{J}$.
	If $T = \infty$, note that $h^\dag$ matches the decreasing function $s \mapsto e^{-rs} rc/\lambda$ on $(t,T)$, and that $h^\dag(t) \geq e^{-rt} rc/\lambda$ since $\bigl(x^{0\dag},X^{1\dag}\bigr)$ is quasi-obedient.
	If instead $T < \infty$, then \eqref{eq:X0_interval} implies that $h^\dag$ is monotone on $[t,T]$, being an affine transformation of $s \mapsto e^{-rs}$.
	Since $h^\dag(t) \geq h^\dag(T)$ and $\{t,T\} \subseteq \cl S$, it follows that $h^\dag$ is decreasing on $[t,T]$.

	It remains to prove that $\Pi_{G^1}\bigl(x^{0\dag},X^{1\dag}\bigr) \geq \Pi_{G^a}\left(x^0,X^1\right)$.
	Note that $C$ equals the continuation payoff of $\bigl(x^{0\dag},X^{1\dag},a\bigr)$, since $x^{0\dag} = x^0$, $X^{1\dag} = X^1$ on $S$, and $a = 0$ a.e. on $\R_+ \setminus S$.
	It therefore suffices to show that $F^1\bigl(X^{1\dag}\bigr) \geq C$, since then applying \Cref{observation:improvement} the triplet $\bigl(x^{0\dag},X^{1\dag},a\bigr)$ and the effort schedule $a^\dag \coloneqq 1$ yields $\Pi_{G^1}\bigl(x^{0\dag},X^{1\dag}\bigr) \geq \Pi_{G^a}\bigl(x^{0\dag},X^{1\dag}\bigr) = \Pi_{G^a}\left(x^0,X^1\right)$.

	To that end, note that for each $t \in \cl S$, there is a sequence $(t^n)_{n \in \N}$ in $S$ such that $t^n \to t$ and $X^1_{t^n} \to X^{1\dag}_t$, and thus
	\begin{equation*}
		F^1\left(X^{1\dag}_t\right)-C_t
		\geq \lim_{n \to \infty} F^1\left(X^1_{t^n}\right) - C_{t^n} \geq 0,
	\end{equation*} 
	where the first inequality holds since $F^1$ is upper semi-continuous and $C$ is continuous.
	Now fix $(t,T) \in \mathcal{J}$.
	If $T = \infty$, then
	\begin{equation*}
		F^1\left(X^{1\dag}\right) = F^1\left(X^0_t + rc/\lambda \right) \geq F^0\left(X^0_t\right) = C
		\quad \text{on $(t,T)$,}
	\end{equation*} 
	where the inequality follows from assumption~\eqref{eq:moralhazard_assumption} on \cpageref{eq:moralhazard_assumption}.
	If $T < \infty$, fix an $s \in [t,T]$ and note that $C_s = F^0\bigl(\widebar{x}^{t,T}\bigr) + e^{-r(T-s)}\bigl[C_T-F^0\bigl(\widebar{x}^{t,T}\bigr)\bigr]$, so that
	\begin{multline*}
		F^1\left(X^{1\dag}_s\right) - C_s
		\\
		\begin{aligned}
			&\geq \frac{1-e^{-r(T-s)}}{1-e^{-r(T-t)}}F^1\left(X^{1\dag}_t\right)
			+ \frac{e^{-r(T-s)}-e^{-r(T-t)}}{1-e^{-r(T-t)}} F^1\left(X^{1\dag}_T\right) - C_s
			\\
			&\geq \min_{s' \in \{t,T\}} F^1\left(X^{1\dag}_{s'}\right) - C_{s'} \geq 0 ,
		\end{aligned}
	\end{multline*}
	where the first inequality holds since $F^1$ is concave, the second holds since its left-hand side is monotone in $s \in [t,T]$ (being an affine transformation of the map $s \mapsto e^{rs}$, by \eqref{eq:X0_interval}), and the last inequality holds since $\{t,T\} \subseteq \cl S$.
\end{proof}

\subsection{Proof of the \texorpdfstring{\hyperref[lemma:euler]{Euler lemma}}{Euler lemma} (\texorpdfstring{\cref{app:Euler:Euler_lemma}}{appendix \ref{app:Euler:Euler_lemma}})}
\label{suppl:Euler_lemma_pf}

Recall from \cref{app:Euler} the definitions of $\mathcal{X}$ and $\pi_G$,
and note that the former is a convex set.
For $j \in \{0,1\}$, define $F^{j\prime} : \R_+^2 \to \R_+$ by
\begin{equation*}
	F^{j\prime}\left(u,u'\right)
	\coloneqq
	\begin{cases}
		F^{j-}(u)	& \text{if $u'<u$} \\
		0			& \text{if $u'=u$} \\
		F^{j+}(u)	& \text{if $u'>u$,}
	\end{cases}
\end{equation*}
where $F^{j-}$ ($F^{j+}$) denotes the left-hand (right-hand) derivative.
Write
\begin{equation*}
	\DD \pi_G\left( x, x^\dag-x \right)
	\coloneqq \lim_{\alpha \downarrow 0}
	\frac{\pi_G\left( x + \alpha \left[ x^\dag - x \right] \right)-\pi_G(x)}{\alpha}
\end{equation*}
for the Gateaux derivative of $\pi_G$ at $x$ in direction $x^\dag-x$.

Let $\mathcal{X}_G$ be the set of $x \in \mathcal{X}$ such that the maps
$\psi^0_{x,u} : \R_+ \to [0,\infty]$
and $\psi^1_{X,u} : \R_+ \to [-\infty,\infty]$ defined by
\begin{equation*}
	\psi^0_{x,u}(t) \coloneqq r \int_0^t e^{-rs}F^{0\prime}\left(x_s,u\right)\dd s
	\quad \text{and} \quad
	\psi^1_{X,u}(t) \coloneqq e^{-rt}F^{1\prime}\left(X_t,u\right)
	\label{eq:euler_psi01_defn}
	\tag*{}
\end{equation*}
are $G$-integrable for any $u \in \left(0,u^0\right)$.

We require three lemmata. The first two are measure-theoretic housekeeping, while the last gives a formula for the Gateaux derivatives of $\pi_G$.

\begin{lemma}
	\label{lemma:euler_integrable_supergradient}
	If $x \in \mathcal{X}$ and $(x,X)$ satisfies the Euler equation for some $\phi^0,\phi^1$,
	then $x$ belongs to $\mathcal{X}_G$,
	and the map $t \mapsto r \int_0^t e^{-rs} \phi^0(s) \dd s$ is $G$-integrable.
\end{lemma}

\begin{lemma}
	\label{lemma:integrable_supergradient}
	$\argmax_{\mathcal{X}}\pi_G \subseteq \mathcal{X}_G$.
\end{lemma}

\begin{namedthm}[Gateaux lemma.]
	\label{lemma:euler_gateaux}
	For any $x,x^\dag \in \mathcal{X}_G$, measurable $\phi^0 : \R_+ \to [0,\infty]$ such that $t \mapsto r \int_0^t e^{-rs} \phi^0(s)\dd s$ is $G$-integrable, and $G$-integrable $\phi^1: \R_+ \to [-\infty,\infty]$,
	the Gateaux derivative $\DD \pi_G(x,x^\dag-x)$ exists and is equal to
	\begin{align*}
		&r \int_0^\infty e^{-rt} \left(
		[1-G(t)] \phi^0(t) + \int_{[0,t]} \phi^1 \dd G
		\right)
		\left(x^\dag_t-x_t\right) \dd t
		\\
		&\qquad + \E_G\left( r \int_0^\tau e^{-rt}
		\left[ F^{0\prime}\left( x_t, x^\dag_t \right) - \phi^0(t) \right]
		\left[ x^\dag_t - x_t \right] \dd t \right)
		\\
		&\qquad + \E_G\left( e^{-r\tau} \left[
		F^{1\prime}\left( X_\tau, X^\dag_\tau \right) - \phi^1(\tau) \right]
		\left[ X^\dag_\tau - X_\tau \right] \right) .
	\end{align*}
\end{namedthm}

\Cref{lemma:euler_integrable_supergradient} and the \hyperref[lemma:euler_gateaux]{Gateaux lemma} follow from standard arguments, which we omit \parencite[see][]{omit}.
\Cref{lemma:integrable_supergradient} is proved below.

\begin{proof}[Proof of the {\hyperref[lemma:euler]{Euler lemma}}]
	Fix a distribution $G$.
	For the first part, suppose that $x \in \mathcal{X}$ and that $(x,X)$
	satisfies the Euler equation with $\phi^0,\phi^1$ (the former measurable, the latter $G$-integrable).
	Then by \Cref{lemma:euler_integrable_supergradient}, $x$ belongs to $\mathcal{X}_G$, and $t \mapsto r \int_0^t e^{-rs} \phi^0(s) \dd s$ is $G$-integrable.

	By \Cref{corollary:optimal_existence} in \cref{suppl:undom_opt_properties} (\cpageref{corollary:optimal_existence}),
	there is a mechanism $\left(x^\star,X^\star\right)$ that is optimal for $G$.
	We must have $x^\star \in \mathcal{X}$ by \Cref{lemma:lequ0} (\cpageref{lemma:lequ0}),
	and thus $x^\star \in \argmax_{\mathcal{X}} \pi_G$.
	So it suffices to show that $\pi_G\left( x^\star \right) \leq \pi_G(x)$.

	By \Cref{lemma:integrable_supergradient}, $x^\star$ belongs to $\mathcal{X}_G$.
	Thus $x,x^\star$ and $\phi^0,\phi^1$ satisfy the hypotheses of the \hyperref[lemma:euler_gateaux]{Gateaux lemma}.
	Moreover, $\pi_G$ is concave since $F^0,F^1$ are and the map $x \mapsto X$ is linear,
	so for any $\alpha \in (0,1)$, we have
	\begin{equation*}
		\frac{ \pi_G\left( x+\alpha\left[ x^\star - x \right] \right) - \pi_G(x) }
		{\alpha}
		\geq \pi_G\left( x^\star \right) - \pi_G(x) .
	\end{equation*}
	The left-hand side converges as $\alpha \downarrow 0$ by the \hyperref[lemma:euler_gateaux]{Gateaux lemma}, yielding
	\begin{equation*}
		\DD \pi_G\left( x, x^\star-x \right)
		\geq \pi_G\left( x^\star \right) - \pi_G(x) .
	\end{equation*}
	It therefore suffices to show that $\DD \pi_G\left( x, x^\star-x \right) \leq 0$.
	And indeed,
	the first term in the \hyperref[lemma:euler_gateaux]{Gateaux lemma's}
	expression for $\DD \pi_G\left( x, x^\star-x \right)$
	is zero by \eqref{eq:euler},
	while second (third) term is non-positive by definition of $\phi^0$ ($\phi^1$) and the concavity of $F^0$ ($F^1$).%
		\footnote{For the second term,
		$\phi^0(t)$ is a supergradient of the concave function $F^0$ at $x_t$ for a.e. $t \in \R_+$ with $G(t)<1$, and at each such $t$,
		$x^\star_t > x_t$ implies $F^{0\prime}( x_t, x^\star_t ) = F^{0+}(x_t) \leq \phi^0(t)$
		and $x^\star_t < x_t$ implies $F^{0\prime}( x_t, x^\star_t ) = F^{0-}(x_t) \geq \phi^0(t)$.
		Analogously for the third term.}

	For the second part, fix an $x^\dag \in \argmax_{\mathcal{X}} \pi_G$.
	Since $\phi^0,\phi^1$ satisfy \eqref{eq:euler}, what must be shown is merely that
	\begin{itemize}

		\item for a.e. $t \in \R_+$ with $G(t) < 1$, $\phi^0(t)$ is a supergradient of $F^0$ at $x^\dag_t$, and

		\item for $G$-a.e. $t \in \R_+$, $\phi^1(t)$ is a supergradient of $F^1$ at $X^\dag_t$.

	\end{itemize}
	$x^\dag$ belongs to $\mathcal{X}_G$ by \Cref{lemma:integrable_supergradient}.
	So by the \hyperref[lemma:euler_gateaux]{Gateaux lemma} (with the roles of $x$ and $x^\dag$ reversed) and \eqref{eq:euler},
	\begin{align*}
	 	\DD \pi_G\left( x^\dag, x-x^\dag \right)
		&= \E_G\left( r \int_0^\tau e^{-rt}
		\left[ F^{0\prime}\left( x^\dag_t, x_t \right) - \phi^0(t) \right]
		\left[ x_t - x^\dag_t \right] \dd t \right)
		\\
		&\quad + \E_G\left( e^{-r\tau} \left[
		F^{1\prime}\left( X^\dag_\tau, X_\tau \right) - \phi^1(\tau) \right]
		\left[ X_\tau - X^\dag_\tau \right] \right) .
	\end{align*}
	We must have $\DD \pi_G\bigl( x^\dag, x-x^\dag \bigr) \leq 0$ since $\pi_G$ is maximised at $x^\dag$.
	On the other hand, the two integrands
	\begin{align*}
		t &\mapsto \left[ F^{0\prime}\left( x^\dag_t, x_t \right) - \phi^0(t) \right]
		\left[ x_t - x^\dag_t \right]
		\\
		\text{and} \quad
		t &\mapsto \left[ F^{1\prime}\left( X^\dag_t, X_t \right) - \phi^1(t) \right]
		\left[ X_t - X^\dag_t \right]
	\end{align*}
	are non-negative at, respectively, a.e. $t \in \R_+$ with $G(t)<1$ (the first integrand) and at $G$-a.e. every $t \in \R_+$ (the second).%
		\footnote{For the first integrand,
		$\phi^0(t)$ is a supergradient of the concave function $F^0$ at $x_t$ for a.e. $t \in \R_+$ with $G(t)<1$, and at every such $t$,
		$x_t < x^\dag_t$
		implies $F^{0\prime}( x^\dag_t, x_t ) = F^{0-}( x^\dag_t ) \leq F^{0+}(x_t) \leq \phi^0(t)$
		and $x_t > x^\dag_t$ implies $F^{0\prime}( x^\dag_t, x_t ) \geq \phi^0(t)$.
		Similarly for the second integrand.}
	Thus the first (second) integrand must be equal to zero at a.e. $t \in \R_+$ at which $G(t)<1$
	(at $G$-a.e. $t \in \R_+$).
	For a.e. $t \in \R_+$ with $G(t)<1$ at which the first integrand is zero,
	$\phi^0(t)$ is a supergradient of $F^0$ at $\smash{x^\dag_t}$.%
		\footnote{If the first integrand is zero at $t$,
		then either $F^{0\prime}( x^\dag_t, x_t ) = \phi^0(t)$
		or $x_t = x^\dag_t$.
		If the former, then $\phi^0(t)$ is a supergradient of $F^0$ at $\smash{x^\dag_t}$.
		And for almost every $t \in \R_+$ with $G(t)<1$ at which the latter holds,
		$\phi^0(t)$ is a supergradient of $F^0$ at $x_t = x^\dag_t$.}
	Similarly, $\phi^1(t)$ is a supergradient of $F^1$ at $\smash{X^\dag_t}$ for $G$-a.e. $t \in \R_+$ at which the second integrand is zero.
\end{proof}

\begin{proof}[Proof of \Cref{lemma:integrable_supergradient}]
	Fix an $x \in \mathcal{X} \setminus \mathcal{X}_G$;
	we must show that it does not belong to $\argmax_{\mathcal{X}}\pi_G$.
	By hypothesis, there is a $u \in \left(0,u^0\right)$ such that either $\psi^0_{x,u}$ or $\psi^1_{X,u}$ (defined on \cpageref{eq:euler_psi01_defn}) fails to be $G$-integrable.
	Define $\smash{x^\dag \equiv u}$; it clearly belongs to $\mathcal{X}$.
	It suffices to show that $\smash{\DD \pi_G\bigl(x,x^\dag-x\bigr) = \infty}$,
	since then
	\begin{equation*}
		\pi_G\left(x + \alpha\left[x^\dag-x\right]\right) > \pi_G(x)
		\quad \text{for $\alpha \in (0,1)$ small enough.}
	\end{equation*}

	Fix an $\eps \in (0,u)$, and define
	\begin{equation*}
		\mathcal{T} \coloneqq \left\{t \in \R_+ : x_t > u - \eps\right\} .
	\end{equation*}
	Choose $\eps' \in \left(0,u \meet \left[u^0 - u \right]\right)$ so that $\eps' < u^1 \meet \left(u^0 - u^1 \right)$ if $u^1 > 0$, and let
	\begin{equation*}
		\mathcal{T}'
		\coloneqq
		\begin{cases}
			\left\{t \in \R_+ : X_t < u + \eps'\right\}
			& \text{if $u^1 = 0$}
			\\
			\left\{t \in \R_+ : \left(u \meet u^1\right)-\eps' < X_t < \left(u \join u^1\right) + \eps'\right\}
			& \text{if $u^1 > 0$.}
		\end{cases}
	\end{equation*}
	(Here `$\meet$' and `$\join$' denote the minimum and maximum, respectively.)

	\begin{namedthm}[Claim.]
		\label{claim:integrable_supergradient}
		$\DD\pi_G\bigl(x,x^\dag-x\bigr)$ exists in $[-\infty,\infty]$,
		and for some $C \in \R$,
		\begin{align*}
			\DD\pi_G\left(x,x^\dag-x\right)
			&= \eps \E_G\left(
			r \int_0^\tau
			e^{-rt} \abs*{ F^{0\prime}\left(x_t,u\right) }
			\1_{\R_+ \setminus \mathcal{T}'}(t)
			\dd t
			\right)
			\\
			&\quad+ \eps' \E_G\left(
			e^{-r\tau} \abs*{F^{1\prime}\left(X_\tau,u\right)}
			\1_{\R_+ \setminus \mathcal{T}'}(\tau)
			\right)
			+ C .
		\end{align*}
	\end{namedthm}

	The \hyperref[claim:integrable_supergradient]{claim} follows from standard arguments, which we omit \parencite[see][]{omit}.
	Now, the maps
	\begin{equation*}
		t \mapsto r \int_0^t
		e^{-rs}F^{0\prime}\left(x_s,u\right)
		\1_{\mathcal{T}}(s)
		\dd s
		\quad \text{and} \quad
		t \mapsto e^{-rt} F^{1\prime}\left(X_t,u\right)
		\1_{\mathcal{T}'}(t)
	\end{equation*}
	are $G$-integrable
	because $F^0$ is Lipschitz continuous on $\left[u-\eps,u^0\right]$
	and
	\begin{equation*}
		\text{$F^1$ is Lipschitz continuous on}\quad
		\begin{cases}
			\left[0,u+\eps'\right]
			& \text{if $u^1=0$} \\
			\left[ \left(u \meet u^1\right)-\eps',
			\left(u \join u^1\right) + \eps'\right]
			& \text{if $u^1>0$.}
		\end{cases}
	\end{equation*}
	Since either $\psi^0_{x,u}$ or $\psi^1_{X,u}$ (defined on \cpageref{eq:euler_psi01_defn}) fails to be $G$-integrable (as $x \notin \mathcal{X}_G$ by hypothesis),
	it must therefore be that one of the maps
	\begin{equation*}
		t \mapsto r \int_0^t
		e^{-rs} F^{0\prime}\left(x_s,u\right)
		\1_{\R_+ \setminus \mathcal{T}}(s)
		\dd s
		\quad \text{and} \quad
		t \mapsto e^{-rt} F^{1\prime}\left(X_t,u\right)
		\1_{\R_+ \setminus \mathcal{T}'}(t)
	\end{equation*}
	fails to be $G$-integrable.
	In either case, the \hyperref[claim:integrable_supergradient]{claim} implies that $\DD \pi_G\bigl(x,x^\dag-x\bigr) = \infty$, as desired.
\end{proof}

\subsection{Proofs of the construction lemmata (\texorpdfstring{\cref{app:Euler:construction}}{appendix \ref{app:Euler:construction}})}
\label{suppl:construction}

In this appendix, we prove the lemmata in \cref{app:Euler:construction}
used to construct a solution of the superdifferential Euler equation (\cref{app:Euler}).
All of the arguments are elementary but tedious.

\subsubsection{Proof of \texorpdfstring{\Cref{lemma:existence_simple}}{Lemma~\ref{lemma:existence_simple}}}
\label{suppl:construction:pf_existence_simple}

Enumerate the support of $G$ as $\supp(G) = \{t_k\}_{k=1}^K \subseteq \R_+$, where $K \in \N$ and
\begin{equation*}
	0 \leq t_1 < \cdots < t_K < \infty .
\end{equation*}
As $F^{0\prime}$ is continuous and strictly decreasing on $\left[u^\star,u^0\right]$, it admits a continuous and decreasing inverse $\inv F^{0\prime} : \left[ F^{0\prime}(u^0), F^{0\prime}\left(u^\star\right) \right] \to \left[u^\star,u^0\right]$.
Extend $\inv F^{0\prime}$ to $\R$ by making it constant on $(-\infty,F^{0\prime}(u^0)]$ and on $[F^{0\prime}\left(u^\star\right),\infty)$, so that continuity and monotonicity are preserved.

For $\lambda \in \left[u^\star,u^0\right]$, let $x^\lambda_{t_K} \coloneqq X^\lambda_{t_K} \coloneqq \lambda$ and, if $K > 1$, define a sequence
$\{ x^\lambda_{t_k}, X^\lambda_{t_k} \}_{k=1}^{K-1}$
in $\left[u^\star,u^0\right]$ recursively by
\begin{align*}
	x^\lambda_{t_k}
	&\coloneqq \inv F^{0\prime}\left(
	\E_G\left(F^{1\prime}\left( X^\lambda_\tau\right) \middle| \tau > t_k \right)
	\right)
	\quad \text{and}
	\\
	X^\lambda_{t_k}
	&\coloneqq \left( 1 - e^{r(t_k-t_{k+1})} \right) x^\lambda_k
	+ e^{r(t_k-t_{k+1})} X^\lambda_{t_{k+1}} .
\end{align*}

\begin{namedthm}[Claim.]
	\label{claim:thmC_pf_finite_decreasing}
	The sequence $\bigl( x^\lambda_{t_k} \bigr)_{k=1}^K$ is decreasing.
\end{namedthm}

\begin{proof}%
	\renewcommand{\qedsymbol}{$\square$}
	We prove that the sequence $\bigl( x^\lambda_{t_k} \bigr)_{k=k'}^K$ is decreasing for every $k' \in \{1,\dots,K-1\}$ by backward induction on $k'$.
	For the base case $k' = K-1$, we have
	\begin{equation*}
		x^\lambda_{t_{K-1}}
		= \inv F^{0\prime}\left( F^{1\prime}(\lambda) \right)
		\geq \lambda
		= x^\lambda_{t_K} ,
	\end{equation*}
	where the inequality holds since $F^{0\prime} \geq F^{1\prime}$ on $\left[ u^\star, u^0 \right] \ni \lambda$.

	For the induction step, suppose for $k' \in \{1,\dots,K-2\}$ that $\smash{( x^\lambda_{t_k} )_{k=k'+1}^K}$ is decreasing;
	we must show that $\smash{x_{t_{k'}} \geq x_{t_{k'+1}}}$.
	The induction hypothesis implies that
	$\smash{( X^\lambda_{t_k} )_{k=k'+1}^K}$
	is also decreasing,
	which since $F^{1\prime}$ is a decreasing function implies that
	\begin{equation*}
		F^{1\prime}\left( X^1_{t_{k'+1}} \right)
		\leq \E_G\left( F^{1\prime}\left( X^1_\tau \right)
		\middle| \tau > t_{k'+1} \right) ,
	\end{equation*}
	and thus
	\begin{align*}
		\E_G\left( F^{1\prime}\left(X^1_\tau\right)
		\middle| \tau > t_{k'} \right)
		&= \frac{G(t_{k'+1})-G(t_{k'})}{1-G(t_{k'})}
		F^{1\prime}\left( X^1_{t_{k'+1}} \right)
		\\
		&\quad + \frac{1-G(t_{k'+1})}{1-G(t_{k'})}
		\E_G\left( F^{1\prime}\left(X^1_\tau\right)
		\middle| \tau > t_{k'+1} \right)
		\\
		&\leq \E_G\left( F^{1\prime}\left(X^1_\tau\right)
		\middle| \tau > t_{k'+1} \right) .
	\end{align*}
	Since $\inv F^{0\prime}$ is decreasing, it follows that $x_{t_{k'}} \geq x_{t_{k'+1}}$.
\end{proof}%
\renewcommand{\qedsymbol}{$\blacksquare$}

Since $\inv F^{0\prime}$ and $F^{1\prime}$ are continuous, $\lambda \mapsto x^\lambda_{t_k}$ and $\lambda \mapsto X^\lambda_{t_k}$ are continuous on $\left[u^\star,u^0\right]$ for every $k \in \{1,\dots,K\}$.%
	\footnote{Proceed by strong backward induction on $k \in \{1,\dots,K\}$. Clearly continuity holds in the base case $k = K$.
	For the induction step, suppose for $k<K$ that $\smash{\lambda \mapsto X^\lambda_{t_{k'}}}$ is continuous for all $k' > k$.
	Then $\smash{\lambda \mapsto x^\lambda_{t_k}}$ is continuous, and thus so is $\smash{\lambda \mapsto X^\lambda_{t_k}}$.}
Thus the map $\psi : \left[u^\star,u^0\right] \to \R$ defined by
\begin{equation*}
	\psi(\lambda)
	\coloneqq \E_G \left( F^{1\prime}\left(X^\lambda_\tau\right) \right)
	\quad \text{for each $\lambda \in \left[u^\star,u^0\right]$}
\end{equation*}
is continuous.
Since $F^0$ and $F^1$ are continuously differentiable and $u^\star > 0$, we have by definition of $u^\star$ that
$F^{1\prime}\left(u^\star\right)
= F^{0\prime}\left(u^\star\right)$
and
$F^{1\prime}\left(u^0\right)
\leq F^{0\prime}\left(u^0\right)$.
Thus if $\lambda \in \left\{ u^\star, u^0 \right\}$,
then $x^\lambda_{t_k} = \lambda$ for every $k \in \{1,\dots,K\}$,
so that $\psi(\lambda) = F^{1\prime}(\lambda)$.%
	\footnote{This follows easily by backward induction on $k \in \{1,\dots,K\}$.}
It follows that
\begin{equation*}
	\psi\left(u^\star\right)
	= F^{0\prime}\left(u^\star\right)
	\geq 0
	= F^{0\prime}\left(u^0\right)
	\geq \psi\left(u^0\right) .
\end{equation*}
Hence the continuous function $\psi$ has a root $\lambda_\star \in \left[u^\star,u^0\right]$ by the intermediate value theorem.

Let $x : \R_+ \to \left[u^\star,u^0\right]$ be given by
\begin{equation*}
	x_t
	\coloneqq
	\begin{cases}
		u^0
		& \text{for $t \in [0,t_1)$}
		\\
		x^{\lambda_\star}_{t_k}
		& \text{for $t \in [t_k,t_{k+1})$ where $k \in \{1,\dots,K-1\}$}
		\\
		x^{\lambda_\star}_{t_K}
		& \text{for $t \in [t_K,\infty)$.}
	\end{cases}
\end{equation*}
The mechanism $(x,X)$ satisfies the Euler equation
by \Cref{observation:euler_differentiable} in \cref{app:Euler:construction} (\cpageref{observation:euler_differentiable}).
\qed

\subsubsection{Proof of \texorpdfstring{\Cref{lemma:approx_G}}{Lemma~\ref{lemma:approx_G}}}
\label{suppl:construction:pf_lemma_approx_G}

Since $F^0,F^1$ are simple and $x$ belongs to $\mathcal{X}'$,
it suffices by \Cref{observation:euler_differentiable} in \cref{app:Euler:construction} (\cpageref{observation:euler_differentiable})
to show that
$\E_G\left( F^{1\prime}(X_\tau) \right) = 0$ and
\begin{equation}
	F^{0\prime}(x_t)
	= \E_G\left( F^{1\prime}(X_\tau) \middle| \tau > t \right)
	\quad \text{for a.e. $t \in \R_+$ such that $G(t) < 1$.}
	\label{eq:euler_forward_diff}
\end{equation}
By \Cref{observation:euler_differentiable} again,
$\left(x^n,X^n\right)$ and $G_n$ satisfy $\E_{G_n} \left( F^{1\prime}(X^n_\tau) \right) = 0$ and \eqref{eq:euler_forward_diff} for each $n \in \N$.

To show that $\E_G \left(F^{1\prime}(X_\tau)\right) = 0$,
note that by simplicity, $F^{1\prime}$ is $L$-Lipschitz on $\left[u^\star,u^0\right]$ for some $L > 0$.
Thus for any $T \in \R_+$, we have
\begin{multline*}
	\abs*{ \E_{G_n}\left( F^{1\prime}\left(X^n_\tau\right)
	- F^{1\prime}\left(X_\tau\right) \right) }
	\leq L \E_{G_n}\left( \abs*{ X^n_\tau-X_\tau } \right)
	\\
	\begin{aligned}
		&\leq L \E_{G_n}\Bigl( \abs*{ X^n_\tau-X_\tau }
		\Bigm| \tau \leq T \Bigr)
		+ L[1-G_n(T)]\left(u^0-u^\star\right)
		\\
		&\leq L \E_{G_n}\left(
		e^{r\tau} r \int_{\tau}^\infty e^{-rt} \abs*{ x^n_t - x_t } \dd t \;
		\middle| \tau \leq T \right)
		+ L[1-G_n(T)]\left(u^0-u^\star\right)
		\\
		&\leq Le^{rT} r \int_0^\infty e^{-rt} \abs*{ x^n_t - x_t } \dd t
		+ L [1-G_n(T)] \left( u^0 - u^\star \right) .
	\end{aligned}
\end{multline*}
Since $T \in \R_+$ was arbitrary and $G_n \to G$ and $x^n \to x$ pointwise,
it follows that the left-hand side vanishes as $n \to \infty$.%
	\footnote{Fix any $\eps > 0$;
	we seek an $N \in \N$ such that
	$\abs{ \E_{G_n}( F^{1\prime}(X^n_\tau) - F^{1\prime}(X_\tau) ) } < \eps$
	for all $n \geq N$.
	To that end, choose a $T \in \R_+$ large enough that $[1-G(T)] L (u^0-u^\star) < \eps/3$.
	Since $G_n \to G$ and $x^n \to x$ pointwise,
	we may find an $N \in \N$ such that
	both $\abs{ G(T) - G_n(T) } L (u^0-u^\star) < \eps/3$ and $Le^{rT} r \int_0^\infty e^{-rt} \abs*{ x^n_t - x_t } \dd t < \eps/3$ for all $n \geq N$.}
Thus since $\left(x^n,X^n\right)$ and $G_n$ satisfy $\E_{G_n} \left(F^{1\prime}(X^n_\tau)\right) = 0$ for each $n \in \N$, we have
\begin{multline*}
	\abs*{ \E_G\left( F^{1\prime}\left( X_\tau \right) \right) }
	= \abs*{ \E_{G_n}\left( F^{1\prime}\left( X^n_\tau \right) \right)
	- \E_G\left( F^{1\prime}\left( X_\tau \right) \right) }
	\\
	\leq \abs*{ \E_{G_n}\left( F^{1\prime}\left( X^n_\tau \right)
	- F^{1\prime}\left( X_\tau \right) \right) }
	+ \abs*{ \E_{G_n}\left( F^{1\prime}\left( X_\tau \right) \right)
	- \E_G\left( F^{1\prime}\left( X_\tau \right) \right) }
	\to 0
\end{multline*}
as $n \to \infty$,
where the second term vanishes because $G_n \to G$ pointwise (hence weakly) and $X$ and $F^{1\prime}$ are bounded and continuous.

It remains to derive \eqref{eq:euler_forward_diff}.
Since $\left(x^n,X^n\right)$ and $G_n$ satisfy \eqref{eq:euler_forward_diff} for each $n \in \N$, we have for a.e. $t \in \R_+$ that
\begin{multline*}
	\abs*{ F^{0\prime}(x_t)
	- \E_G\left( F^{1\prime}\left(X_\tau\right) \middle| \tau > t \right) }
	\\
	\leq \abs*{ F^{0\prime}(x_t) - F^{0\prime}\left( x^n_t \right) }
	+ \abs*{ \E_{G_n}\left( F^{1\prime}\left(X^n_\tau\right) \middle| \tau > t \right)
	- \E_G\left( F^{1\prime}\left(X_\tau\right) \middle| \tau > t \right) } .
\end{multline*}
The first term vanishes as $n \to \infty$ since $F^{0\prime}$ is continuous and $x^n \to x$ pointwise.
The second term vanishes by a straightforward variation on the above argument,
using the fact that since $G_n$ converges pointwise to $G$, the same is true of the conditional CDFs given $\tau > t$.%
	\footnote{For all sufficiently large $n \in \N$,
	we have $G_n(t)<1$ since $G_n(t) \to G(t) < 1$,
	so the conditional CDF $G_n / [1-G_n(t)]$ is well-defined
	and converges pointwise to $G / [1-G(t)]$.}
\qed

\subsubsection{Proof of \texorpdfstring{\Cref{lemma:approx_F}}{Lemma~\ref{lemma:approx_F}}}
\label{suppl:construction:pf_lemma_approx_F}

Choose a sequence $\left( F^0_n, F^1_n \right)_{n \in \N}$ of technologies
satisfying the following:%
	\footnote{For an explicit example, see \textcite{omit}.}
\begin{enumerate}[label=(\alph*)]

	\item \label{bullet:thmC_pf_general:a}
	$F^0_n,F^1_n$ are simple for every $n \in \N$,

	\item \label{bullet:thmC_pf_general:b}
	$u^0_n \uparrow u^0$, $u^1_n \to u^1$ and $u^\star_n \to u^\star$ as $n \to \infty$,

	\item \label{bullet:thmC_pf_general:c}
	\begin{enumerate}[label=(\roman*)]

		\item for any $u \in \left(0,u^0\right]$ at which $F^{1-}$ is finite,
		$\left(F^{0\prime}_n\right)_{n \in \N}$ and $\left(F^{1\prime}_n\right)_{n \in \N}$ are uniformly bounded below on $\left[0,u\right]$,

		\item for any $u \in \left[0,u^0\right)$ at which $F^{0+},F^{1+}$ are finite,
		$\left(F^{0\prime}_n\right)_{n \in \N}$ and $\left(F^{1\prime}_n\right)_{n \in \N}$ are uniformly bounded above on $\left[u,u^0\right]$, and

	\end{enumerate}

	\item \label{bullet:thmC_pf_general:d}
	for both $j \in \{0,1\}$ and any convergent sequence $(u_n)_{n \in \N}$ with $0 < u_n \leq u^0_n$ for each $n \in \N$,
	every subsequential limit of the sequence $( F^{j\prime}_n(u_n) )_{n \in \N}$ is a supergradient of $F^j$ at $\lim_{n \to \infty} u_n$.

\end{enumerate}

Fix a mechanism $(x,X)$ and a CDF $G$ with unbounded support,
and suppose that $x$ is the pointwise limit of a sequence $(x^n)_{n \in \N}$ such that
for each $n \in \N$,
$x^n$ belongs to $\mathcal{X}'_n$
and $(x^n,X^n)$ satisfies the Euler equation for $\left(F^0_n,F^1_n,G\right)$ and $x^n \in \mathcal{X}'_n$.
Assume without loss that each $x^n$ is decreasing and right-continuous.%
	\footnote{Each $x^n$ admits a decreasing right-continuous version, e.g. $\widetilde{x}^n$ given by
	$\widetilde{x}^n_t \coloneqq \sup_{s>t} x^n_t$.}

Since $F^0_n,F^1_n$ are simple for each $n \in \N$ (by property \ref{bullet:thmC_pf_general:a}),
we have by \Cref{observation:euler_differentiable} in \cref{app:Euler:construction} (\cpageref{observation:euler_differentiable}) that
\begin{equation*}
	[1-G(t)] F^{0\prime}_n\left(x^n_t\right)
	+ \int_{[0,t]} F^{1\prime}_n\left(X^n_s\right) G(\dd s) = 0
	\quad \text{for all $t \in \R_+$ and $n \in \N$.}
	\label{eq:euler_n}
	\tag{$\mathcal{E}$}
\end{equation*}
(In particular, this holds for a.e. $t \in \R_+$ by \Cref{observation:euler_differentiable},
and thus for every $t$ since $F^{0\prime}$ is continuous and $x^n$ is right-continuous.)

\begin{namedthm}[Claim.]
	\label{claim:approx_F}
	$X_0 < u^0$ unless $F^{1-}\left(u^0\right)$ is finite,
	and $X > 0$ unless $F^{0+}(0),F^{1+}(0)$ are finite.
\end{namedthm}

\begin{proof}%
	\renewcommand{\qedsymbol}{$\square$}
	For the first part, suppose toward a contradiction that $F^{1-}\left(u^0\right) = -\infty$ and $X_0 = u^0$.
	Then $x = X = u^0$ since $x \leq u^0$ (as $x \in \mathcal{X}'$),
	so that $x^n \to u^0$ pointwise and (thus) $X^n \to u^0$ pointwise.
	Fix a $t \in \R_+$ at which $G(t) > 0$.
	Property \ref{bullet:thmC_pf_general:d} implies that
	\begin{equation*}
		\limsup_{n \to \infty} F^{0\prime}_n\left(x^n_t\right)
		\leq F^{0-}\left(u^0\right)
		\quad \text{and} \quad
		\lim_{n \to \infty} F^{1\prime}_n\left(X^n_s\right)
		= -\infty
		\quad \text{for any $s \in [0,t]$.}
	\end{equation*}

	Since $u^1_n \to u^1 < u^0$ by property \ref{bullet:thmC_pf_general:b}, there is an $N \in \N$ such that $X^n_t \geq u^1_n$ for every $n \geq N$,
	and thus $X^n \geq u^1_n$ on $[0,t]$ since $X^n$ is decreasing (as $x^n \in \mathcal{X}'_n$).
	It follows that $s \mapsto F^{1\prime}_n\left(X^n_s\right)$ is non-positive
	for every $n \geq N$,
	so that Fatou's lemma applies, yielding
	\begin{multline*}
		\limsup_{n \to \infty}
		\left\{
		[1-G(t)] F^{0\prime}_n\left(x^n_t\right)
		+ \int_{[0,t]} F^{1\prime}_n\left(X^n_s\right) G(\dd s)
		\right\}
		\\
		\begin{aligned}
			&\leq [1-G(t)] F^{0-}\left(u^0\right)
			+ \limsup_{n \to \infty} \int_{[0,t]}
			F^{1\prime}_n\left(X^n_s\right) G(\dd s)
			\\
			&\leq [1-G(t)] F^{0-}\left(u^0\right)
			+ \int_{[0,t]} \limsup_{n \to \infty}
			F^{1\prime}_n\left(X^n_s\right) G(\dd s)
			= -\infty ,
		\end{aligned}
	\end{multline*}
	where the equality holds by $F^{0-}\left(u^0\right) < \infty$ and $G(t) > 0$.
	This is a contradiction with \eqref{eq:euler_n}.

	For the second part, suppose toward a contradiction that $X_t = 0$ for some $t \in \R_+$ and that either $F^{0+}(0) = \infty$ or $F^{1+}(0) = \infty$.
	Choose $u'' \in \left[X_0,u^0\right]$ such that $X_0 < u'' < u^0$ if $X_0 < u^0$,
	and note that $F^{1-}(u'')$ is finite by the first part of the \hyperref[claim:approx_F]{claim}.
	Since $X^n \to X$ pointwise and $X^n \leq u^0_n \leq u^0$ for each $n \in \N$ (by $x^n \in \mathcal{X}'_n$ and property \ref{bullet:thmC_pf_general:b}),
	there is an $N' \in \N$ such that $X^n_0 \leq u''$ for all $n \geq N'$.
	Since $X^n$ is decreasing, it follows that $X^n \leq u''$ for all $n \geq N'$.
	Thus by property \ref{bullet:thmC_pf_general:c},
	the sequence of maps
	$\left( s \mapsto F^{1\prime}_n\left(X^n_s\right) \right)_{n = N'}^\infty$
	is uniformly bounded below,
	so satisfies the hypothesis of Fatou's lemma.

	We have $X = x = 0$ on $[t,\infty)$ since $x$ is decreasing.
	So by property \ref{bullet:thmC_pf_general:d},
	\begin{equation*}
		\liminf_{n \to \infty} F^{0\prime}_n\left(x^n_s\right)
		\geq F^{0+}(0)
		\quad \text{and} \quad
		\liminf_{n \to \infty} F^{1\prime}_n\left(X^n_s\right)
		\geq F^{1+}(0)
		\quad \text{for any $s \geq t$.}
	\end{equation*}
	As $G$ has unbounded support, there is a $t' > t$ with $G(t) < G(t') < 1$.
	Then
	\begin{multline*}
		\liminf_{n \to \infty}
		\left\{
		[1-G(t')] F^{0\prime}_n\left(x^n_{t'}\right)
		+ \int_{[0,t']} F^{1\prime}_n\left(X^n_s\right)G(\dd s)
		\right\}
		\\
		\begin{aligned}
			&\geq
			[1-G(t')] \liminf_{n \to \infty}
			F^{0\prime}_n\left(x^n_{t'}\right)
			+ \int_{[0,t']} \liminf_{n \to \infty}
			F^{1\prime}_n\left(X^n_s\right)G(\dd s)
			\\
			&\geq
			[1-G(t')] F^{0+}(0)
			+ G(t') F^{1+}(0)
			= \infty ,
		\end{aligned}
	\end{multline*}
	where the first inequality holds by Fatou's lemma.
	This contradicts \eqref{eq:euler_n}.
\end{proof}%
\renewcommand{\qedsymbol}{$\blacksquare$}

Define $\phi^0_n,\phi^1_n : \R_+ \to \R$ by
\begin{equation*}
	\phi^0_n(t) \coloneqq F^{0\prime}_n\left( x^n_t \right)
	\quad \text{and} \quad
	\phi^1_n(t) \coloneqq F^{1\prime}_n\left( X^n_t \right)
	\quad \text{for each $t \in \R_+$.}
\end{equation*}
We shall show that for any $t \in \R_+$,
$\left( \phi^0_n \right)_{n \in \N}$ and $\left( \phi^1_n \right)_{n \in \N}$ are uniformly bounded on $[0,t]$.
So fix a $t \in \R_+$.
Choose $u' \in [0,X_t]$ so that $0 < u' < X_t$ in case $X_t > 0$,
and let $u'' \in \left[X_0,u^0\right]$ be such that $X_0 < u'' < u^0$ if $X_0 < u^0$.
By the \hyperref[claim:approx_F]{claim}, $F^{0+}(u')$, $F^{1+}(u')$ and $F^{1-}(u'')$ are finite.
Since $X^n \to X$ pointwise and $0 < u^\star_n \leq X^n \leq u^0_n \leq u^0$ (by $x^n \in \mathcal{X}'_n$ and property \ref{bullet:thmC_pf_general:b}), there is an $N \in \N$ such that $X^n_0 \leq u''$ and $X^n_t \geq u'$ for all $n \geq N$.
Since $x^n$ is decreasing for each $n \in \N$, it follows that
\begin{equation*}
	u' \leq x^n_s \leq u^0
	\quad \text{and} \quad
	u' \leq X^n_s \leq u''
	\quad \text{for all $s \in [0,t]$ and $n \geq N$.}
\end{equation*}
This together with property \ref{bullet:thmC_pf_general:c}
and the fact that $\phi^0_n \geq 0$ for each $n \in \N$
implies that $\left( \phi^0_n \right)_{n = N}^\infty$ and $\left( \phi^1_n \right)_{n = N}^\infty$
are uniformly bounded on $[0,t]$.
Since $\phi^0_n,\phi^1_n$ are bounded for each $n \in \N$
by property \ref{bullet:thmC_pf_general:a},
it follows that $\left( \phi^0_n \right)_{n \in \N}$ and $\left( \phi^1_n \right)_{n \in \N}$ are uniformly bounded on $[0,t]$, as desired.

For each $n \in \N$,
$\phi^0_n,\phi^1_n$ are increasing
since $x^n$ is decreasing and $F^0,F^1$ are concave.
Since $\left( \phi^0_n \right)_{n \in \N}$ and $\left( \phi^1_n \right)_{n \in \N}$ are also uniformly bounded on $[0,t]$ for any $t \in \R_+$,
it follows that
$\left( x^n \right)_{n \in \N}$ admits a subsequence along which $\phi^0_n,\phi^1_n$ converge pointwise to some increasing $\phi^0 : \R_+ \rightarrow [0,\infty]$ and $\phi^1 : \R_+ \rightarrow [-\infty,\infty]$,
by the Helly selection theorem.%
	\footnote{E.g. \textcite[p. 167]{Rudin1976}.
	For any subsequence of $\left( x^n \right)_{n \in \N}$ and $t \in \R_+$, Helly yields a sub-subsequence along which $( \phi^0_n )_{n \in \N}$ converges on $[0,t]$;
	a diagonalisation argument yields a subsequence of $\left( x^n \right)_{n \in \N}$ along which $( \phi^0_n )_{n \in \N}$ converges pointwise on $\R_+$.
	The same reasoning yields a further subsequence along which $( \phi^1_n )_{n \in \N}$ converges on $\R_+$.}

Clearly $\phi^0$ is measurable.
Moreover, since $x^n \to x$ pointwise and (thus) $X^n \to X$ pointwise,
the same is true along the subsequence.
So by property \ref{bullet:thmC_pf_general:d},
$\phi^0(t)$ ($\phi^1(t)$) is a supergradient of $F^0$ at $x_t$ (of $F^1$ at $X_t$)
for every $t \in \R_+$.
Moreover,
letting $n \to \infty$ in \eqref{eq:euler_n} yields that $\phi^0,\phi^1$ satisfy \eqref{eq:euler} for each $t \in \R_+$ by bounded convergence.

It remains only to show that $\phi^1$ is $G$-integrable.
Note first that $\phi^1$ is bounded below
by $\phi^1(0) = \lim_{n \to \infty} \phi^1_n(0) \in \R$
since it is increasing.
Hence
\begin{equation*}
	\phi^1(0)
	\leq \E_G\left(\phi^1(\tau)\right)
	\leq \liminf_{t \to \infty} \int_{[0,t]} \phi^1 \dd G
	= - \limsup_{t \to \infty} \left[ 1 - G(t) \right] \phi^0(t)
	\leq 0
\end{equation*}
by Fatou's lemma (second inequality),
\eqref{eq:euler} (the equality)
and the non-negativity of $\phi^0$ (final inequality).
Hence $\phi^1$ is $G$-integrable.
\qed

\subsection{Comparative statics}
\label{suppl:mcs}

When the likely time of the breakthrough becomes later,
the agent is optimally provided with a higher continuation utility $X_t$ in every period $t$:

\begin{namedthm}[Comparative statics theorem.]
	\label{theorem:mcs}
	Suppose that $F^0$ is strictly concave.
	Let $G,G^\dag$ be absolutely continuous distributions with equal, unbounded support.
	If $G$ MLR-dominates $G^\dag$,%
		\footnote{I.e. the ratio $G' / G^{\dag\prime}$ of their densities is increasing on the support.}
	then $X \geq X^\dag$ for any mechanisms $(x,X)$ and $\bigl( x^\dag, X^\dag \bigr)$
	that are optimal for $G$ and $G^\dag$, respectively.
\end{namedthm}

The restriction to absolutely continuous distributions $G,G^\dag$ with equal support is merely for simplicity.
The proof relies on the following two lemmata, which are proved below.
Recall from \cref{app:Euler} the definitions of the (superdifferential) Euler equation,
$\mathcal{X}$, $\mathcal{X}'$ and `simple'.

\begin{lemma}
	\label{lemma:mcs}
	Suppose that $F^0,F^1$ are simple.
	Let $G,G^\dag$ be finite-support distributions with $G(0)=G^\dag(0)=0$
	and equal support.
	If $G$ MLR-dominates $G^\dag$,%
		\footnote{I.e. the ratio $g / g^\dag$ of their probability mass functions is increasing on the support.}
	then $X \geq X^\dag$ for any mechanisms $(x,X)$ and $\bigl( x^\dag, X^\dag \bigr)$
	with $x,x^\dag \in \mathcal{X}'$
	that satisfy the Euler equation for $G$ and $G^\dag$, respectively.
\end{lemma}

\begin{lemma}
	\label{lemma:euler_optimal}
	If $F^0$ is strictly concave
	and $G$ has unbounded support,
	then a mechanism $(x,X)$ which satisfies $x \in \mathcal{X}$
	and the Euler equation
	is uniquely optimal for $G$.
\end{lemma}

We shall also use the construction lemmata (\ref{lemma:existence_simple}, \ref{lemma:approx_G} and \ref{lemma:approx_F}) in \cref{app:Euler:construction}.

\begin{proof}[Proof of the {\hyperref[theorem:mcs]{comparative statics theorem}}]
	Let $\smash{\left(F^0_n,F^1_n\right)_{n \in \N}}$ be the simple technologies delivered by \Cref{lemma:approx_F}.
	Choose sequences $(G_m)_{m \in \N}$ and $(G^\dag_m)_{m \in \N}$ of finite-support CDFs
	converging pointwise to (respectively) $G$ and $G^\dag$ such that for each $m \in \N$,
	$G_m(0) = G^\dag_m(0) = 0$,
	$G_m$ and $G^\dag_m$ have equal support,
	and the former MLR-dominates the latter.%
		\footnote{For example:
		let $\{Q_n\}_{n=0}^\infty$ be an enumeration of $\supp(G) \intersect \Q$
		with $Q_0 = \min \supp(G)$ and $G(Q_1), G^\dag(Q_1) > 0$,
		write $\{Q_k\}_{k=0}^m = \{q^m_k\}_{k=0}^m$ where $q^m_0 < \cdots < q^m_m$,
		and define
		\begin{equation*}
			G^{(\dag)}_m(t)
			= \frac{1}{G^{(\dag)}\left(q^m_m\right)}
			\sum_{k=1}^m
			\1_{\left[0,q^m_k\right]}(t)
			\left[ G^{(\dag)}\left(q^m_k\right)
			- G^{(\dag)}\left(q^m_{k-1}\right) \right]
			\quad \text{for each $t \in \R_+$.}
		\end{equation*}
		%
		}

	Fix an arbitrary $n \in \N$.
	For every $m \in \N$,
	\Cref{lemma:existence_simple} assures us of the existence of $x^{nm},x^{\dag,nm} \in \mathcal{X}'_n$ such that $(x^{nm},X^{nm})$ and $\bigl(x^{\dag,nm},X^{\dag,nm}\bigr)$
	satisfy the Euler equation for $\left(F^0_n,F^1_n,G_m\right)$ and $\bigl(F^0_n,F^1_n,G^\dag_m\bigr)$, respectively.
	Since $\mathcal{X}'_n$ is sequentially compact
	by \Cref{observation:thmC_pf_sequential_compactness} in \cref{app:Euler:construction} (\cpageref{observation:thmC_pf_sequential_compactness}),
	we may assume (passing to a subsequence if necessary) that
	\begin{equation*}
		x^{nm} \to x^n
		\quad \text{and} \quad
		x^{\dag,nm} \to x^{\dag,n}
		\quad \text{pointwise as $m \to \infty$}
	\end{equation*}
	for some $x^n,x^{\dag,n} \in \mathcal{X}'_n$.
	Since $u^0_n \to u^0$ and $u^\star_n \to u^\star$,
	\Cref{observation:thmC_pf_sequential_compactness} permits us to assume (again passing to a subsequence if required) that
	\begin{equation*}
		x^n \to x
		\quad \text{and} \quad
		x^{\dag,n} \to x^\dag
		\quad \text{pointwise as $n \to \infty$}
	\end{equation*}
	for some $x,x^\dag \in \mathcal{X}'$.
	We have $X^{nm} \geq X^{nm,\dag}$ for any $n,m \in \N$ by \Cref{lemma:mcs},
	so that letting $m \to \infty$ and $n \to \infty$ yields $X \geq X^\dag$.

	$(x^n,X^n)$ satisfies the Euler equation for $\left(F^0_n,F^1_n,G\right)$ for each $n \in \N$ by \Cref{lemma:approx_G},
	and thus $(x,X)$ satisfies the Euler equation for $\left(F^0,F^1,G\right)$ by \Cref{lemma:approx_F}.
	Hence $( x, X )$ is uniquely optimal for $G$ by \Cref{lemma:euler_optimal}.
	Similarly, $\bigl( x^\dag, X^\dag \bigr)$ is uniquely optimal for $G^\dag$.
\end{proof}

\subsubsection{Proof of \texorpdfstring{\Cref{lemma:mcs}}{Lemma~\ref{lemma:mcs}}}
\label{suppl:mcs:pf_prop_mcs}

The argument is elementary but tedious.
Fix $x,x^\dag \in \mathcal{X}'$ such that $(x,X)$ and $\bigl( x^\dag, X^\dag \bigr)$ satisfy the Euler equation for $G$ and $G^\dag$, respectively; we must show that $X \geq X^\dag$.
Enumerate the (common) support of $G$ and $G^\dag$ as $\{t_k\}_{k=1}^K \subseteq \R_+$, where $K \in \N$ and
\begin{equation*}
	0 < t_1 < \cdots < t_K < \infty .
\end{equation*}
Since $F^0,F^1$ are simple and $x,x^\dag \in \mathcal{X}'$, \Cref{observation:euler_differentiable} in \cref{app:Euler:construction} (\cpageref{observation:euler_differentiable}) implies that
for some $u_1 \geq \cdots \geq u_K$ in $\left[u^\star,u^0\right]$, we have
\begin{equation*}
	x_t =
	\begin{cases}
		u^0 & \text{for a.e. $t \in [0,t_1)$}
		\\
		u_k & \text{for a.e. $t \in [t_k,t_{k+1})$ where $k \in \{1,\dots,K-1\}$}
		\\
		u_K & \text{for a.e. $t \in [t_K,\infty)$,}
	\end{cases}
	\label{eq:simple_optimal_x}
	\tag{S}
\end{equation*}
and that $x^\dag$ satisfies also \eqref{eq:simple_optimal_x} with some $u^\dag_1 \geq \cdots \geq u^\dag_K$ in $\left[u^\star,u^0\right]$.

\begin{namedthm}[Claim.]
	\label{claim:pf_prop_mcs}
	It suffices to show that $X_{t_k} \geq X^\dag_{t_k}$ for every $k \in \{1,\dots,K\}$.
\end{namedthm}

\begin{proof}%
	\renewcommand{\qedsymbol}{$\square$}
	Suppose that $X_{t_k} \geq X^\dag_{t_k}$ for every $k \in \{1,\dots,K\}$,
	and fix an arbitrary $t \in \R_+$; we shall show that $\smash{X_t \geq X^\dag_t}$.
	If $t \geq t_K$, then
	\begin{equation*}
		X_t = X_{t_K} \geq X^\dag_{t_K} = X^\dag_t
	\end{equation*}
	by \eqref{eq:simple_optimal_x}.
	Assume for the remainder that $t < t_K$.

	Suppose first that for some $k \leq K$,
	we have $t \leq t_k$ and $x \geq x^\dag$ a.e. on $(t,t_k)$.
	(This holds if $t \leq t_1$, since $x = u^0 = x^\dag$ a.e. on $[0,t_1)$ by \eqref{eq:simple_optimal_x}.)
	Then
	\begin{align*}
		X_t - X^\dag_t
		= r \int_t^{t_k} e^{-r(s-t)}\left(x_s-x^\dag_s\right) \dd s
		+{}& e^{-r(t_k-t)} \left(X_{t_k}-X^\dag_{t_k}\right)
		\\
		\geq{}& e^{-r(t_k-t)} \left(X_{t_k}-X^\dag_{t_k}\right)
		\geq 0 .
	\end{align*}
	Suppose instead that $t \in (t_k,t_{k+1})$ for some $k < K$ and that $x < x^\dag$ on a non-null subset of $(t_k,t_{k+1})$.
	Then $x < x^\dag$ a.e. on $(t_k,t)$ by \eqref{eq:simple_optimal_x}, so that
	\begin{align*}
		0
		\leq X_{t_k} - X^\dag_{t_k}
		= r \int_{t_k}^t e^{-r(s-t_k)}\left(x_s-x^\dag_s\right) \dd s
		+{}& e^{-r(t-t_k)} \left(X_t-X^\dag_t\right)
		\\
		\leq{}& e^{-r(t-t_k)} \left(X_t-X^\dag_t\right) .
		\qedhere
	\end{align*}
\end{proof}%
\renewcommand{\qedsymbol}{$\blacksquare$}

To show that $X_{t_k} \geq X^\dag_{t_k}$ for every $k \in \{1,\dots,K\}$,
suppose not; we shall derive a contradiction.
Let $k'$ denote the largest $k \in \{1,\dots,K\}$ at which $X_{t_k} < X^\dag_{t_k}$.
We shall prove that for every $k \leq k'$, it holds that
\begin{align}
	X_{t_k}
	&< X^\dag_{t_k}
	\label{eq:x_le_xdag}
	\\
	\text{and} \quad
	\E_{G}\left(F^{1\prime}(X_\tau) \middle|\tau \geq t_k \right)
	&> \E_{G^\dag}\left(F^{1\prime}\left(X^\dag_\tau\right)
	\middle|\tau \geq t_k \right) .
	\label{eq:E_le_Edag}
\end{align}
This suffices because it contradicts the fact that
\begin{multline*}
	\E_{G}\left(F^{1\prime}(X_\tau)
	\middle| \tau \geq t_1 \right)
	= \E_{G}\left(F^{1\prime}(X_\tau) \right)
	\\
	= 0
	= \E_{G^\dag}\left( F^{1\prime}\left(X^\dag_\tau\right) \right)
	= \E_{G^\dag}\left( F^{1\prime}\left(X^\dag_\tau\right)
	\middle| \tau \geq t_1 \right) ,
\end{multline*}
which holds by \Cref{observation:euler_differentiable} in \cref{app:Euler:construction} (\cpageref{observation:euler_differentiable}) since $(x,X)$ and $\bigl( x^\dag, X^\dag \bigr)$ satisfy the Euler equation for $G$ and $G^\dag$.
We proceed by (backward) induction on $k \in \{k',\dots,1\}$.

\emph{Base case: $k=k'$.}
Here \eqref{eq:x_le_xdag} holds by hypothesis, so we need only derive \eqref{eq:E_le_Edag}.
If $k' = K$, then we have by strict concavity of $F^1$ that
\begin{equation*}
	\E_{G}\left(F^{1\prime}(X_\tau)
	\middle|\tau \geq t_k \right)
	= F^{1\prime}\left(X_{t_K}\right)
	> F^{1\prime}\left(X^\dag_{t_K}\right)
	= \E_{G^\dag}\left(F^{1\prime}\left(X^\dag_\tau\right)
	\middle|\tau \geq t_k \right) .
\end{equation*}
Assume for the remainder that $k' < K$.

Since $X_{t_k} < X^\dag_{t_k}$ and $X_{t_{k+1}} \geq X^\dag_{t_{k+1}}$ by definition of $k'$, \eqref{eq:simple_optimal_x} yields
\begin{align*}
	\left( 1 - e^{-r\left(t_{k+1}-t_k\right)} \right) u_k
	&= X_{t_k} - e^{-r\left(t_{k+1}-t_k\right)} X_{t_{k+1}}
	\\
	&< X^\dag_{t_k} - e^{-r\left(t_{k+1}-t_k\right)} X^\dag_{t_{k+1}}
	= \left( 1 - e^{-r\left(t_{k+1}-t_k\right)} \right) u^\dag_k ,
\end{align*}
so that $u_k < u^\dag_k$.
It follows by the strict concavity of $F^0$ that
\begin{equation*}
	\E_{G}\left(F^{1\prime}(X_\tau) \middle|\tau > t_k \right)
	= F^{0\prime}(u_k)
	> F^{0\prime}\left(u^\dag_k\right)
	= \E_{G^\dag}\left(F^{1\prime}\left(X^\dag_\tau\right) \middle|\tau > t_k \right) ,
\end{equation*}
which is to say that \eqref{eq:E_le_Edag} holds at $k+1$.
Thus \eqref{eq:E_le_Edag} holds at $k$:
\begin{align*}
	\E_{G}\left(F^{1\prime}\left(X_\tau\right)\middle|\tau \geq t_k\right)
	&= \PP_{G}\left(\tau = t_k \middle| \tau \geq t_k\right)
	F^{1\prime}\left(X_{t_k}\right)
	\\
	&\quad + \PP_{G}\left(\tau > t_k \middle| \tau \geq t_k\right)
	\E_{G}\left(F^{1\prime}\left(X_\tau\right)\middle|\tau > t_k\right)
	\\
	&\geq \PP_{G^\dag}\left(\tau = t_k \middle| \tau \geq t_k\right)
	F^{1\prime}\left(X_{t_k}\right)
	\\
	&\quad +\PP_{G^\dag}\left(\tau > t_k \middle| \tau \geq t_k\right)
	\E_{G}\left(F^{1\prime}\left(X_\tau\right)
	\middle|\tau > t_k\right)
	\\
	&> \PP_{G^\dag}\left(\tau = t_k \middle| \tau \geq t_k\right)
	F^{1\prime}\left(X^\dag_{t_k}\right)
	\\
	&\quad + \PP_{G^\dag}\left(\tau > t_k \middle| \tau \geq t_k\right)
	\E_{G^\dag}\left(F^{1\prime}\left(X^\dag_\tau\right)
	\middle|\tau > t_k\right)
	\\
	&= \E_{G^\dag}\left(F^{1\prime}\left(X^\dag_\tau\right)
	\middle|\tau \geq t_k\right) ,
\end{align*}
where the weak inequality holds since
$G|_{\tau \geq t_k}$ MLR-dominates $G^\dag|_{\tau \geq t_k}$ and
\begin{equation*}
	F^{1\prime}\left(X_{t_k}\right)
	\leq \E_{G}\left( F^{1\prime}\left(X_\tau\right)
	\middle| \tau > t_k \right)
	\quad \text{since $X$ and $F^{1\prime}$ are decreasing,}
\end{equation*}
and the strict inequality holds by \eqref{eq:x_le_xdag} and strict concavity of $F^1$ (first term)
and the fact that \eqref{eq:E_le_Edag} holds at $k+1$ (second term).

\emph{Induction step:} Assume that \eqref{eq:x_le_xdag} and \eqref{eq:E_le_Edag} hold at $k+1 \leq K$; we must show that they hold at $k$.
Since \eqref{eq:E_le_Edag} holds at $k+1$, we have
\begin{equation*}
	F^{0\prime}\left(u_k\right)
	= \E_{G}\left(F^{1\prime}(X_\tau)
	\middle|\tau \geq t_{k+1} \right)
	> \E_{G^\dag}\left(F^{1\prime}\left(X^\dag_\tau\right)
	\middle|\tau \geq t_{k+1} \right)
	= F^{0\prime}\left(u^\dag_k\right) ,
\end{equation*}
so that $u_k < u^\dag_k$ by strict concavity of $F^0$.
Using \eqref{eq:simple_optimal_x} and the fact that \eqref{eq:x_le_xdag} holds at $k+1$ yields
\begin{align*}
	X_{t_k}
	&= \left( 1 - e^{-r\left(t_{k+1}-t_k\right)} \right) u_k
	+ e^{-r\left(t_{k+1}-t_k\right)} X_{t_{k+1}}
	\\
	&< \left( 1 - e^{-r\left(t_{k+1}-t_k\right)} \right) u^\dag_k
	+ e^{-r\left(t_{k+1}-t_k\right)} X^\dag_{t_{k+1}}
	= X^\dag_{t_k} ,
\end{align*}
showing that \eqref{eq:x_le_xdag} holds at $k$.
Since \eqref{eq:x_le_xdag} holds at $k$
and \eqref{eq:E_le_Edag} holds at $k+1$,
the (exact) same argument as in the base case yields that \eqref{eq:E_le_Edag} holds at $k$.
\qed

\subsubsection{Proof of \texorpdfstring{\Cref{lemma:euler_optimal}}{Lemma~\ref{lemma:euler_optimal}}}
\label{suppl:mcs:pf_lemma_euler_optimal}

Recall the definitions of $\mathcal{X}$ and $\pi_G$ from \cref{app:Euler}.
Note that $\mathcal{X}$ is convex.

\begin{observation}
	\label{observation:piG_str_conc}
	If $F^0$ is strictly concave
	and $G$ has unbounded support,
	then $\argmax_{\mathcal{X}} \pi_G$ has at most one element.
\end{observation}

We omit the easy proof; see \textcite{omit}.

\begin{proof}[Proof of \Cref{lemma:euler_optimal}]
	If $(x,X)$ is an optimal mechanism,
	then we must have $x \in \mathcal{X}$ by \Cref{lemma:lequ0} (\cpageref{lemma:lequ0}),
	and thus $x$ must belong to $\argmax_{\mathcal{X}} \pi_G$.

	By \Cref{corollary:optimal_existence} in \cref{suppl:undom_opt_properties} (\cpageref{corollary:optimal_existence}),
	there is a mechanism $\bigl(x^\dag,X^\dag\bigr)$ that is optimal for $G$.
	Thus $\argmax_{\mathcal{X}} \pi_G = \bigl\{ x^\dag \bigr\}$ by \Cref{observation:piG_str_conc}.

	Now, if a mechanism $(x,X)$ satisfies $x \in \mathcal{X}$ and the Euler equation,
	then $x$ belongs to $\argmax_{\mathcal{X}} \pi_G$ by the \hyperref[lemma:euler]{Euler lemma} in \cref{app:Euler} (\cpageref{lemma:euler}),
	so $(x,X)$ must be the uniquely optimal mechanism $\bigl(x^\dag,X^\dag\bigr)$.
\end{proof}

\end{appendices}



\printbibliography[heading=bibintoc]


\end{document}